\documentclass[a4paper, 11pt]{article}

\usepackage{amsmath}
\usepackage{amssymb}
\usepackage{amsfonts}
\usepackage{amsthm}
\usepackage{bm}
\usepackage{caption}
\usepackage{subcaption}
\usepackage{dsfont}
\usepackage{setspace}
\usepackage{geometry}
\usepackage[hidelinks]{hyperref}
\usepackage[inline]{enumitem}
\usepackage{rotating}
\usepackage{graphicx}
\usepackage{float}

\usepackage{accents}

\usepackage[round, sort, authoryear]{natbib}
\usepackage{multirow}

\usepackage{upgreek}
\usepackage{prodint}
\usepackage{makecell}

\newcommand{\real}[1][]{\ensuremath{\mathbb{R}^{#1}}}
\newcommand{\naturalnumber}{\ensuremath{\mathbb{N}}}
\newcommand{\family}[1]{\ensuremath{\mathcal{#1}}}

\newcommand{\rbracket}[1]{\ensuremath{\left( #1\right)}}
\newcommand{\sbracket}[1]{\ensuremath{\left\lbrack #1\right\rbrack}}
\newcommand{\cbracket}[1]{\ensuremath{\left\lbrace #1\right\rbrace}}
\newcommand{\set}[2]{\ensuremath{\cbracket{#1:#2}}}
\newcommand{\indicator}[1]{\ensuremath{\mathds{1}{\lbrace{#1}\rbrace}}}
\newcommand{\mat}[1]{\ensuremath{\bm{\mathrm{#1}}}}

\DeclareMathOperator*{\argmax}{arg\,max}

\newcommand{\converge}[1][]{\ensuremath{\xrightarrow{#1}}}

\newcommand{\prob}[2][\rbracket]{\ensuremath{\mathbb{P}#1{#2}}}
\newcommand{\expect}[2][\rbracket]{\ensuremath{\mathbb{E}#1{#2}}}

\newcommand{\dd}[1]{\ensuremath{\mathrm{d}{#1}}}

\newtheorem{prop}{Proposition}
\newtheorem{lemma}{Lemma}
\newtheorem{theorem}{Theorem}
\newtheorem{remark}{Remark}

\usepackage{authblk}

\title{Testing for sufficient follow-up\\ in  cure models with categorical covariates}
\author[a]{Tsz Pang Yuen\thanks{t.p.yuen@uva.nl}}
\author[a]{Eni Musta\thanks{e.musta@uva.nl}}
\author[b]{Ingrid Van Keilegom\thanks{ingrid.vankeilegom@kuleuven.be}\thanks{I. Van Keilegom acknowledges funding from the FWO and F.R.S. - FNRS (Excellence of Science programme, project ASTeRISK, grant no. G0I3422N), and from the FWO (senior research projects fundamental research, grant no. G047524N).}}
\affil[a]{Korteweg-de Vries Institute for Mathematics, University of Amsterdam, Netherlands}
\affil[b]{ORSTAT, KU Leuven, Belgium}
\date{}

\begin{document}
\maketitle
\begin{abstract}
In survival analysis, estimating the fraction of `immune' or `cured' subjects who will never experience the event of interest, requires a sufficiently long follow-up period. A few statistical tests have been proposed to test the assumption of sufficient follow-up, i.e. whether the right extreme of the censoring distribution exceeds that of the survival time of the uncured subjects. However, in practice the problem remains challenging. To address this, a relaxed notion of `practically' sufficient follow-up has been introduced recently, suggesting that the follow-up would be considered sufficiently long if the probability for the event occurring after the end of the study is very small. All these existing tests do not incorporate covariate information, which might affect the cure rate and the survival times. 
We extend the test for `practically' sufficient follow-up to  settings with 
categorical covariates. While a straightforward intersection-union type test could reject the null hypothesis of insufficient follow-up only if such hypothesis is rejected for all covariate values, in practice this approach is overly conservative and lacks power.
To improve upon this, we propose a novel test procedure that relies on the test decision for one properly chosen covariate value. 
Our approach relies on the assumption that 
the conditional density of the uncured survival time is a non-increasing function of time in the tail region. 
We show that both methods yield tests of asymptotically level $\alpha$ and investigate their finite sample performance through simulations. The practical application of the methods is illustrated using a skin melanoma dataset.
\end{abstract}

\section{Introduction}\label{sec:intro}
The analysis of censored survival data containing subjects who will never experience the event of interest has become popular from both methodological and application point of view. The subjects who are immune to the event of interest are referred to as `cured', while the others, who will eventually experience the event of interest, are referred to as `susceptible'.
Cure models have been developed to analyze such data and utilized to model the effect of covariates on the cure probabilities
and the conditional survival function of the susceptible subjects.
For instance, 
advancements in cancer treatments have led to a growing proportion of patients being cured, making cure models particularly relevant in this context \citep{LB2019}. Other applications of cure models can be found in fertility studies and economics. 

Two major families of cure models have been proposed in the literature, namely, mixture cure models and promotion time models. The mixture cure models were first introduced by \cite{boag_maximum_1949,berkson_survival_1952}.
To incorporate covariates into the model, fully parametric models assuming a logistic model for the cure probability and various parametric assumptions for the conditional survival function of the susceptible subjects were proposed \citep{farewell_model_1977,farewell_use_1982}.
Extensions to semi-parametric models for the conditional survival function, using Cox proportional hazards models \citep{peng_nonparametric_2000,sy_estimation_2000} and accelerated failure time models \citep{li_semi-parametric_2002,zhang_new_2007,vKP2024}, were later introduced.
\cite{PvK2020,burke_likelihood-based_2021} studied estimation approaches for the cure probability under a parametric structure, without imposing assumptions on the conditional survival function of the susceptible subjects.
The use of semi-parametric single-index models for the cure probability, while retaining the Cox proportional hazards assumption for the susceptible subjects, were also proposed \citep{amico_singleindexcox_2019,eni_musta_single-index_2024}.
Nonparametric models for the cure probability have also been introduced \citep{XP2014,LCJK2017}.
Mixture cure models that account for interval censoring or the presence of mismeasured covariates have also been studied \citep{kim_cure_2008, ma_mixed_2010, pal_new_2024, musta_eni_simulation-extrapolation_2021}.
The second class of cure models, the promotion time model, was originally proposed by \cite{yakovlev_stochastic_1996}. This model adopts the Cox proportional hazards framework to allow for a cure fraction.
Since then, various extensions of this model have been studied \citep{zeng_semiparametric_2006, liu_semiparametric_2009, beyhum_jad_extension_2022, lou_semiparametric_2025}.
For a comprehensive review of cure models, we refer to \cite{MZ1996}, \cite{AK2018}, \cite{PY2022} and \cite{MRSZ2024}.
Among the two families of cure models, the present study focuses on the mixture cure model.

Estimation of cure probabilities is difficult because  the cured subjects cannot be distinguished from the censored uncured ones. This requires careful statistical analysis in order to avoid overestimation of the cure fraction. 
One common assumption is that the 
study has a sufficiently long follow-up, in the sense that the support of the uncured survival time is included in the support of the censoring time \citep{MZ1992,MZ1994}. The importance of such assumption underlines the need for a reliable statistical test. A procedure for testing the hypothesis of insufficient follow-up in the absence of covariate information was proposed by \cite{MZ1996} and \cite{MRS2023}. Other methods have also been developed, 
see for example \cite{S2000}, \cite{SO2023} and \cite{XEK2024}.
More recently, a relaxed notion of sufficient follow-up was introduced by \cite{YM2024}, which is characterized by the quantile of the uncured survival time distribution instead of its right end point. Under the assumption that the uncured survival time distribution has a non-increasing density in the tail region, a test procedure based on the smooth Grenander estimator was proposed.

While the existing procedures 
test the hypothesis of insufficient follow-up without accounting for covariate information, 
there is a lack of tests 
that address this issue in presence of covariates. A common assumption of sufficient follow-up over all covariate values is made when we apply the semiparametric mixture cure models \citep{PvK2020,vKP2024,MPvK2022,MPvK2024}
or the nonparametric estimators of \cite{XP2014} and \cite{LCJK2017} to study the effect of covariates on the cure fraction. Accounting for covariates in testing sufficient follow-up is crucial because sufficient follow-up
in the absence of covariates (i.e. under the unconditional setting) does not imply sufficient
follow-up
over all covariate values (i.e. under the conditional setting) as illustrated in Remark~\ref{remark1} below. Hence, relying on just the unconditional test might lead to wrong conditional estimates of cure probabilities. This underscores the necessity for a procedure for testing insufficient follow-up over all covariate values.

Based on the relaxed notion of sufficient follow-up introduced by \cite{YM2024}, we first formulate the extended hypothesis of sufficient follow-up for all covariate values.
Focusing on categorical covariates, we extend the test statistic from the unconditional to the conditional setting and introduce two testing procedures. 
The first method is related to
the intersection-union tests \citep{B1982,BH1996} and rejects the null hypothesis of overall insufficient follow-up only if such hypothesis is rejected by all the individual tests for each covariate value. This method is straightforward to implement but in practice leads to a very conservative procedure and lacks power to detect sufficient follow-up, particularly in settings with many possible covariate values. Instead, the second method relies on choosing appropriately one covariate value, the one for which sufficient follow-up is less likely, and deciding just based on the conditional test decision for that covariate value.
The asymptotic properties of the tests are investigated and their finite sample performance is studied through simulation studies.
We focus solely on accounting for categorical covariates in this paper since continuous covariates need to be handled differently. Moreover, this is typically  the most relevant case in practical applications since even continuous covariates are often discretized.

At first glance, the extension from the unconditional to the conditional setting may appear straightforward, especially for the first method based on the intersection-union type test. However, because of lack of power, the first method is not the focus of this paper. Instead, the novelty consists in introducing the second method, which
is inspired by a max-type test that considers the maximum of individual test statistics over all covariate values. Although the asymptotic properties of the max-type test can be established using the Delta method, its practical performance was very poor.
Motivated by
the idea of the Delta method for the maximum function, the second method first selects an appropriate covariate value at which the follow-up is less likely to be sufficient. This selection is a novel aspect of the present work and will be discussed further later.
Accurate estimation of this covariate value is the main challenge of the method and using bootstrap samples for this estimation is not straightforward. Even if the individual test procedure for fixed covariate value performs well, choosing an inappropriate covariate value can lead to an incorrect test decision.%

The paper is organized as follows. Section~\ref{sec:model_desc} contains the model description. In Section~\ref{sec:proc}, we formulate the problem of testing insufficient follow-up over all covariate values and introduce two test procedures.
The asymptotic properties of the test statistics are presented in Section~\ref{sec:asym_prop} and the finite sample performance of the two methods is investigated
through simulation studies in Section~\ref{sec:simulations}. We analyze a skin melanoma dataset in Section~\ref{sec:app} to illustrate the methods in practice.
The conclusion and discussion are presented in Section~\ref{sec:discuss}.
The proofs are relegated to Appendix~\ref{sec:appendix}.
Additional simulation results are collected in the Supplementary Material.
Software in the form of R code is available on the GitHub repository \url{https://github.com/tp-yuen/cureSFUTest}.

\section{Model description}
\label{sec:model_desc}
Let $T$ be a non-negative random variable denoting the time to the event of interest, which equals infinity if the subject is cured. Suppose $B=\indicator{T<\infty}$ is the uncure status, where $B=1$ if the subject is susceptible to the event of interest (uncured) and $B=0$ if the subject is non-susceptible to the event (cured). The fraction of cured subjects and the survival function of the uncured subjects can depend on various features of the subjects, denoted by a set of covariates ${X}$. We note that the cure probability and the survival time of the uncured subjects do not necessarily depend on the same set of covariates; however, we assume both depend on the same set of covariates for simplicity. Let
\begin{equation}
	\label{eqn:def_p}
	p({x})=\prob{B=1|{X}={x}}
\end{equation} be the conditional probability of being uncured and $F(t|{x})=\prob{T \leq t|{X}={x}}$ be the conditional distribution of $T$.
We denote the survival function of $F(t|{x})$ by $S(t|{x})$.
For the covariate values under which cure is attainable, the survival function {$S(t|{x})$} is improper since $\lim_{t\converge\infty}S(t|{x})=1-p({x})$ is strictly positive. Suppose $F_{u}(t|{x})=\prob{T\leq t|{X}={x}, B=1}$ is the conditional distribution function of $T$ given that the subject is uncured and $S_{u}(t|{x})$ is the corresponding survival function. The survival function under the mixture cure model is given by
\begin{equation}
	\label{eq:mcm}
	S(t|{x}) = 1 - p({x}) + p({x})S_{u}(t|{x}),
\end{equation}
where $1 - p({x})$ and $S_{u}(t|{x})$ are called the incidence and the latency, respectively.

Let $C$ be the non-negative censoring time with conditional distribution function $G(\cdot|x)$. Under the random right censoring assumption, we observe the follow-up time $Y=\min(T,C)$ with distribution function $H$ and the censoring indicator $\Delta=\indicator{T\leq C}$. In the presence of right censoring, the uncure status $B$ is a latent variable and the cured individuals cannot be distinguished from the censored but uncured subjects. We further assume that $C$ and $T$ are conditionally independent given the covariates ${X}$.
The observations are independent and identically distributed (i.i.d.) realizations of $(X, Y, \Delta)$, denoted by $\cbracket{(X_i, Y_i, \Delta_i), i=1,\dots,n}$.

In this paper we consider the case in which the covariates $X$ are categorical. Without loss of generality, we assume that $X$ is univariate with finite support $\family{X}$ of cardinality $|\family{X}|=d$.
The case with multiple categorical covariates can be reduced to this setting;
see, for example, the simulation Setting~\ref*{enum:sim_exp_unif_2cov} in Section~\ref{sec:sim_settings}. To estimate the incidence $1-p(x)$, we can first split the sample into $d$ subsamples corresponding to different values 
of $X$. We denote each subsample by $\cbracket{(Y_{xi}, \Delta_{xi}), i=1,\dots,n_x}$, for $x\in\family{X}$, where $n_x=\sum_{i=1}^{n}\indicator{X_i=x}$. 
Let 
\[
\hat{H}_{k,n}(t|x)=\sum_{i=1}^{n_x}\frac{\indicator{Y_{xi}\leq t,\Delta_{xi}=k}}{n_x}
\]
be an empirical estimator of
$H_k(t|x)=\prob{Y_{x1}\leq t, \Delta_{x1}=k}$, for $k=0,1$. Also let $\hat{H}_n(t|x)=\sum_{i=1}^{n_x}\indicator{Y_{xi}\leq t}/n_x$ be the empirical estimator of $H(t|x)=\prob{Y_{x1}\leq t}=H_0(t|x)+H_1(t|x)$.
The Kaplan--Meier estimator (KME) \citep{KM1958} can be used to estimate the conditional survival function for each subsample:
\begin{equation}
	\label{eq:kme}
	{
		\hat{S}_n(t|x)=\Prodi_{s\leq t}\left(
		1 - \frac{\hat{H}_{1,n}(ds|x)}{1-\hat{H}_n(s{-}|x)}
		\right),
	}
\end{equation}
where $\prodi_{s\in A}$ denotes the product-integral over the set $A$ and $\hat{H}_{1,n}(ds|x)$ is the measure associated with $\hat{H}_{1,n}(\cdot|x)$.
We denote $\hat{F}_{n}(t|x)=1-\hat{S}_{n}(t|x)$ for the KME of $F(t|x)$.
Then, the incidence can be estimated by $1-\hat{p}_n(x)=\hat{S}_n(Y_{x(n_x)}|x)$, as proposed by \cite{MZ1992},
where $Y_{x(n_x)}=\max_iY_{xi}$ is the maximum survival time for the subsample corresponding to $X=x$. 

\section{Testing procedures}\label{sec:proc}
Let $\uptau_{G}(x)=\sup\set{t}{G(t|x)<1}$ and $\uptau_{F_{u}}(x)=\sup\set{t}{F_{u}(t|x)<1}$, where $F_u(t|x)=1-S_u(t|x)$.
For each $x\in\family{X}$, the assumption of sufficient follow-up, $\uptau_{F_{u}}(x) \le \uptau_{G}(x)$, is crucial to guarantee that the incidence estimator does not overestimate the cure rate \citep{MZ1992,MZ1994}. 
In the unconditional setting, \cite{MZ1992,MZ1994} considered testing $H_0:\uptau_{F_{u}} \geq \uptau_{G}$ versus $H_1: \uptau_{F_{u}} < \uptau_{G}$. More recently, \cite{YM2024} proposed a more relaxed formulation of sufficient follow-up, which was characterized by the quantile of the uncured survival time distribution. Specifically, the authors considered testing
\begin{equation}
	\label{eq:suff_follow_up_hyp_uncond}
	H_0: q_{1-\epsilon} \geq \uptau_{G}
	\quad\text{versus}\quad
	H_1: q_{1-\epsilon} < \uptau_{G},
\end{equation}
where $q_{1-\epsilon}$ is the $(1-\epsilon)$-quantile of $F_u$ in the unconditional setting and $\epsilon$ is small (for example 1\%).
Such hypotheses are practically motivated by the idea that the follow-up is considered sufficient if the probability of the event occurring after the end of the study is negligible. In a similar context, \cite{SO2023} used simulations to study the performance of estimation and inference for mixture cure model parameters as follow-up time increased. Their findings suggested that the follow-up can be considered sufficient when the probability of the event happening after the study ends is less than 1\%. Specifically, the results showed that improvements in mean squared error and the coverage of the confidence interval for the cure fraction estimate diminish as the follow-up time exceeds the 99\%-quantile of $F_u$.
It is natural to extend the hypotheses in \eqref{eq:suff_follow_up_hyp_uncond} to the conditional setting for testing the sufficient follow-up assumption for fixed $x\in\family{X}$:
\begin{equation}
	\label{eq:suff_follow_up_hyp_fix_cov}
	H_{0x}: q_{1-\epsilon}(x) \geq \uptau_{G}(x)
	\quad\text{versus}\quad 
	H_{1x}: q_{1-\epsilon}(x) < \uptau_{G}(x),
\end{equation}
where $q_{1-\epsilon}(x)=\inf\set{t}{F_{u}(t|x)\geq 1-\epsilon}$.
The test procedure proposed by \cite{YM2024} can then be applied to the corresponding subsample with covariate value $x$ for testing \eqref{eq:suff_follow_up_hyp_fix_cov}.

\begin{remark}
	\label{remark1} Note that it is possible for the follow-up to be sufficient in the unconditional setting but not in the conditional setting. For example, consider a Bernoulli covariate $X\sim B(1,0.5)$ and  $\family{X}=\lbrace0,1\rbrace$ with $F_u(t|0)=1-e^{-t}$ ($q_{0.99}(0)\approx4.6052$), $F_u(t|1)=1-e^{-2t}$ ($q_{0.99}(1)\approx2.3059$), $\uptau_{G}=\uptau_{G}(k)=3$, $k=0,1$. This gives $F_u(t)=0.5(F_u(t|0)+F_u(t|1))$ ($q_{0.99}\approx2.7721$). Then, with $\epsilon=0.01$, we have $q_{1-\epsilon}(1)<\uptau_{G}(1)=\uptau_{G}(0)<q_{1-\epsilon}(0)$, while $q_{1-\epsilon}<\uptau_{G}$. Thus, the follow-up is sufficient for $x=1$ but insufficient for $x=0$, although it remains sufficient in the unconditional setting. 
\end{remark}

Therefore, when estimating the cure rate that depends on covariates, it is crucial to test the sufficient follow-up assumption that holds for all $x\in\family{X}$, i.e., $q_{1-\epsilon}(x) < \uptau_{G}(x)$ for all $x\in\family{X}$. This is tantamount to considering the follow-up as sufficient over all covariate values if, for all $x\in\family{X}$, the conditional probability for the event to happen after the end of the study $\uptau_{G}(x)$ is negligible, 
i.e. $S_u(\uptau_{G}(x)|x)<\epsilon$. Therefore we consider the following hypotheses for testing the overall sufficient follow-up assumption:
\begin{equation}
	\label{eq:suff_follow_up_hyp_all_cov}
	H_0: q_{1-\epsilon}(x) \geq \uptau_{G}(x)\quad\text{for some}\ x\in\family{X}
	\quad\text{versus}\quad 
	H_1: q_{1-\epsilon}(x) < \uptau_{G}(x)\quad\text{for all}\ x\in\family{X}.
\end{equation}

Note that the condition of sufficient follow-up for a fixed $x\in\family{X}$ in the conditional setting guarantees that the cure rate estimator determined from the KME (using the subsample corresponding to the value of $x$) does not overestimate the true cure rate for that particular covariate value $x$, so testing $H_{0x}$ in \eqref{eq:suff_follow_up_hyp_fix_cov} is sufficient for this purpose.
We, however, emphasize that the follow-up is required to be sufficient over all covariate values when we are studying the effect of covariate on cure rate using mixture cure models \citep{PK2023,MPvK2022}. In such cases, one would need to
consider testing $H_0$ in \eqref{eq:suff_follow_up_hyp_all_cov}.

\subsection*{Main idea of the test}
Our test relies on the assumption that, for any $x\in\family{X}$, $F_u(t|x)$ has a density function $f_u(t|x)$, which is non-increasing with respect to $t$ in the tail region $[a_x,\tau_{F_u}(x))$ for some $0\leq a_x<\tau_G(x)$.
As in \cite{YM2024}, let $0<\eta<\epsilon$ be a very small number compared to $\epsilon$
and $\tau>q_{1-\epsilon}(x)$ for all $x\in\family{X}$ be such that $F_u(\tau|x)\geq 1-\eta$  for all $x\in\family{X}$. This essentially means that the probability of an event happening after $\tau$ is negligible.
By the monotonicity of $f_u(\cdot|x)$, we have
	\[
	\epsilon-\eta=\int_{q_{1-\epsilon}(x)}^{\tau}f_u(t|x)\dd{t}
	\le
	f_u(q_{1-\epsilon}(x)|x)(\tau-q_{1-\epsilon}(x)).
	\]
	Furthermore, since $f(t|x)=p(x)f_u(t|x)$, where $p(x)$ is the conditional uncure probability defined in~\eqref{eqn:def_p} and $p(x)\ge F(\uptau_G(x)|x)$,
we have a lower bound on $f(q_{1-\epsilon}(x)|x)$
\[
f(q_{1-\epsilon}(x)|x)
\geq\frac{(\epsilon-\eta)p(x)}{\tau-q_{1-\epsilon}(x)}
\geq\frac{(\epsilon-\eta)F(\uptau_{G}(x)|x)}{\tau-q_{1-\epsilon}(x)}.
\]
This means that under $H_0: q_{1-\epsilon}(x) \geq \uptau_{G}(x)\ \text{for some}\ x\in\family{X}$, we have
\begin{equation}
	\label{eq:insuff_follow_up_conseq}
	f(\uptau_{G}(x)|x)\geq f(q_{1-\epsilon}(x)|x)
	\geq\frac{(\epsilon-\eta)F(\uptau_{G}(x)|x)}{\tau-\uptau_{G}(x)}
	\quad\text{for some}\ x\in\family{X}.
\end{equation}

To test the hypotheses in \eqref{eq:suff_follow_up_hyp_all_cov}, we split the sample into subsamples corresponding to different values
of $X$ and adopt the procedure proposed by \cite{YM2024} for each subsample. 
The smoothed Grenander estimator with boundary correction of $f(t|x)$ is used; see, for example, \cite{LM2017,GJ2014} for references. The estimator is given by:
\begin{equation}
	\label{eq:sg}
	\hat{f}_{nb_x}(t|x) = 
	\int_{0\vee(t-b_x)}^{(t+b_x)\wedge Y_{x(n_x)}}
	\frac{1}{b_x}k^B_{b_x,t}\rbracket{
		\frac{t-u}{b_x}
	}\hat{f}_{n}^{G}(u|x)
	\dd{u},
	\quad t\in[a_x,Y_{x(n_x)}],
\end{equation}
where $\hat{f}_n^G(\cdot|x)$ is the left derivative of the least concave majorant on $[a_x,Y_{x(n_x)}]$ of the KME $\hat{F}_n(\cdot|x)$, $k^B_{b_x,t}$ is the boundary kernel defined below and $b_x$ is the bandwidth parameter which can be different for each $x\in\family{X}$ and/or can depend on the subsample size $n_x$ instead of $n$. For $v\in[-1,1]$, the boundary kernel $k^B_{b_x,t}(v)$ (see \cite{DGL2013}) is defined by 
\[
k^B_{b_x,t}(v)=
\begin{cases}
	\phi\rbracket{\frac{t-a_x}{b_x}}k(v) + \psi\rbracket{\frac{t-a_x}{b_x}}vk(v)&\quad t\in[a_x,b_x],\\
	k(v)&\quad t\in(b_x,Y_{x(n_x)}-b_x),\\
	\phi\rbracket{\frac{Y_{x(n_x)} - t}{b_x}}k(v) - \psi\rbracket{\frac{Y_{x(n_x)}-t}{b_x}}vk(v)&\quad t\in[Y_{x(n_x)}-b_x,Y_{x(n_x)}],
\end{cases}
\]
where $k$ is a symmetric, twice continuously differentiable kernel with support $[-1,1]$ such that $\int k(u)\dd{u}=1$ and its first derivative is bounded uniformly. 
Here, for $s\in[-1,1]$, the coefficients $\phi(s)$ and $\psi(s)$ (see \cite{ZK1998}) are determined by
\begin{equation}
	\label{eqn:phi,psi}
	\begin{split}
		\phi(s)\int_{-1}^{s}k(v)\dd{v} &+ \psi(s)\int_{-1}^{s}vk(v)\dd{v}=1,\\
		\phi(s)\int_{-1}^{s}vk(v)\dd{v} &+ \psi(s)\int_{-1}^{s}v^{2}k(v)\dd{v}=0.
	\end{split}
\end{equation}
Below we introduce two methods for testing $H_0$ in \eqref{eq:suff_follow_up_hyp_all_cov}.
\paragraph{First method.}
From \eqref{eq:insuff_follow_up_conseq}, $H_0$ entails the following:
\[
f(\uptau_{G}(x)|x)-\frac{(\epsilon-\eta)F(\uptau_{G}(x)|x)}{\tau-\uptau_{G}(x)}\geq0,\quad\text{for some }x\in\family{X}.
\] 
Since $0<\eta<\epsilon$ is negligible, $H_0$ is rejected if
\begin{equation}
	\label{eq:test_discrete}
	\hat{f}_{nb_x}(Y_{x(n_x)}|x)-\frac{\epsilon \hat{F}_{n}(Y_{x(n_x)}|x)}{\tau-Y_{x(n_x)}} 
	< {r_n^{-1}}{\xi_{x,\alpha}}\quad\text{for all}\ x\in\family{X},
\end{equation}
where $\xi_{x,\alpha}$ is the $\alpha$-quantile of the limit distribution of $r_n\lbrace\hat{f}_{nb_x}(Y_{x(n_x)}|x) - {f}(\uptau_G(x)|x)\rbrace$ for some sequence $r_n\to\infty$. This means that we test the sufficient follow-up assumption in each subsample determined by $x$ separately (using different critical values) and at the end reject the combined null hypothesis $H_0$ if $H_{0x}$ was rejected for all values of $x$. Since the individual tests, for each $x\in\family{X}$, have an asymptotic level of $\alpha$, the test with rejection region given in \eqref{eq:test_discrete} also has an asymptotic level of $\alpha$ (see Theorem~\ref{theorem:test_discrete_level}\ref{enum:test_M1_level} in Section~\ref{sec:asym_prop}).
Indeed, if under $H_0$,
$q_{1-\epsilon}(z) \geq \uptau_{G}(z)$ for some $z\in\family{X}$, we have
\begin{align*}
	\begin{split}
		&\prob{\bigcap_{x\in\family{X}}\cbracket{
				\hat{f}_{nb_x}(Y_{x(n_x)}|x)-\frac{\epsilon \hat{F}_{n}(Y_{x(n_x)}|x)}{\tau-Y_{x(n_x)}} 
				< {{r_n^{-1}}\xi_{x,\alpha}}
			}\bigg| H_0}\\
		&\leq\prob{
			\hat{f}_{nb_{z}}(Y_{z(n_z)}|z)-\frac{\epsilon \hat{F}_{n}(Y_{z(n_z)}|z)}{\tau-Y_{z(n_z)}} 
			< {{r_n^{-1}}\xi_{z,\alpha}\bigg| H_0}},
	\end{split}
\end{align*}
where the limit of the probability on the right-hand side of the above inequality is bounded above by $\alpha$.
This test is conservative at first sight because, after conditioning on the covariates, 
the subsamples are independent,
which gives
\begin{equation}
	\label{eq:level_product}
	\begin{aligned}
		&\prob{\bigcap_{x\in\family{X}}\cbracket{
				\hat{f}_{nb_x}(Y_{x(n_x)}|x)-\frac{\epsilon \hat{F}_{n}(Y_{x(n_x)}|x)}{\tau-Y_{x(n_x)}} 
				< {{r_n^{-1}}\xi_{x,\alpha}}
			}\bigg| H_0, X_1,\dots,X_n}\\
		&=\prod_{x\in\family{X}}\prob{
			\hat{f}_{nb_x}(Y_{x(n_x)}|x)-\frac{\epsilon \hat{F}_{n}(Y_{x(n_x)}|x)}{\tau-Y_{x(n_x)}} 
			< {{r_n^{-1}}\xi_{x,\alpha}}
			\bigg| H_0, X_1,\dots,X_n}.
	\end{aligned}
\end{equation}
We note that the worst-case scenario for the level occurs when follow-up is insufficient for only one of the values of $x\in\mathcal{X}$. Then asymptotically, we expect that only one of the probabilities on the right-hand side of \eqref{eq:level_product} is bounded by $\alpha$ and the others converge to one. This also illustrates that, under $H_1$, the power of the test is expected to be the product of the powers for the individual tests for each value of $x\in\family{X}$. In a scenario where the powers of the individual tests for each $x\in\family{X}$ are not close to one, we expect the test procedure to have low power and would further decrease as the number of covariates $|\family{X}|$ increases.
This method relates to the intersection-union tests \citep{B1982,BH1996}, which can be used to test the null hypothesis that an unknown parameter belongs to the union of certain subsets of the parameter space, against the alternative that the parameter belongs to the intersection of the complements of those subsets.

To approximate the critical values, we apply a smoothed bootstrap procedure for each subsample. Let $\hat{f}_{nb_0^x}(\cdot|x)$ be the smoothed Grenander estimator for $f(\cdot|x)$ defined in \eqref{eq:sg} with bandwidth $b_0^x$. Let $\hat{F}_{n}(\cdot|x)$ and $\hat{G}_{n}(\cdot|x)$ be the KMEs for $F(\cdot|x)$ and $G(\cdot|x)$, respectively, based on the subsample corresponding to $X=x$.
Since $\hat{f}_{nb_0^x}(\cdot|x)$ is defined on $[a_x,Y_{x(n_x)}]$, we generate $T_i^\star$ as follows: draw a standard uniform random variable $U_i$; if $U_i\le\hat{F}_n(a_x|X_i)$, set $T_i^\star=\hat{F}_n^{-1}(U_i|X_i)$ with $\hat{F}_n^{-1}(p|X_i)=\inf\lbrace{t\in[0,a_x]:\hat{F}_n(t|X_i)\ge p}\rbrace$; otherwise, generate $T_i^\star$ from $\hat{f}_{nb_0^x}(\cdot|X_i)$ on $[a_x,Y_{x(n_x)}]$.
The censoring time $C_i^\star$ is generated from $\hat{G}_n(\cdot|X_i)$. 
The bootstrap sample is $\lbrace(X_i, Y_i^\star, \Delta_i^\star), i=1,\dots,n\rbrace$, where $Y_i^\star=\min(T_i^\star,C_i^\star)$ and $\Delta_i^\star=\indicator{T_i^\star\leq C_i^\star}$.
This bootstrapping method is suitable for estimating the conditional variance of $\hat{f}_{nb_x}(t|x)$ given $X_1,\dots,X_n$, as discussed in \cite{B1994} based on the principle of ancillarity, which is essential for approximating the critical values for the test.

\paragraph{Second method.}
For simplifying the notation, let
\[
T(x)={f}(\uptau_{G}(x)|x)-\frac{\epsilon {F}(\uptau_{G}(x)|x)}{\tau-\uptau_{G}(x)}
\quad\text{and}\quad
\hat{T}_n(x)=\hat{f}_{nb_x}(Y_{x(n_x)}|x)-\frac{\epsilon\hat{F}_{n}(Y_{x(n_x)}|x)}{\tau-Y_{x(n_x)}}.
\]
From \eqref{eq:insuff_follow_up_conseq}, $H_0$ implies that
\[\max_{x\in\family{X}}\left\lbrace
f(\uptau_{G}(x)|x)-\frac{(\epsilon-\eta)F(\uptau_{G}(x)|x)}{\tau-\uptau_{G}(x)}
\right\rbrace
\geq0.\]
Since $0<\eta<\epsilon$ is negligible, one possible decision rule is: 
\[
H_0 \text{ is rejected if } \max_{x\in\family{X}}\hat{T}_n(x)<r_n^{-1}\xi_{\alpha},
\]
where $\xi_{\alpha}$ is the $\alpha$-quantile of the asymptotic distribution of
$r_n\lbrace\max_{x}\hat{T}_n(x)-\max_{x}{T}(x)\rbrace$.
The asymptotic distribution can be obtained using the traditional Delta method if $T(x)\neq T(y)$ for $x\neq y$, or using the Delta method for Hadamard directional differentiable maps \citep{FS2019}. Specifically, for each $x\in\family{X}$, the asymptotic distribution of $r_n\lbrace\hat{T}_n(x)-T(x)\rbrace$
is the same as that
of $r_n\lbrace \hat{f}_{nb_x}(\uptau_{G}(x)|x)-{f}(\uptau_{G}(x)|x)\rbrace$. 
Denote 
$x_\ast=\argmax_{x\in\family{X}}T(x)$
and assume for simplicity that $x_\ast$ is unique.
Then, applying the Delta method for the maximum function, the limit distribution of 
$r_n\lbrace\max_{x\in\family{X}}\hat{T}_n(x)-\max_{x\in\family{X}}{T}(x)\rbrace$
is the same as that of $r_n\lbrace\hat{f}_{nb_x}(\uptau_{G}(x_\ast)|x_\ast)-{f}(\uptau_{G}(x_\ast)|x_\ast)\rbrace$.
Hence, $\xi_{\alpha}$ is equivalent to $\xi_{x_\ast,\alpha}$  used in the first method. 

While this method is asymptotically sound, the negative bias and variability of the smooth Grenander estimator at the right boundary, $\hat{f}_{nb_x}(Y_{x(n_x)}|x)$, in finite sample sizes make it challenging to estimate $\max_{x\in\family{X}}T(x)$ using $\max_{x\in\family{X}}\hat{T}_n(x)$, leading to poor practical performance.
Therefore, we do not pursue this test procedure further. Instead, we introduce a method for estimating $x_\ast$ using bootstrap samples, denoted by $\hat{x}_{\ast,n}$. By selecting the corresponding subsample with $x=\hat{x}_{\ast,n}$, we then apply the procedure from \cite{YM2024} on this subsample for testing $H_{0x_\ast}$ and hence also $H_0$, since $H_0$ implies
$H_{0x_\ast}$.

To obtain the bootstrap samples, we can apply the same bootstrap procedure as described in the first method.
Using the bootstrap subsample corresponding to $X=x$, we compute the bootstrap estimate 
\[
\hat{T}_n^\star(x)=\hat{f}_{nb_x}^{\star}(Y_{x(n_x)}^\star|x)-\frac{\epsilon \hat{F}_{n}^{\star}(Y_{x(n_x)}^\star|x)}{\tau-Y_{x(n_x)}^\star},
\]
where $Y_{x(n_x)}^\star$ is the maximum survival time for the bootstrap subsample corresponding to $X=x$. With the bootstrap estimate $\hat{T}_n^\star(x)$, we introduce an estimator for $x_\ast$ as follows:
Let $Q_{1-\gamma,n}^\star(x)$ be the $(1-\gamma)$-quantile of $\hat{T}_n^\star(x)$, for $\gamma\in(0,1)$. The estimator \[\hat{x}_{\ast,n}=\argmax_{x\in\family{X}}Q_{1-\gamma,n}^\star(x)\]
is then used to estimate $x_\ast$. We will show its consistency for $x_\ast$ given that the bootstrap is consistent.
While the above holds theoretically for any $\gamma\in(0,1)$, in practice, we choose a small $\gamma$, say $\gamma=0.025$, to estimate $x_\ast$ better based on higher quantiles. 

With the selected subsample corresponding to $x=\hat{x}_{\ast,n}$, $H_0$ is rejected if
\begin{equation}
	\label{eq:test_M2}
	\hat{T}_n(\hat{x}_{\ast,n})<r_n^{-1}\xi_{\hat{x}_{\ast,n},\alpha},
\end{equation}
where $\xi_{x,\alpha}$ is the $\alpha$-quantile of the limit distribution of $r_n\lbrace\hat{f}_{nb_x}(Y_{x(n_x)}|x) - {f}(\uptau_G(x)|x)\rbrace$. $\xi_{x,\alpha}$ is approximated from the bootstrap estimate $r_n\lbrace\hat{f}_{nb_x}^\star(Y_{x(n_x)}^\star|x) - \hat{f}_{nb_0^x}(Y_{x(n_x)}|x)\rbrace$,
where $\hat{f}_{nb_0^x}(\cdot|x)$ is the smoothed Grenander estimator
used to generate the bootstrap sample.
Essentially, using the estimate $\hat{x}_{\ast,n}$, we apply the procedure from \cite{YM2024} on the subsample corresponding to $x=\hat{x}_{\ast,n}$ for testing $H_0$. 
Since the tests for each subsample have an asymptotic level of $\alpha$ and $\hat{x}_{\ast,n}$ is consistent for $x_\ast$, the test with decision rule given in \eqref{eq:test_M2} also has an asymptotic level of $\alpha$ (see Theorem~\ref{theorem:test_discrete_level}\ref{enum:test_M2_level} in Section~\ref{sec:asym_prop}). Specifically, under $H_0: q_{1-\epsilon}(x)\geq\uptau_{G}(x)$ for some $x\in\family{X}$, we have
\[
\begin{split}
	&\prob{\hat{T}_n(\hat{x}_{\ast,n})< {{r_n^{-1}}\xi_{\hat{x}_{\ast,n},\alpha}}\bigg| H_0}\\
	&\leq
	\prob{\hat{T}_n(\hat{x}_{\ast,n})< {{r_n^{-1}}\xi_{\hat{x}_{\ast,n},\alpha},\ \hat{x}_{\ast,n}=x_{\ast}}\bigg| H_0}
	+
	\prob{\hat{x}_{\ast,n}\neq x_{\ast}\bigg| H_0}\\
	&\leq
	\prob{\hat{T}_n(x_\ast)-T(x_\ast)< {{r_n^{-1}}\xi_{x_\ast,\alpha}}\bigg| H_0}
	+
	\prob{\hat{x}_{\ast,n}\neq x_{\ast}\bigg| H_0},
\end{split}
\]
where the limit of the first probability on the right-hand side of the last inequality is bounded above by $\alpha$, while the second converges to zero by the consistency of $\hat{x}_{\ast,n}$. We note that $x_\ast$ can be non-unique and we consider this more general case in the proof of Theorem~\ref{theorem:test_discrete_level}\ref{enum:test_M2_level} below.

We note that the bootstrap samples are generated using the smoothed Grenander estimator, $\hat{f}_{nb_0^x}(\cdot|x)$, without imposing additional constraints related to $H_0$ or $H_1$. The fact that this is sufficient to determine the critical value for the prespecified significance level $\alpha$ follows from Proposition 2 of \cite{YM2024} in the unconditional setting.
From there one can see that, using the inequality \eqref{eq:insuff_follow_up_conseq} which holds under $H_0$, we only need to estimate quantiles of $\hat{f}_{nb_x}(Y_{x(n_x)}|x)- f(\tau_G(x)|x)$, 
using quantiles of the bootstrap estimates
$
\hat{f}_{nb_x}^\star(Y_{x(n_x)}^\star|x)-\hat{f}_{nb_0^x}(Y_{x(n_x)}|x)
$.

\section{Asymptotic properties}
\label{sec:asym_prop}
The asymptotic distribution of the smoothed Grenander estimator of $f$ in the interior of the support under the unconditional setting has been studied by \cite{LM2017}. \cite{YM2024} further extended the result to the estimator of $f(\uptau_G)$ at the right boundary $\uptau_{G}$. Based on such a result, we study the asymptotic normality of the smoothed Grenander estimator $\hat{f}_{nb_x}(Y_{x(n_x)}|x)$ in the conditional setting. We consider the following assumptions, which are required for establishing the asymptotic results for the test statistics:
\begin{enumerate}[label=(A\arabic*), series=assumptions_discrete]
	\item\label{assumption:kernel_time}
	The kernel function $k$ is a symmetric, twice continuously differentiable kernel with support $[-1,1]$ such that $\int k(u)\dd{u} = 1$ and its first derivative is bounded uniformly.
	\item\label{assumption:discrete_cov}
	The probability mass function of the discrete covariate $X$, $\prob{X=x}=\rho_x\in(0,1)$, has a finite support $\family{X}$.
	\item\label{assumption:uncured_survival_discrete_cov}
	The first derivative of $F_u(t|x)$ with respect to $t$ exists for each $x\in\family{X}$, denoted by $f_u(t|x)$. Furthermore, for each $x\in\family{X}$, $f_u(t|x)$ is non-increasing and twice continuously differentiable on $[a_x,\uptau_{F_{u}}(x)]$, for some $ 0 \leq  a_x  <\uptau_{G}(x)$, and satisfies $0<\inf_{t\in[a_x,\uptau_{F_{u}}(x)]}|f_u^\prime(t|x)|\leq\sup_{t\in[a_x,\uptau_{F_{u}}(x)]}|f_u^\prime(t|x)|<\infty$.
	\item\label{assumption:censoring_discrete_cov}
	The conditional distribution of the censoring time $C$, denoted by $G(\cdot|x)$, is continuous on $[0,\uptau_G(x))$ and has a positive mass at $\uptau_G(x)$ for each $x\in\family{X}$; i.e., $\Delta G(\uptau_G(x))=1-G(\uptau_G(x){-}|x)>0$. Furthermore, for each $x\in\family{X}$, the first derivative of $G(t|x)$ with respect to $t$ exists, denoted by $g(t|x)$, and $g(t|x)$ is continuous on $[0,\uptau_G(x))$ for all $x\in\family{X}$.
	\item\label{assumption:positive_cure_prob}
	The incidence $1-p(x)$ satisfies $1-p(x) <1$ for all $x\in\family{X}$, meaning that cure does not happen with probability one, and $F_u(\uptau_G(x)|x)<1$ for all $x\in\family{X}$, meaning that there is positive probability for the event to happen after the end of the study.
\end{enumerate}
Assumption~\ref*{assumption:kernel_time} is common under kernel density estimation and is required in Theorem~\ref*{thm:sg_discrete_normality}. Assumptions~\ref*{assumption:discrete_cov} and \ref*{assumption:censoring_discrete_cov} are necessary in Theorem~\ref*{thm:sg_discrete_normality} since its proof requires the strong representation of the conditional Kaplan--Meier estimator in Lemma~\ref*{lemma:strong_rep_brownian_disc_x} and the order of the supremum distance between the conditional Kaplan--Meier estimator and its least concave majorant in Lemma~\ref*{lemma:lcm_kme_kme_unif_dist_order}.
Assumption~\ref*{assumption:uncured_survival_discrete_cov} is required in Theorem~\ref*{thm:sg_discrete_normality} and Lemma~\ref*{lemma:lcm_kme_kme_unif_dist_order}. 
We note that the monotonicity and twice continuously differentiability assumptions of $f_u(\cdot|x)$ in the right tail are needed in Lemma~\ref*{lemma:lcm_kme_kme_unif_dist_order}. In practice, since $a_x$ is unknown, we need to identify an interval where the non-increasing assumption is reasonable. This can be done by comparing the Kaplan-Meier curve and its least concave majorant (see Section~\ref{sec:app}).
Assumption~\ref*{assumption:positive_cure_prob} is required in Theorem~\ref*{thm:sg_discrete_normality} and Lemma~\ref*{lemma:max_surv_time_discrete_conv}.
Note that \ref*{assumption:positive_cure_prob} implies that $\uptau_{G}(x)<\uptau_{F_{u}}(x)$, which is necessary for constructing the critical region for testing $H_0: q_{1-\epsilon}(x)\geq \uptau_{G}(x)$ for some $x\in\family{X}$ with an asymptotic level $\alpha$. This means that, even if the follow-up is `practically' sufficient, the event is still possible after the end of the study and is usually quite realistic in practice.

\begin{theorem}
	\label{thm:sg_discrete_normality}
	Suppose that
	Assumptions~\ref*{assumption:kernel_time}--\ref*{assumption:positive_cure_prob} hold.
	Let $c_x\in(0, \infty)$.
	\begin{enumerate}[label={(\roman*)}]
		\item\label{enum:sg_discrete_normality_cond}
		If $b_x=c_xn_x^{-1/5}$, then we have,
		conditional on $\mat{X}_1^n=(X_1,\cdots,X_n)$,
		\begin{equation}
			\label{eq:sg_discrete_normality_cond}
			\left.
			n^{2/5}\left\lbrace{
				\hat{f}_{nb_x}(Y_{x(n_x)}|x) - f(\uptau_{G}(x)|x) 
			}\right\rbrace\middle\vert\mat{X}_1^n
			\right.
			\converge[d]
			N(\rho_x^{-2/5}\mu_x, \rho_x^{-4/5}\sigma_x^2)
		\end{equation}
		almost surely with respect to $\mat{X}_1^n$;
		\item\label{enum:sg_discrete_normality_uncond}
		If $b_x=c_xn_x^{-1/5}$, then we have (without conditioning on $\mat{X}_1^n$)
		\begin{equation}
			\label{eq:sg_discrete_normality_uncond}
			n^{2/5}\left\lbrace{
				\hat{f}_{nb_x}(Y_{x(n_x)}|x) - f(\uptau_{G}(x)|x) 
			}\right\rbrace
			\converge[d]
			N(\rho_x^{-2/5}\mu_x, \rho_x^{-4/5}\sigma_x^2),
		\end{equation}
		\item\label{enum:sg_discrete_normality_uncond_undersmooth}
		If $b_x=c_xn_x^{-\kappa}$ with $1/5<\kappa<1/3$ (undersmoothing), then we have
		\begin{equation}
			\label{eq:sg_discrete_normality_uncond_undersmooth}
			n^{(1-\kappa)/2}\left\lbrace{
				\hat{f}_{nb_x}(Y_{x(n_x)}|x) - f(\uptau_{G}(x)|x) 
			}\right\rbrace
			\converge[d]
			N(0, \rho_x^{-(1-\kappa)}\sigma_x^2),
		\end{equation}
	\end{enumerate}
	where
	\[
		\mu_x = \frac{1}{2}c_x^{2}f^{\prime\prime}(\uptau_{G}(x)|x)\int_{0}^{1}v^{2}k_{B}(v)\dd{v},\
		\sigma_x^2 = \frac{f(\uptau_{G}(x)|x)}{c_x\rbracket{1 - G(\uptau_{G}(x){-}|x)}}\int_{0}^{1}k_{B}^{2}(v)\dd{v},
	\]
	with $k_{B}(v)=\phi(0)k(v)-\psi(0)vk(v)$, and the coefficients $\phi(0)$ and $\psi(0)$ are defined as in~\eqref{eqn:phi,psi}.
\end{theorem}
While the first statement of the theorem follows immediately from the result about the smoothed Grenander estimate in the unconditional setting (see Theorem~4.4 in \cite{LM2017}), obtaining the second statement is more technical because of the lack of a strong approximation result for the conditional Kaplan--Meier estimator with categorical covariates as well as the randomness of the subsample size $n_x$.
We also note that the bandwidth is allowed to depend on the subsample size $n_x$ instead of $n$. If $n$ is used, 
the mean and variance of the normal distribution in \eqref{eq:sg_discrete_normality_uncond} become $\mu_x$ and $\sigma_x^2/\rho_x$, respectively, and similarly for the limiting variance in \eqref{eq:sg_discrete_normality_uncond_undersmooth}.
With the asymptotic normality of the smoothed Grenander estimator $\hat{f}_{nb_x}(Y_{x(n_x)}|x)$, and assuming that $\eta$ (defined before \eqref{eq:insuff_follow_up_conseq}) decays faster than a certain rate, we can show that the levels of the tests in \eqref{eq:test_discrete}
and in \eqref{eq:test_M2} are asymptotically bounded by $\alpha$.
\begin{theorem}
	\label{theorem:test_discrete_level}
	Suppose that Assumptions~\ref*{assumption:kernel_time}--\ref*{assumption:positive_cure_prob} hold, and $b_x=c_xn_x^{-\kappa}$ for some $c_x\in(0,\infty)$ and $\kappa\in(1/5,1/3)$. If $0<\eta_n<\epsilon$ is such that $n^{(1-\kappa)/2}\eta_n\converge 0$ as $n\converge\infty$, then, under $H_0: q_{1-\epsilon}(z) \geq \uptau_{G}(z)$ for some $z\in\family{X}$, 
	\begin{enumerate}[label={(\roman*)}]
		\item\label{enum:test_M1_level}
		for the first method using \eqref{eq:test_discrete}, we have
		\[
		\lim_{n\converge\infty}\prob{\bigcap_{x\in\family{X}}\cbracket{
				\hat{T}_n(x)< n^{-(1-\kappa)/2}\xi_{x,\alpha}
			}\middle\vert H_0}\leq\alpha;
		\]
		\item\label{enum:test_M2_level}
		for the second method using \eqref{eq:test_M2},
		if the bootstrap is consistent, i.e., for each $x\in\family{X}$
		\[
		\sup_{t\in\real}\left\vert
		\mathbb{P}_n^\star\left({
			n^{(1-\kappa)/2}\lbrace\hat{T}_n^\star(x)-\tilde{T}_n(x)\rbrace
			\leq t
		}\right)
		-\Phi\left(\frac{t}{\rho_x^{-(1-\kappa)/2}\sigma_x}\right)
		\right\vert\converge[\mathbb{P}]0,
		\]
		then, we have
		\[
		\lim_{n\converge\infty}\prob{
			\hat{T}_n(\hat{x}_{\ast,n})< {{n^{-(1-\kappa)/2}}\xi_{\hat{x}_{\ast,n},\alpha}}
			\middle\vert H_0}\leq\alpha,
		\]
	\end{enumerate}
	where $\tilde{T}_n(x)$ equals $\hat{T}_n(x)$ except $\hat{f}_{nb_x}(Y_{x(n_x)}|x)$ being replaced by $\hat{f}_{nb_0^x}(Y_{x(n_x)}|x)$, $\xi_{x,\alpha}$ is the $\alpha$-quantile of the normal distribution in \eqref{eq:sg_discrete_normality_uncond_undersmooth}, $\mathbb{P}_n^\star$ is the conditional probability given the original observations, and $\Phi$ is the distribution function of the standard normal distribution.
\end{theorem}
We comment on the consistency of the smoothed bootstrap procedure described in Section~\ref{sec:proc} when an undersmoothed bandwidth $b_x=c_xn_x^{-\kappa}$, $\kappa\in(1/5,1/3)$ is used. This requires bootstrap versions of Assumptions~\ref*{assumption:uncured_survival_discrete_cov}--\ref*{assumption:censoring_discrete_cov} to hold, meaning that conditionally on the original observations, these assumptions with $f(\cdot|x)$ and $G(\cdot|x)$ replaced by their bootstrap version, $\hat{f}_{n,b_0^x}(\cdot|x)$ and $\hat{G}_n(\cdot|x)$, hold on an event with probability converging to one. Specifically, it can be shown by arguments similar to those used in the proof of Lemma~5 in the supplement of \cite{DGL2013} and the proof of Lemma~2.1 in \cite{K2008} that the bootstrap version of Assumption~\ref*{assumption:uncured_survival_discrete_cov} holds. For the bootstrap version of \ref*{assumption:censoring_discrete_cov}, we note that,
at the price of more complicated technical details,
the assumption of $G(\cdot|x)$ having a continuous density can be relaxed to $G(\cdot|x)$ being Lipschitz as in \cite{DA2002}. Then we can consider a continuous version of the KME $\hat{G}_n(\cdot|x)$ as in the proof of Theorem~3.6 in \cite{DGL2013}, which is a piecewise linear function.
Nevertheless, we do not expect such a modification to lead to an improvement in
the practical performance of the test.
Furthermore, if $b_x=c_xn_x^{-1/5}$ is used, similar results to Theorem~\ref{theorem:test_discrete_level} can be obtained provided that the bootstrap is consistent. This requires in addition that the second derivative of the smooth Grenander estimator $\hat{f}_{nb_0^x}$ (used to generate the bootstrap sample) at $\uptau_G(x)$ is consistent for $f''(\uptau_G(x)|x)$, which appears in the asymptotic bias. 
Since
estimation of the second derivative is
more difficult at the boundary, we chose to avoid it by using undersmoothing. 

\section{Simulation study}
\label{sec:simulations}
In this section, we investigate the finite sample performance of the testing procedures introduced in Section~\ref{sec:proc}. In order to include various scenarios with different survival time and censoring time distributions conditional on categorical covariates, we consider four settings as described below.

\subsection{Simulation settings}
\label{sec:sim_settings}
In Settings~\ref*{enum:sim_logistic_weibull_unif}, \ref*{enum:sim_exp_exp} and \ref*{enum:sim_exp_unif} below, a covariate $X$ having a Bernoulli distribution with parameter $\rho\in\lbrace0.3,0.5\rbrace$ is considered to study the effect of $\rho$ on the test performance. In Setting~\ref{enum:sim_exp_unif_2cov}, two dependent binary covariates, $X_1$ and $X_2$, are considered. The conditional distribution of the uncured event time is assumed to have a decreasing density. The censoring time $C$ given $X=x$ is generated as $C=\min(\tilde{C},\uptau_{G}(x))$, where $\tilde{C}$ may depend on the covariate value $x$ and has a support that includes $[0,\uptau_G(x)]$. To study the influence
of $\Delta G(\uptau_{G}(x))$, 
which affects the limit variance of the estimator $\hat{f}_{nb_x}(Y_{x(n_x)}|x)$, we consider $\tilde{C}$ with a uniform distribution on $[0,\zeta]$, where $\zeta\geq\uptau_{G}(x)$ is chosen such that $\Delta G(\uptau_{G}(x))\in\lbrace0,0.01\rbrace$ in
Setting~\ref*{enum:sim_exp_unif}.

For all four settings, 5 different $\uptau_G(x)$'s are considered to examine the empirical level and power of the tests. Specifically, $\uptau_{G}(x)$'s are chosen based on quantiles of $F_u(\cdot|x)$ for each $x\in\family{X}$. We use the 95, 97.5, 99, 99.5 and 99.9\% quantiles of $F_u(\cdot|x)$. For instance, in the scenarios where $\family{X}=\lbrace0,1\rbrace$, $\uptau_{G}(0)$ and $\uptau_{G}(1)$ being the 99.5\% or 99.9\% quantiles of the corresponding $F_u(\cdot|x)$ are considered as sufficient follow-up when $\epsilon=0.01$.

\begin{enumerate}[label={\textit{Setting \arabic*.}},align=left,wide,ref=\arabic*]
	\item\label{enum:sim_logistic_weibull_unif}
	The covariate $X$ has a Bernoulli distribution with parameter $\rho\in\lbrace0.3,0.5\rbrace$.
	The uncure fraction $p(x)$ is:
	\[
	p(x) = \frac{1}{1+\exp(0.6-1.3x)},\quad x=0,1.
	\]
	In particular, $p(0)\approx0.35$ and $p(1)\approx0.67$.
	The uncured subjects
	have a Weibull conditional distribution
	$
	F_u(t|x)=1-\exp\lbrace-\sigma\exp(x\beta)t^a\rbrace,\ x=0,1,
	$
	where $\sigma=1.2$, $\beta=-0.9$ and $a=0.8$.
	The censoring times given $X=x$ are generated using the aforementioned uniform distribution for $\tilde{C}$
	with $\Delta G(\uptau_{G}(x))=0.01$.
	\item\label{enum:sim_exp_exp}
	The covariate $X$ has a Bernoulli distribution with parameter $\rho\in\lbrace0.3,0.5\rbrace$;
	The uncure fraction $p(x)\in\lbrace0.3,0.7\rbrace$ for $x=0,1$;
	The uncured subjects 
	have an exponential conditional distribution
	$
	F_u(t|x)=1-\exp\lbrace-(5-0.5x)t\rbrace,\ x=0,1.
	$
	The censoring times given $X=x$ are generated with $\tilde{C}$ having the conditional distribution function:
	$
	G(t|x)=1-\exp\lbrace-(1+1.5x)t\rbrace,\ x=0,1.
	$
	\item\label{enum:sim_exp_unif_2cov}
	The covariates $(X_1,X_2)$ are two binary random variables:
	$X_1\sim B(1,0.4)$ and $X_2|X_1\sim B(1,\rho(X_1))$,
	where $\rho(0)=0.4$ and $\rho(1)=0.6$. Therefore, 
	the joint probability mass function of $(X_1,X_2)$ is $\mathbb{P}(X_1=x_1, X_2=x_2)=0.36/0.24/0.16/0.24$ if $(x_1,x_2)=(0,0)/(0,1)/(1,0)/(1,1)$ and
	$\family{X}{=}\lbrace(0,0),(0,1),(1,0),(1,1)\rbrace$.
	The uncure fraction $p(x_1,x_2)$ is:
	\[
	p(x_1,x_2)=\frac{1}{1+\exp(0.6-1.3x_1+0.3x_2)},\quad(x_1,x_2)\in\family{X}.
	\]
	Specifically, $p(x_1,x_2)\approx0.35/0.29/0.67/0.60$ when $(x_1,x_2)=(0,0)/(0,1)/(1,0)/(1,1)$.
	The uncured subjects 
	have an exponential conditional distribution
	\[
	F_u(t|x_1,x_2)=1-\exp\lbrace-(1-0.25x_1-0.3x_2)t\rbrace,\quad(x_1,x_2)\in\family{X}.
	\]
	The censoring times given $X_1=x_1$, $X_2=x_2$ are generated from
	$C=\min\lbrace\tilde{C},\uptau_{G}(x_1,x_2)\rbrace$ with $\tilde{C}$ being the aforementioned uniform distribution such that $\Delta G(\uptau_G(x_1,x_2))=0.01$. Here, $\uptau_{G}(x_1,x_2)$ are the 95, 97.5, 99, 99.5 and 99.9\% quantiles of the corresponding $F_u(t|x_1,x_2)$.
	\item\label{enum:sim_exp_unif}
	The covariate $X$ has a Bernoulli distribution with parameter $\rho\in\lbrace0.3,0.5\rbrace$.
	The uncure fraction $p(x)\in\lbrace0.4,0.6\rbrace$ for $x=0,1$.
	The uncured subjects given $X=x$ have a standard exponential distribution for $x=0,1$.
	The censoring times given $X=x$ are generated using the aforementioned uniform distribution for $\tilde{C}$
	with $\Delta G(\uptau_{G}(x))\in\lbrace0,0.01\rbrace$.
\end{enumerate}

In the above four settings, different conditional distributions for the uncured survival times and censoring times are considered. This allows us to cover various scenarios in which subjects corresponding to different values of the covariates
may have different uncured survival or censoring distributions, alongside the uncure fraction that varies with $x$.
In Setting~\ref*{enum:sim_logistic_weibull_unif}, the uncured survival time conditional on the covariate follows a Weibull distribution with a shape parameter of 0.8, which has a density that decreases faster than that of the standard exponential distribution.
In Setting~\ref*{enum:sim_exp_exp}, both the uncured survival time and the censoring time have exponential distributions with rate parameter that varies with the covariate, while a non-uniform censoring distribution is considered as compared with the other three settings.
Setting~\ref*{enum:sim_exp_unif_2cov} is a more realistic setting that consists of two binary covariates leading to four possible values of
$(x_1,x_2)$. Every sub-population corresponding to different values of
$(x_1,x_2)$ has a different uncure fraction and uncured survival time distribution.
Setting~\ref*{enum:sim_exp_unif}, in which the uncured subjects have
the same exponential distribution for all values of $x$, is used
to study the sole effect of $p(x)$, $\rho$, and $\Delta G(\uptau_G(x))$ on the test performance. Tables~S1--S7 in Section~S1.1 in the Supplementary Material report the censoring rates (at population and sub-population level) for all considered cases.

For the proposed test, $\epsilon=0.01$ is set, indicating that the follow-up is considered as sufficient when all $\uptau_{G}(x)$'s are greater than the 99\% quantiles of the corresponding $F_u(\cdot|x)$. The parameter $\tau$ is set to the maximum of the 99.99\% quantiles of $F_u(\cdot|x)$ over $x\in\family{X}$, i.e., $\tau=\max_{x\in\family{X}}q_{0.9999}(x)$. For the smoothed Grenander estimator $\hat{f}_{nb_x}$, the tri-weight kernel is used together with an undersmoothing bandwidth $b_x=Y_{x(n_x)}(n_x^{-7/30}\wedge0.5)$.
Such bandwidth of order $n^{-7/30}$ is chosen
since we want to undersmooth but also remain
close to the optimal bandwidth of order $n^{-1/5}$. This seems to perform well in practice and the purpose was to use a fixed bandwidth in order to avoid data-driven bandwidth selection procedures which would significantly increase the computational cost. Fixed bandwidths are common in the literature of smooth monotone estimators.
For the smoothed bootstrap procedure, the bandwidth $b_0^x=Y_{x(n_x)}(n_x^{-1/9}\wedge0.5)$ is used,
as proposed by \cite{GJ2024} in the monotone regression setting. We use 500 bootstrap samples to compute the critical values of the test procedures. 
Since the conditional density functions of the uncured survival times are decreasing in all four settings, we set $a_x=0$.
For the second method, $\gamma=0.025$ is used to compute
the estimator 
$\hat{x}_{\ast,n}$ of $x_\ast=\argmax_{x\in\family{X}}T(x)$. The significance level of the tests is set at 0.05.

For Setting~\ref*{enum:sim_logistic_weibull_unif}, a sample size of 1000 is considered. Two sample sizes (500 and 1000) are considered for Settings~\ref*{enum:sim_exp_exp} and \ref*{enum:sim_exp_unif}.
For Setting~\ref*{enum:sim_exp_unif_2cov}, since it consists of two binary covariates, meaning that there are 4 subsamples, a larger sample size of 2000 is considered.
Recall that 5 different $\uptau_{G}(x)$'s are considered for each $x\in\family{X}$, which leads to numerous scenarios ($5^{|\family{X}|}$ combinations of $\uptau_{G}(x)$'s). Instead of covering all combinations, we first group them into three main scenarios (Scenario A, B and C) and then consider some choices from each of these three scenarios to study the empirical level and power of the test procedures.
Denote the 99\% quantile of $F_u(\cdot|x)$ by $q_{0.99}(x)$.
Scenario~B consists of those cases with $q_{0.99}(x)\leq\uptau_{G}(x)$ for all $x\in\family{X}$ while $\uptau_{G}(x)=q_{0.975}(x)$ for at most one $x\in\family{X}$. This means that the follow-up is sufficient or nearly sufficient in the conditional setting ($H_{1x}$ in \eqref{eq:suff_follow_up_hyp_fix_cov}) for all but at most one $x\in\family{X}$. The cases under this scenario are considered as insufficient follow-up over all covariate values ($H_0$ in \eqref{eq:suff_follow_up_hyp_all_cov}). In practice, these borderline cases are considered as sufficient or nearly sufficient follow-up.
Alternatively, Scenario C consists of those cases with $q_{0.99}(x)<\uptau_{G}(x)$ for all $x\in\family{X}$, meaning that the follow-up is sufficient in the conditional setting ($H_{1x}$ in \eqref{eq:suff_follow_up_hyp_fix_cov}) for each $x\in\family{X}$. Therefore the cases under this scenario are considered as sufficient follow-up over all covariate values ($H_1$ in \eqref{eq:suff_follow_up_hyp_all_cov}) and hence are used for investigating the empirical power of the test. 
Scenario A consists of the remaining cases that are not in Scenarios B and C.
These cases are considered as insufficient follow-up ($H_0$) and are useful to examine the empirical level of the test.%

\subsection{Simulation results}
We report the rejection rate of $H_0$ for the first and second methods together with the rejection rate of $H_{0x}$ (insufficient follow-up for fixed $x$ in the conditional setting in \eqref{eq:suff_follow_up_hyp_fix_cov}). The rejection rates are computed from 500 replications. 
The results for Settings~\ref*{enum:sim_exp_exp} and \ref*{enum:sim_exp_unif_2cov} are presented in this section, while the remaining results can be found in the Supplementary Material. 
An overall performance summary will be presented at the end of this section.%

\begin{table}[t]
	\begin{center}
		\caption{Rejection rate for Setting~\ref*{enum:sim_exp_exp} with $\rho=0.5$, $p(0)=0.7$, $p(1)=0.7$, $n=1000$.
				`Scenario' groups cases as A (insufficient follow-up, $H_0$), B (borderline, $H_0$), and C (sufficient follow-up, $H_1$).
				Columns under $\uptau_G(x)$ are follow-up lengths in quantiles of $F_u(\cdot|x)$. `Rej. rate of $H_{0x}$' gives rejection rates for fixed $x=0/1$ using \cite{YM2024}. Columns M1/M2 report rejection rates of $H_0$ for the first/second method. The significance level $\alpha$ is 0.05.
			\label{tab:sim_exp_exp_rho_05_p0_07_p1_07_n_1000}}
		\footnotesize
		\begin{tabular}{ccc@{\extracolsep{12pt}}cc@{\extracolsep{12pt}}cc}
\multirow{2}{*}{Scenario} & \multicolumn{2}{c}{$\uptau_G(x)$} & \multicolumn{2}{c}{Rej. rate of $H_{0x}$} & \multicolumn{2}{c}{Rej. rate of $H_0$} \\ \cline{2-3}\cline{4-5}\cline{6-7}
& $x=0$           & $x=1$           & $x=0$                  & $x=1$                 & M1             & M2             \\ \hline
\multirow{4}{*}{A} & 0.950 & 0.950 & 0.030 & 0.038 & 0.000 & 0.020 \\ 
   & 0.950 & 0.975 & 0.030 & 0.034 & 0.004 & 0.026 \\ 
   & 0.975 & 0.975 & 0.034 & 0.032 & 0.000 & 0.034 \\ 
   & 0.999 & 0.950 & 0.894 & 0.036 & 0.030 & 0.036 \\ 
   \hline
\multirow{7}{*}{B} & 0.975 & 0.990 & 0.034 & 0.142 & 0.006 & 0.032 \\ 
   & 0.975 & 0.999 & 0.034 & 0.936 & 0.032 & 0.036 \\ 
   & 0.990 & 0.975 & 0.064 & 0.034 & 0.002 & 0.046 \\ 
   & 0.990 & 0.990 & 0.064 & 0.136 & 0.012 & 0.066 \\ 
   & 0.990 & 0.999 & 0.064 & 0.930 & 0.060 & 0.064 \\ 
   & 0.995 & 0.975 & 0.172 & 0.032 & 0.008 & 0.058 \\ 
   & 0.999 & 0.990 & 0.894 & 0.150 & 0.140 & 0.156 \\ 
   \hline
\multirow{4}{*}{C} & 0.995 & 0.995 & 0.172 & 0.452 & 0.054 & 0.158 \\ 
   & 0.995 & 0.999 & 0.172 & 0.936 & 0.162 & 0.174 \\ 
   & 0.999 & 0.995 & 0.894 & 0.450 & 0.382 & 0.436 \\ 
   & 0.999 & 0.999 & 0.894 & 0.928 & 0.832 & 0.890 \\ 
   \hline
\end{tabular}

	\end{center}
\end{table}
\begin{table}[t]
	\begin{center}
		\caption{Rejection rate for Setting~\ref*{enum:sim_exp_exp} with $\rho=0.3$, $p(0)=0.7$, $p(1)=0.7$, $n=1000$.
				`Scenario' groups cases as A (insufficient follow-up, $H_0$), B (borderline, $H_0$), and C (sufficient follow-up, $H_1$).
				Columns under $\uptau_G(x)$ are follow-up lengths in quantiles of $F_u(\cdot|x)$. `Rej. rate of $H_{0x}$' gives rejection rates for fixed $x=0/1$ using \cite{YM2024}. Columns M1/M2 report rejection rates of $H_0$ for the first/second method. The significance level $\alpha$ is 0.05.
			\label{tab:sim_exp_exp_rho_03_p0_07_p1_07_n_1000}}
		\footnotesize
		\begin{tabular}{ccc@{\extracolsep{12pt}}cc@{\extracolsep{12pt}}cc}
\multirow{2}{*}{Scenario} & \multicolumn{2}{c}{$\uptau_G(x)$} & \multicolumn{2}{c}{Rej. rate of $H_{0x}$} & \multicolumn{2}{c}{Rej. rate of $H_0$} \\ \cline{2-3}\cline{4-5}\cline{6-7}
& $x=0$           & $x=1$           & $x=0$                  & $x=1$                 & M1             & M2             \\ \hline
\multirow{4}{*}{A} & 0.950 & 0.950 & 0.012 & 0.020 & 0.000 & 0.018 \\ 
   & 0.950 & 0.975 & 0.012 & 0.036 & 0.000 & 0.010 \\ 
   & 0.975 & 0.975 & 0.044 & 0.030 & 0.000 & 0.032 \\ 
   & 0.999 & 0.950 & 0.936 & 0.024 & 0.022 & 0.028 \\ 
   \hline
\multirow{7}{*}{B} & 0.975 & 0.990 & 0.044 & 0.218 & 0.010 & 0.044 \\ 
   & 0.975 & 0.999 & 0.044 & 0.896 & 0.040 & 0.046 \\ 
   & 0.990 & 0.975 & 0.092 & 0.034 & 0.008 & 0.058 \\ 
   & 0.990 & 0.990 & 0.092 & 0.222 & 0.018 & 0.086 \\ 
   & 0.990 & 0.999 & 0.092 & 0.896 & 0.084 & 0.092 \\ 
   & 0.995 & 0.975 & 0.184 & 0.034 & 0.012 & 0.080 \\ 
   & 0.999 & 0.990 & 0.936 & 0.222 & 0.208 & 0.244 \\ 
   \hline
\multirow{4}{*}{C} & 0.995 & 0.995 & 0.184 & 0.536 & 0.090 & 0.196 \\ 
   & 0.995 & 0.999 & 0.184 & 0.900 & 0.172 & 0.190 \\ 
   & 0.999 & 0.995 & 0.936 & 0.526 & 0.484 & 0.508 \\ 
   & 0.999 & 0.999 & 0.936 & 0.892 & 0.834 & 0.852 \\ 
   \hline
\end{tabular}

	\end{center}
\end{table}
\begin{table}[t]
	\begin{center}
		\caption{Rejection rate for Setting~\ref*{enum:sim_exp_exp} with $\rho=0.5$, $p(0)=0.3$, $p(1)=0.7$, $n=1000$.
				`Scenario' groups cases as A (insufficient follow-up, $H_0$), B (borderline, $H_0$), and C (sufficient follow-up, $H_1$).
				Columns under $\uptau_G(x)$ are follow-up lengths in quantiles of $F_u(\cdot|x)$. `Rej. rate of $H_{0x}$' gives rejection rates for fixed $x=0/1$ using \cite{YM2024}. Columns M1/M2 report rejection rates of $H_0$ for the first/second method.  The significance level $\alpha$ is 0.05.
			\label{tab:sim_exp_exp_rho_05_p0_03_p1_07_n_1000}}
		\footnotesize
		\begin{tabular}{ccc@{\extracolsep{12pt}}cc@{\extracolsep{12pt}}cc}
\multirow{2}{*}{Scenario} & \multicolumn{2}{c}{$\uptau_G(x)$} & \multicolumn{2}{c}{Rej. rate of $H_{0x}$} & \multicolumn{2}{c}{Rej. rate of $H_0$} \\ \cline{2-3}\cline{4-5}\cline{6-7}
& $x=0$           & $x=1$           & $x=0$                  & $x=1$                 & M1             & M2             \\ \hline
\multirow{4}{*}{A} & 0.950 & 0.950 & 0.032 & 0.038 & 0.000 & 0.028 \\ 
   & 0.950 & 0.975 & 0.032 & 0.038 & 0.004 & 0.026 \\ 
   & 0.975 & 0.975 & 0.010 & 0.032 & 0.000 & 0.030 \\ 
   & 0.999 & 0.950 & 0.874 & 0.040 & 0.032 & 0.048 \\ 
   \hline
\multirow{7}{*}{B} & 0.975 & 0.990 & 0.010 & 0.138 & 0.000 & 0.012 \\ 
   & 0.975 & 0.999 & 0.010 & 0.932 & 0.010 & 0.014 \\ 
   & 0.990 & 0.975 & 0.092 & 0.034 & 0.004 & 0.052 \\ 
   & 0.990 & 0.990 & 0.092 & 0.144 & 0.010 & 0.076 \\ 
   & 0.990 & 0.999 & 0.092 & 0.940 & 0.086 & 0.094 \\ 
   & 0.995 & 0.975 & 0.300 & 0.034 & 0.006 & 0.104 \\ 
   & 0.999 & 0.990 & 0.874 & 0.156 & 0.140 & 0.630 \\ 
   \hline
\multirow{4}{*}{C} & 0.995 & 0.995 & 0.300 & 0.436 & 0.132 & 0.270 \\ 
   & 0.995 & 0.999 & 0.300 & 0.922 & 0.270 & 0.298 \\ 
   & 0.999 & 0.995 & 0.874 & 0.458 & 0.394 & 0.798 \\ 
   & 0.999 & 0.999 & 0.874 & 0.938 & 0.822 & 0.864 \\ 
   \hline
\end{tabular}

	\end{center}
\end{table}
Tables~\ref{tab:sim_exp_exp_rho_05_p0_07_p1_07_n_1000}--\ref{tab:sim_exp_exp_rho_05_p0_03_p1_07_n_1000} show the rejection rate of $H_{0x}$ for $x\in\{0,1\}$
and the rejection rate of $H_0$ for both proposed methods and
selected cases 
of Setting~\ref*{enum:sim_exp_exp} with different $\rho$ and $p(x)$ when $n=1000$. Each column under $\uptau_G(x)$ indicates the length of the follow-up in terms of quantiles
of the corresponding $F_u(\cdot|x)$. 
For example, a case with 0.95 in the column of $x=0$ under $\uptau_G(x)$ means $\uptau_G(0)$ is the 95\% quantile of $F_u(\cdot|0)$. The values in the column $x=0/x=1$ under `Rej. rate of $H_{0x}$'  are the rejection rate of $H_{0x}$ for fixed $x=0/1$ when the procedure of \cite{YM2024} is used for the subsample corresponding to $x=0/1$.

The values in the column M1/M2 are the rejection rate of $H_0$ for the first/second method. From Tables~\ref{tab:sim_exp_exp_rho_05_p0_07_p1_07_n_1000}--\ref{tab:sim_exp_exp_rho_05_p0_03_p1_07_n_1000}, we see that both methods perform well in terms of empirical level for the cases under Scenario A for different $\rho$ and uncure fraction $p(x)$. Method 1 possesses a better control on level compared with Method 2 for the cases under Scenario B, especially for the last case under Scenario B in Table~\ref{tab:sim_exp_exp_rho_05_p0_03_p1_07_n_1000}. Method 2 has higher power than Method 1 for the cases under Scenario C. As discussed in Section~\ref*{sec:proc}, if the powers of the individual tests are not close to one for each $x\in\family{X}$, Method 1 is expected to have low power and hence being conservative. This is also observed in the simulation results, see for example the first two cases under Scenario C in Tables~\ref{tab:sim_exp_exp_rho_05_p0_07_p1_07_n_1000}--\ref{tab:sim_exp_exp_rho_05_p0_03_p1_07_n_1000}.

From Tables~\ref{tab:sim_exp_exp_rho_05_p0_07_p1_07_n_1000}--\ref{tab:sim_exp_exp_rho_03_p0_07_p1_07_n_1000}, we see that the rejection rate of $H_0$ for Method 2 is close to the rejection rate of $H_{0x_\ast}$, especially for the cases under Scenario A. However, we point out that for the scenarios where the follow-up is sufficient or nearly sufficient for each fixed $x\in\family{X}$, such as the last case under Scenario B, Method 2 has a rejection rate larger than the nominal level of 0.05.
This exacerbates when $p(0)=0.3$ (Table~\ref{tab:sim_exp_exp_rho_05_p0_03_p1_07_n_1000}), which results from the estimate $\hat{x}_{\ast,n}$ being $0$ more often (356 out of 500 replications instead of 32 when $p(0)=0.7$) while $x_\ast=1$, leading to wrong selection of
the subsample corresponding to $x=0$ instead of $x=1$ for the test decision.
For the effect of $\rho$, we note that it influences the subsample sizes and hence also
the performance of testing $H_{0x}$, with better level when the corresponding subsample size increases. This then
affects Method 1 through the individual test performance as explained in \eqref{eq:level_product}.
Its effect on Method 2 is more intricate, since it also influences $\hat{x}_{\ast,n}$ besides the individual test performance. 
For example, we consider the last case under Scenario B where $x_\ast=1$. In such case, Method 2 would have a rejection rate close to that for testing $H_{0x_\ast}$ if $\hat{x}_{\ast,n}$ estimated $x_\ast$ correctly most of the time. We observe from Table~\ref{tab:sim_exp_exp_rho_05_p0_07_p1_07_n_1000} ($\rho=0.5$) that Method 2 possesses a rejection rate (0.156) slightly larger than that of the individual test for $H_{0x_\ast}$ (0.150). In contrast, we see from Table~\ref{tab:sim_exp_exp_rho_03_p0_07_p1_07_n_1000} ($\rho=0.3$) that Method 2 has a rejection rate (0.244) larger than that for testing $H_{0x_\ast}$ (0.222). This mainly results from $\hat{x}_{\ast,n}$ being 1 less often when $\rho=0.3$ (in 455 out of 500 replications) as compared with the result when $\rho=0.5$ (in 468 out of 500 replications), which is probably caused by the underestimation of $f(\uptau_G(1)|1)$ when the subsample size corresponding to $x=1$ is small.

\begin{table}[t]
	\begin{center}
		\caption{Rejection rate for Setting~\ref*{enum:sim_exp_unif_2cov} with $n=2000$.
				`Scenario' groups cases as A (insufficient follow-up, $H_0$), B (borderline, $H_0$), and C (sufficient follow-up, $H_1$).
				Columns under $\uptau_G(x_1,x_2)$ are follow-up lengths in quantiles of $F_u(\cdot|x_1,x_2)$. `Rej. rate of $H_{0x}$' gives rejection rates for fixed $(x_1,x_2)$ using \cite{YM2024}. Columns M1/M2 report rejection rates of $H_0$ for the first/second method.  The significance level $\alpha$ is 0.05.	
			\label{tab:sim_res_exp_unif_2cov}}
		\footnotesize
		\begin{tabular}{ccccc@{\extracolsep{8pt}}cccc@{\extracolsep{8pt}}cc}
\multirow{2}{*}{Scenario} & \multicolumn{4}{c}{$\uptau_G(x_1,x_2)$} & \multicolumn{4}{c}{Rej. rate of $H_{0x}$} & \multicolumn{2}{c}{Rej. rate of $H_0$} \\ \cline{2-5}\cline{6-9}\cline{10-11}
& $(0,0)$  & $(0,1)$  & $(1,0)$ & $(1,1)$ & $(0,0)$    & $(0,1)$   & $(1,0)$   & $(1,1)$   & M1             & M2             \\ \hline
\multirow{7}{*}{A} & 0.950 & 0.950 & 0.950 & 0.950 & 0.006 & 0.016 & 0.004 & 0.012 & 0.000 & 0.018 \\ 
   & 0.950 & 0.975 & 0.950 & 0.950 & 0.006 & 0.016 & 0.004 & 0.018 & 0.000 & 0.016 \\ 
   & 0.975 & 0.975 & 0.950 & 0.950 & 0.006 & 0.014 & 0.006 & 0.014 & 0.000 & 0.018 \\ 
   & 0.975 & 0.975 & 0.975 & 0.950 & 0.006 & 0.014 & 0.008 & 0.016 & 0.000 & 0.022 \\ 
   & 0.975 & 0.975 & 0.975 & 0.975 & 0.006 & 0.014 & 0.008 & 0.008 & 0.000 & 0.012 \\ 
   & 0.975 & 0.990 & 0.975 & 0.975 & 0.006 & 0.178 & 0.008 & 0.008 & 0.000 & 0.060 \\ 
   & 0.999 & 0.999 & 0.950 & 0.999 & 0.694 & 0.824 & 0.008 & 0.934 & 0.002 & 0.418 \\ 
   \hline
\multirow{7}{*}{B} & 0.975 & 0.990 & 0.999 & 0.999 & 0.006 & 0.178 & 0.574 & 0.930 & 0.000 & 0.114 \\ 
   & 0.975 & 0.999 & 0.999 & 0.999 & 0.006 & 0.824 & 0.594 & 0.934 & 0.004 & 0.012 \\ 
   & 0.990 & 0.990 & 0.990 & 0.990 & 0.044 & 0.180 & 0.012 & 0.076 & 0.000 & 0.130 \\ 
   & 0.990 & 0.995 & 0.995 & 0.990 & 0.044 & 0.388 & 0.096 & 0.058 & 0.000 & 0.216 \\ 
   & 0.990 & 0.999 & 0.990 & 0.990 & 0.044 & 0.820 & 0.010 & 0.068 & 0.000 & 0.086 \\ 
   & 0.990 & 0.999 & 0.995 & 0.990 & 0.044 & 0.820 & 0.098 & 0.062 & 0.000 & 0.088 \\ 
   & 0.999 & 0.999 & 0.990 & 0.999 & 0.694 & 0.824 & 0.010 & 0.934 & 0.008 & 0.722 \\ 
   \hline
\multirow{5}{*}{C} & 0.995 & 0.995 & 0.995 & 0.995 & 0.168 & 0.386 & 0.098 & 0.236 & 0.000 & 0.340 \\ 
   & 0.995 & 0.999 & 0.995 & 0.995 & 0.168 & 0.830 & 0.092 & 0.240 & 0.002 & 0.270 \\ 
   & 0.999 & 0.999 & 0.995 & 0.995 & 0.694 & 0.824 & 0.098 & 0.252 & 0.008 & 0.730 \\ 
   & 0.999 & 0.999 & 0.999 & 0.995 & 0.694 & 0.824 & 0.578 & 0.238 & 0.074 & 0.766 \\ 
   & 0.999 & 0.999 & 0.999 & 0.999 & 0.694 & 0.824 & 0.578 & 0.938 & 0.320 & 0.790 \\ 
   \hline
\end{tabular}

	\end{center}
\end{table}
Next we consider Setting~\ref*{enum:sim_exp_unif_2cov}, in which two binary covariates are considered. Table~\ref{tab:sim_res_exp_unif_2cov} shows the rejection rate of $H_{0x}$ for
all values of $x$ and the rejection rate of $H_0$ for both the proposed methods and 
the selected cases of Setting~\ref*{enum:sim_exp_unif_2cov}.
Recall that $\family{X}=\lbrace(0,0),(0,1),(1,0),(1,1)\rbrace$ for this setting, and that each column under $\uptau_G(x_1,x_2)$ indicates the length of the follow-up in terms of the quantiles
of the corresponding $F_u(\cdot|x_1,x_2)$.
For example, a case with 0.95 in the column of $(0,0)$ under $\uptau_G(x_1,x_2)$ means $\uptau_G(0,0)$ is the 95\% quantile of $F_u(\cdot|0,0)$. 
For the columns under `Rej. rate of $H_{0x}$', each value is the rejection rate of $H_{0x}$ for the corresponding $x$ when the test procedure of \cite{YM2024} is applied to the subsample corresponding to $X=x$. We highlight that for the subsample corresponding to $(x_1,x_2)=(0,1)$, the rejection rate of $H_{0x}$ when $\uptau_G(0,1)=q_{0.99}(0,1)$ is around 0.18 which is higher than the nominal level of 0.05. This subpopulation has an expected sample size of $2000*0.24=480$ and an uncure fraction $p(0,1)$ of around 0.29. The small uncure fraction leads to a very high censoring rate and hence it is difficult to detect if the follow-up for this category is insufficient as indicated in \cite{YM2024}.

For the column M1/M2, each value is the rejection rate of $H_0$ for the first/second method. Both methods perform well in terms of empirical level for the cases under Scenario A, except for the last case for Method 2. For this case, $x_\ast=(1,0)$. The rejection rate of Method 2 would be close to that for testing $H_{0x_\ast}$ (0.008) if $\hat{x}_{\ast,n}=x_\ast$ in most of the experiments. The rejection rate of Method 2 differs from that of the individual test mostly resulting from $\hat{x}_{\ast,n}$ being $(1,0)$ less often (231/500 replications), which is probably due to the underestimation of $f(\uptau_{G}(x_\ast)|x_\ast)$ when the subsample size is relatively small ($\mathbb{P}(X_1=1, X_2=0)=0.16$).
Method 1 has better control on the level compared with Method 2	for the cases under Scenario B, while Method 2 possesses higher power than Method 1 for the cases under Scenario C. This suggests that Method 1 is more conservative which is coherent to the discussion in Section~\ref*{sec:proc} that a low power is expected when the powers of the individual tests for each $x\in\family{X}$ are not close to one, see for example the first few cases under Scenario C in Table~\ref*{tab:sim_res_exp_unif_2cov}.
We note that for the scenarios where the follow-up is sufficient or nearly sufficient in the conditional setting for each $x\in\family{X}$, such as the last case under Scenario B, Method 2 possesses a rejection rate far larger than the nominal level of 0.05.
For such case,
$x_\ast=(1,0).$
Method 2 would have a rejection rate close to that for testing $H_{0x_\ast}$ (0.01) if $\hat{x}_{\ast,n}$ estimated $x_\ast$ correctly in the majority of the replications. The discrepancy between the rejection rate of Method 2 and that of the individual test mainly results from $\hat{x}_{\ast,n}$ being $(1,0)$ less often (53/500 replications), which is probably caused by the underestimation of $f(\uptau_{G}(x_\ast)|x_\ast)$ when the subsample size is relatively small ($\mathbb{P}(X_1=1, X_2=0)=0.16$).

It remains intriguing to investigate the effect of sample size, $\Delta G(\uptau_G(x))$ and $p(x)$ on the test performance. For this we consider a simpler setting, Setting~\ref*{enum:sim_exp_unif}, in which the uncured subjects 
have the same conditional survival and censoring distribution for each $x$.
Complete results can be found in the Supplementary Materials. In summary, we observe a better control of level when sample size is larger, see for example Tables~S13--S14 (or Tables~S9--S12 for  Setting~\ref*{enum:sim_exp_exp}). The proposed methods still work well when the assumption of $\Delta G(\uptau_G(x))>0$ for all $x\in\family{X}$ is violated, especially when $n=1000$, (see for example Tables~S13 and S14). Both methods perform better in terms of level and power
when the uncure fraction $p(x)$ is larger (see Table~S16
or Tables~S11/S12 for Setting~\ref*{enum:sim_exp_exp}).

For Setting~\ref*{enum:sim_logistic_weibull_unif}, the uncured subjects given $X=x$ have a Weibull distribution with shape parameter 0.8, which has a density that decreases fast. As reported in \cite{YM2024}, 
the procedure for testing $H_{0x}$ under such scenarios 
does not perform very well
in terms of level control. 
Therefore we do not expect both Methods~1 and 2 to perform well
in terms of empirical level, especially for the cases under Scenario B for Setting~\ref*{enum:sim_logistic_weibull_unif} (see Table~S8).

\subsubsection*{Overall performance summary}%
In the simulation study, we grouped different cases into three main scenarios, namely Scenario~A, B and C. For the cases under Scenario~A, which are considered as insufficient follow-up, both methods control the level well. For the cases under Scenario~B, which represent borderline follow-up sufficiency, Method~1 maintains the nominal level, whereas Method~2 can reject more in some settings. For the cases under Scenario~C, which correspond to sufficient follow-up, Method~2 achieves substantially higher power, while Method~1 performs more conservatively in detecting follow-up sufficiency.%

The performance of each method can be influenced by the subsample size proportion $\rho$ and the uncure fraction $p(x)$. For the effect of $\rho$, it determines the subsample sizes and hence affects the test performance for $H_{0x}$, where the level is better maintained when the corresponding subsample size increases. This, in turns, influences Method~1 through the performance of individual tests. Its effect on Method~2 is more intricate, since it also affects $\hat{x}_{\ast,n}$ in addition to the individual test performances. When the subsample size corresponding to a particular covariate value is small, $\hat{x}_{\ast,n}$ may select the correct $x_\ast$ less often, leading to incorrect test decisions. This is probably caused by the underestimation of $f(\uptau_G(x)|x)$ under those scenarios.
Regarding the effect of $p(x)$, both methods perform better in terms of level and power when $p(x)$ is larger. However, when $p(x)$ is small for some $x$, leading to a very high censoring rate, the performance of both methods in maintaining the test level deteriorates, especially under the borderline cases.
The shape of the conditional density of uncured survival times, $f_u(\cdot|x)$, also influences the performance of both methods. For example, the procedure for testing $H_{0x}$, and consequently both methods for testing $H_0$, does not perform well in terms of level control when $f_u(\cdot|x)$ decreases fast.%

Overall we conclude that the first method performs well in terms of empirical level but has a lower power to detect sufficient follow-up. This is particularly problematic in scenarios with more than two possible covariate values and sufficient, but not extremely long follow-up. On the other hand, the second method possesses higher power, at the cost of also higher level in the borderline cases when follow-up is insufficient but very close to being sufficient. In practice, a lack of control for such scenarios is less problematic and the behaviour improves as the sample size or the uncure fraction increases.

\section{Real data application}
\label{sec:app}
Melanoma of the skin is one of the most common cancers in the United States. The 5-year survival rate for skin melanoma patients increased steadily from 1975 to 2019 \citep{SGJ2024}, which results from the improvement in melanoma treatments such as immunotherapy and targeted therapy. \cite{andersson_estimating_2014} studied the effect of age, cancer stage, gender, and anatomical site on the cure proportion of cutaneous malignant melanoma in Sweden using cure models. A study of Belgian patients using mixture cure models showed that skin melanoma has high cure proportions and that the mean uncured survival time is higher for females than for males \cite{silversmit_cure_2017}.%

We analyze a dataset of skin melanoma patients from the Surveillance, Epidemiology and End Results (SEER) database to illustrate the test procedures. The complete SEER database can be retrieved from the website \url{https://seer.cancer.gov/}. The data was extracted from the database `Incidence--SEER Research Data, 8 Registries, Nov 2023 Sub (1975--2021)' with follow-up until December 2021. We selected white patients with skin melanoma who were younger than 50 years at diagnosis, had regional or distant cancer stage, and were diagnosed between 2008 and 2021. This allows a maximum of 167 months (about 14 years) of follow-up. We further excluded the observations with zero or unknown follow-up time. The event time of interest is the time to death from skin melanoma. 
We included two binary covariates, namely gender (male/female) and cancer stage (regional/distant), into our analysis. Thus, $\family{X}$ consists of 4 different categories. This cohort has 2050 observations with follow-up ranging from 1 to 167 months and has a censoring rate of 76.49\%. Table~\ref{tab:seer_melanoma_data} reports the number of observations, maximum uncensored survival time, maximum survival time, censoring rate, and proportion of alive among censored observations for each subsample. We note that the proportions of patients alive among the censored observations for each subsample are at least 85\%, which suggests that at least 85\% of the censored patients did not die from other causes and thus were not subject to competing risks.
A similar skin melanoma dataset extracted from the SEER database was studied by \cite{TaiEtal2005}, which suggested the minimum follow-up time required to estimate the cure fraction of skin melanoma is 18.2 years, assuming a log-normal distribution for the uncured survival time.
\begin{table}[t]
	\begin{center}
		\caption{Sample size, proportion $\rho_x$, maximum uncensored survival time (months) $\tilde{Y}_{x(n_x)}$, maximum observed survival time (months) ${Y}_{x(n_x)}$, censoring rate, and proportion of patients alive among censored observations for each subsample of the melanoma data.\label{tab:seer_melanoma_data}}
		\small
		\begin{tabular}{lcccccc}
	Type & Sample size & $\rho_x$ & $\tilde{Y}_{x(n_x)}$ & ${Y}_{x(n_x)}$ & \makecell{Censoring\\ rate} & \makecell{Proportion of \\alive}  \\ 
	\hline
	all &  2050 & 1.0000 &   133 &   167 & 0.7649 & 0.9471 \\ 
	Distant and Female &   151 & 0.0737 &    89 &   166 & 0.5762 & 0.8966 \\ 
	Regional and Female &   746 & 0.3639 &   131 &   166 & 0.8579 & 0.9641 \\ 
	Distant and Male &   272 & 0.1327 &    98 &   167 & 0.5037 & 0.8613 \\ 
	Regional and Male &   881 & 0.4298 &   133 &   167 & 0.7991 & 0.9545 \\
	\hline
\end{tabular}
	\end{center}
\end{table}
\begin{figure}[ht]
	\centering
	\includegraphics[width=0.9\textwidth]{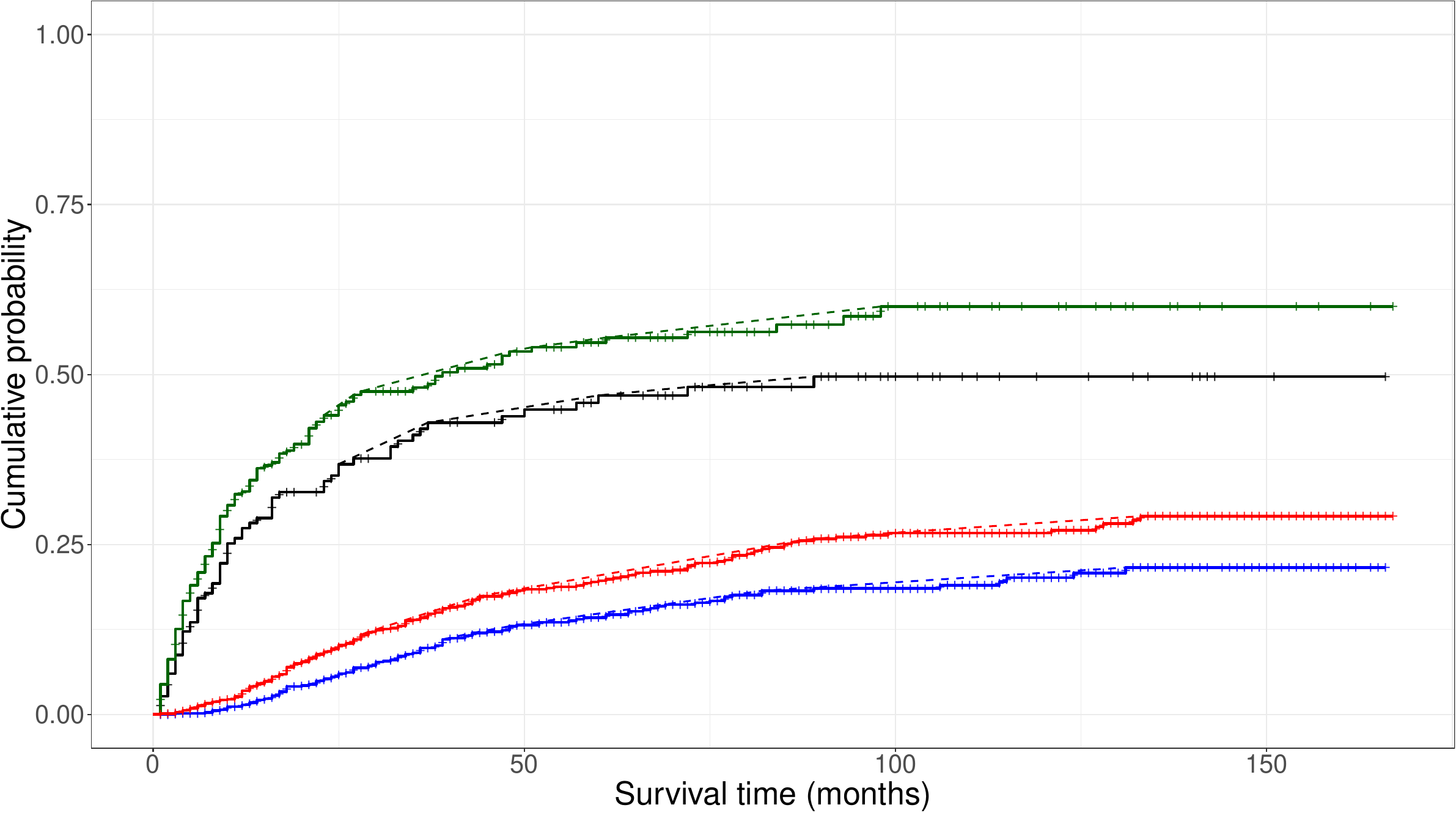}
	\caption{Kaplan--Meier curves (solid) and their least concave majorants (dashed) for each subsample of the melanoma data: Distant and Female (Black), Regional and Female (Blue), Distant and Male (Green), Regional and Male (Red).\label{fig:seer_melanoma_kme}}
\end{figure}%

Figure~\ref{fig:seer_melanoma_kme} depicts the Kaplan--Meier curves for each subsample. The curves for the subsamples corresponding to `Distant and Male' and `Distant and Female' categories exhibit a long plateau, while the remaining two curves do not. The Kaplan--Meier curves also support the assumption of decreasing conditional density in the tail for each subsample
(curves are close to concave and close to their least concave majorants), where we set $a_x=0$. We apply the proposed test procedures to this data for testing the null hypothesis of insufficient follow-up over all covariate values with $\epsilon=0.01$, indicating the follow-up is considered as sufficient over all covariate values if $q_{0.99}(x)<\uptau_G(x)$ for all $x\in\family{X}$. The distributions of $\hat{f}_{nb_x}(Y_{x(n_x)}|x)$ are approximated using the bootstrap procedure described in Section~\ref{sec:proc} with 1000 bootstrap samples and bandwidths as mentioned in Section~\ref{sec:sim_settings}.
From \cite{TaiEtal2005}, we infer from the fitted log-normal distribution of the uncured survival time for the skin melanoma patients that the 99.5\% quantile is approximately 373 months (around 31 years). Therefore, $\tau$ is set to 373 months meaning that the probability of death from skin melanoma is negligible after 31 years from diagnosis.
We also apply 
the $Q_n$ test \citep{MRS2023} using the parameter $\gamma=1$ to each subsample for testing whether the follow-up is insufficient in the conditional setting for each fixed $x\in\family{X}$.
We do not apply the $T_n$ test \citep{XEK2024} to each subsample because the Kaplan--Meier curves in Figure~\ref{fig:seer_melanoma_kme} do not exhibit long plateaus for two subgroups, suggesting that there is no strong evidence of sufficient follow-up for those subgroups.

Table~\ref{tab:seer_melanoma_result} reports the $p$-value of testing the null hypothesis of insufficient follow-up in the conditional setting for each fixed $x$. At the 5\% significance level, the proposed method $\hat{f}_{nb_x}$ does not reject $H_{0x}$ for categories `Regional and Female' and `Regional and Male'. Therefore Method 1 does not reject the hypothesis of insufficient follow-up over all covariate values $H_0$. For Method 2, $\hat{x}_{\ast,n}$ is the `Regional and Female' category, indicating that the test decision is determined based on the individual test corresponding to that category and hence $H_0$ is not rejected. Note that if a larger value of $\tau$ was chosen, the test would become more conservative meaning that it would still not reject the hypothesis of insufficient follow-up. The $Q_n$ test also does not reject the hypothesis of insufficient follow-up for the `Regional and Female' and `Regional and Male' categories, which is consistent with the results from the proposed method.

\begin{table}[t]
	\begin{center}
		\caption{$p$-values of testing $H_{0x}$ of the melanoma data
				using the $Q_n$ test \citep{MRS2023} and the proposed method.\label{tab:seer_melanoma_result}}
		\begin{tabular}{lcc}
	Type & $Q_n$ & $\hat{T}_{n}(x)$ \\ 
	\hline
	Distant and Female & 0.0004 & 0.0140 \\ 
	Distant and Male & 0.0075 & 0.0280 \\ 
	Regional and Female & 0.2373 & 0.2710 \\ 
	Regional and Male & 0.1780 & 0.3620 \\ 
	\hline
\end{tabular}

	\end{center}
\end{table}

\section{Discussion}\label{sec:discuss}
In this paper, we proposed new procedures for testing the null hypothesis of `practically' insufficient follow-up over all 
values of categorical covariates. The test relies on the assumption that the conditional density function of the survival times is non-increasing in the tail region and a smoothed Grenander estimator is considered to construct the test statistic. The challenge arises when combining the decisions of the individual tests for each covariate value. Method 1 follows the principle of intersection-union test and rejects the null hypothesis of overall insufficient follow-up if such hypothesis is rejected for all covariate values. On the other hand, Method 2 decides based on just one appropriately chosen covariate value for which sufficient follow-up is less likely. The asymptotic normality of the test statistics under the null hypothesis was established and the two proposed test procedures were shown to be asymptotically correct. Simulation studies showed that the procedures perform satisfactorily in finite samples, albeit the rejection rate can be higher than the nominal level under the borderline scenarios where the overall follow-up is nearly sufficient. 
As recommended in \cite{YM2024}, the choice of $\tau$ for the proposed method should be based on prior knowledge in the sense that the event happening after time $\tau$ is almost impossible (with a larger $\tau$ resulting in a more conservative test).
Given the challenging nature of the problem and the different handling of continuous and discrete covariates,  we focused only on accounting for categorical covariates. This is usually also the most relevant case in practical applications since often even continuous covariates are discretized. The study of the continuous covariate case remains interesting and will be considered in a future work.
A potential limitation of the proposed method is that the procedure may be less reliable in subsamples with small sample sizes. In our simulation study, the empirical level remained well controlled in a setting with an expected subsample size of 150. Nevertheless, developing methods that account for small subsample situations, for example by imposing semiparametric modeling assumptions, would be an interesting direction for future research.

\appendix
\section{Technical lemmas and proofs}
\label{sec:appendix}
\begin{lemma}[Lemma 1 of \cite{YM2024}]
	\label{lemma:strong_approx_cond}
	Suppose that Assumptions~\ref*{assumption:discrete_cov}--\ref*{assumption:censoring_discrete_cov} hold.
	Then, conditional on $\mat{X}_1^n=(X_1,\dots,X_n)$, we have, for $u>0$,
	\begin{align*}
		\begin{split}
			&\mathbb{P}\left[
			\sup_{t \leq \uptau_{G}(x)}{n_x}
			\left\vert
			\hat{F}_{n}(t|x)-F(t|x) - n_x^{-1/2}\cbracket{1-F(t|x)}W\circ L(t|x)
			\right\vert
			> K_{1}\log n_x + u
			\middle\vert\mat{X}_1^n
			\right]\\
			&\qquad< K_{2}e^{-K_{3}u},
		\end{split}
	\end{align*}
	where $K_{1}$, $K_{2}$ and $K_{3}$ are positive constants, $W$ is a Brownian motion, and 
	\[
	L(t|x) = \int_{0}^{t}
	\frac{\dd F(u|x)}{\rbracket{1-G(u{-}|x)}\rbracket{1-F(u|x)}^{2}}.
	\]
\end{lemma}

\begin{proof}[Proof of Theorem~\ref{thm:sg_discrete_normality}]
	Recall that $k_{B}(v)=\phi(0)k(v)-\psi(0)vk(v)$ and
	\[
	\hat{f}_{nb,Y_{x(n_x)}}(Y_{x(n_x)}|x):=\hat{f}_{nb}(Y_{x(n_x)}|x)=\frac{1}{b}\int_{Y_{x(n_x)}-b}^{Y_{x(n_x)}}k^B_{b,Y_{x(n_x)}}\left(
	\frac{Y_{x(n_x)}-u}{b}
	\right)\dd{\hat{F}_n^G(u|x)},
	\]
	where $b$ is the bandwidth and the boundary kernel $k^B_{b,t}$ is defined as in Section~\ref{sec:proc}. In particular, $k^B_{b,Y_{x(n_x)}}(v)=k_B(v)$, for $v\in[-1,1]$. Let 
	\[
	\hat{f}_{nb,\uptau_{G}(x)}(\uptau_{G}(x)|x):=\frac{1}{b}\int_{\uptau_{G}(x)-b}^{\uptau_{G}(x)}\tilde{k}_{b,\uptau_{G}(x)}^B\left(
	\frac{\uptau_{G}(x)-u}{b}
	\right)\dd{\hat{F}_n^G(u|x)},
	\]
	where the boundary kernel $\tilde{k}^B_{b,t}$ is defined as follows:
	\[
	\tilde{k}_{b,t}^B(v)=
	\begin{cases}
		\phi\rbracket{\frac{a_x-t}{b}}k(v) + \psi\rbracket{\frac{a_x-t}{b}}vk(v)&\quad t\in[a_x,b],\\
		k(v)&\quad t\in(b,\uptau_{G}(x)-b),\\
		\phi\rbracket{\frac{\uptau_{G}(x) - t}{b}}k(v) - \psi\rbracket{\frac{\uptau_{G}(x)-t}{b}}vk(v)&\quad t\in[\uptau_{G}(x)-b,\uptau_{G}(x)].
	\end{cases}
	\]
	So $\tilde{k}^B_{b,\tau_G(x)}(v)=k_B(v)$.
	We have that
	\begin{align}
		\label{eq:sg_normality_split}
		\begin{split}
			&
			\hat{f}_{nb,Y_{x(n_x)}}(Y_{x(n_x)}|x)-f(\uptau_G(x)|x)\\
			&=
			\left\lbrace
			\hat{f}_{nb,Y_{x(n_x)}}(Y_{x(n_x)}|x)-\hat{f}_{nb,\uptau_{G}(x)}(\uptau_G(x)|x)
			\right\rbrace
			+
			\left\lbrace
			\hat{f}_{nb,\uptau_{G}(x)}(\uptau_G(x)|x)-f(\uptau_G(x)|x)
			\right\rbrace
			\\
			&=(A)+(B).
		\end{split}
	\end{align}
	We will first show in Part 1 that $n^{2/5}(A)$ converges to zero in probability and then show separately for \ref{enum:sg_discrete_normality_cond} in Part 2 and \ref{enum:sg_discrete_normality_uncond} in Part 3 that $n^{2/5}(B)$ converges to the required normal distribution. Based on these results, we claim the convergence in \eqref{eq:sg_discrete_normality_uncond_undersmooth} for \ref{enum:sg_discrete_normality_uncond_undersmooth} in Part 4.
	
	\paragraph*{Part 1}
	For $n^{2/5}(A)$, since $Y_{x(n_x)}\leq\uptau_{G}(x)$ and $\hat{F}_n^G(u|x)=\hat{F}_n^G(Y_{x(n_x)}|x)$ for $u\geq Y_{x(n_x)}$, we have that
	\begin{align*}
		\begin{split}
			&\hat{f}_{nb,Y_{x(n_x)}}(Y_{x(n_x)}|x)-\hat{f}_{nb,\uptau_{G}(x)}(\uptau_G(z)|x)\\
			&=
			\frac{1}{b}\int_{Y_{x(n_x)}-b}^{\uptau_G(x)-b}k^B_{b,Y_{x(n_x)}}\left(
			\frac{Y_{x(n_x)}-u}{b}
			\right)\dd{\hat{F}_n^G(u|x)}\\
			&\quad+
			\frac{1}{b}\int_{\uptau_G(x)-b}^{\uptau_G(x)}\left\lbrace
			k^B_{b,Y_{x(n_x)}}\left(
			\frac{Y_{x(n_x)}-u}{b}
			\right)
			- \tilde{k}^B_{b,\uptau_G(x)}\left(
			\frac{\uptau_G(x)-u}{b}
			\right)
			\right\rbrace
			\dd{\hat{F}_n^G(u|x)}\\
			&=
			\frac{1}{b}\int_{Y_{x(n_x)}-b}^{\uptau_G(x)-b}k_{B}\left(
			\frac{Y_{x(n_x)}-u}{b}
			\right)\dd{\hat{F}_n^G(u|x)}\\
			&\quad+
			\frac{1}{b}\int_{\uptau_G(x)-b}^{\uptau_G(x)}\left\lbrace
			k_{B}\left(
			\frac{Y_{x(n_x)}-u}{b}
			\right)
			- k_{B}\left(
			\frac{\uptau_G(x)-u}{b}
			\right)
			\right\rbrace
			\dd{\hat{F}_n^G(u|x)}.
		\end{split}
	\end{align*}
	Since $\phi(0)$ and $\psi(0)$ are bounded, and $\sup_{v\in[-1,1]}|k^{(l)}(v)|<\infty$ for $l=0,1$, with $k^{(0)}(v)=k(v)$ and $k^{(1)}(v)=k'(v)$, it can be shown that each of the two terms on the right-hand side of the above display is $O_P\left(b^{-2}|Y_{z(n_z)}-\uptau_G(z)|\right)$. 
	Thus, 
	$n^{2/5}(A)$ converges to zero in probability by Lemma~\ref{lemma:max_surv_time_discrete_conv}, since
	$n^{2/5}(A)=O_P\left(n^{4/5}|Y_{z(n_z)}-\uptau_G(z)|\right)$. The convergence also holds when the bandwidth $b_x=c_xn_x^{-1/5}$ is used.
	
	\paragraph*{Part 2}
	We show the convergence of $n^{2/5}(B)$ for \ref{enum:sg_discrete_normality_cond}.
	We note that $n_x=\sum_{i=1}^{n}\indicator{X_i=x}$. Therefore, the convergence of the sequence $n_x^{2 / 5}\lbrace
	\hat{f}_{nb}(\uptau_{G}(x)|x) - f(\uptau_{G}(x)|x) 
	\rbrace$ to $N(\mu_x, \sigma_x^2)$, given $\mat{X}_1^n=(X_1,\dots,X_n)$, follows from Theorem 2 of \cite{YM2024}, which is based on Lemma~\ref{lemma:strong_approx_cond}. Since $n_x/n\to\rho_x$ a.s. by the strong law of large numbers, we have that for almost every $\mat{X}_1^n$, $n_x/n\to\rho_x$ as a deterministic sequence. Hence we have that
	\[
	\left.
	n^{2 / 5}\left\lbrace{
		\hat{f}_{nb}(\uptau_{G}(x)|x) - f(\uptau_{G}(x)|x) 
	}\right\rbrace\middle\vert\mat{X}_1^n
	\right.
	\converge[d]
	N(\rho_x^{-2/5}\mu_x, \rho_x^{-4/5}\sigma_x^2)
	\]
	for almost every $\mat{X}_1^n$.
	
	\paragraph*{Part 3}
	Next, we show the convergence of $n^{2/5}(B)$ for \ref{enum:sg_discrete_normality_uncond}. One approach is to use the unconditional version of the strong approximation in Lemma~\ref{lemma:strong_approx_cond}. However, since such a result is not available to our best knowledge, we instead show the convergence based on a strong representation of the conditional Kaplan--Meier estimator, together with results from empirical processes.
	We will divide the proof into two parts. For the first part (Part 3.1), 
	we will first give a proof when the bandwidth $b=cn^{-1/5}$
	is used, which is deterministic. Using this result, we will then show, in Part 3.2, the asymptotic normality of $n^{2/5}(B)$ when the bandwidth $b_x=c_xn_x^{-1/5}$ is used, which is stochastic.
	\paragraph*{Part 3.1}
	Following the proof of Theorem~2 of \cite{YM2024}, we have
	\begin{equation}
		\label{eq:sg_normality_split_disc_cov}
		\begin{split}
			\hat{f}_{nb}(\uptau_G(x)|x) - f(\uptau_G(x)|x)
			&=
			\int_{\uptau_G(x) - b}^{\uptau_G(x)}\frac{1}{b}k_{B}\rbracket{\frac{\uptau_G(x)-u}{b}}\dd{F(u|x)}-f(\uptau_G(x)|x)\\
			&\quad+
			\int_{\uptau_G(x) - b}^{\uptau_G(x)}\frac{1}{b}k_{B}\rbracket{\frac{\uptau_G(x)-u}{b}}\dd{(\hat{F}_{n}-F)(u|x)}\\
			&\quad+
			\int_{\uptau_G(x) - b}^{\uptau_G(x)}\frac{1}{b}k_{B}\rbracket{\frac{\uptau_G(x)-u}{b}}\dd{(\hat{F}_{n}^{G}-\hat{F}_{n})(u|x)}\\
			&=(I)+(II)+(III).
		\end{split}
	\end{equation}
	For $(I)$ in \eqref{eq:sg_normality_split_disc_cov}, since $f_u(\cdot|x)$ is twice continuously differentiable on $[a_x,\uptau_{F_{u}}(x)]$, which contains $[\uptau_G(x)-b, \uptau_G(x)]$, and by the properties of the boundary kernel, we have
	\[
	\begin{split}
		n^{2/5}(I)
		&=
		n^{2/5}\int_{0}^{1}k_{B}(v)\cbracket{-f^{\prime}(\uptau_G(x)|x)bv+\frac{1}{2}f^{\prime\prime}(\xi_{n}^x|x)b^{2}v^{2}}\dd{v}\\
		&\converge
		\frac{1}{2}c^{2}f^{\prime\prime}(\uptau_G(x)|x)\int_{0}^{1}v^{2}k_{B}(v)\dd{v},\quad\text{as}~n\converge\infty,
	\end{split}
	\]
	where $0 < \uptau_{G}(x)-\xi_{n}^x<bv<b\converge 0$ as $n\converge\infty$.
	For $(III)$ in \eqref{eq:sg_normality_split_disc_cov}, using Lemma~\ref{lemma:lcm_kme_kme_unif_dist_order} and the boundedness of $k^\prime$ and hence that of $k_{B}^\prime$, we can show that $n^{2/5}(III)$ converges to zero in probability. For $(II)$ in \eqref{eq:sg_normality_split_disc_cov}, we have that
	\[
	n^{2/5}(II)=\frac{1}{\sqrt{bn^{1/5}}}\int_0^1k_{B}(u)\dd{\hat{W}_n(u|x)},
	\]
	where, for $u\in[0,1]$,
	\[
	\hat{W}_n(u|x)=\sqrt{\frac{n}{b}}\left\lbrace
	\hat{F}_n(\uptau_{G}(x)-bu|x)-\hat{F}_n(\uptau_{G}(x)|x)-{F}(\uptau_{G}(x)-bu|x)+{F}(\uptau_{G}(x)|x)
	\right\rbrace.
	\]
	We will show that $\hat{W}_n(u|x)$ converges weakly to a Gaussian process in $\ell^\infty([0,1])$ and hence $n^{2/5}(II)$ converges to a mean zero normal distribution with variance as in \eqref{eq:sg_discrete_normality_uncond}. Using the strong representation in Lemma~\ref{lemma:strong_rep_brownian_disc_x}\ref{enum:strong_rep_disc_x}, we have
	\begin{align}
		\label{eq:W_hat_split}
		\begin{split}
			&\hat{W}_n(u|x)=\\
			&\sqrt{\frac{n}{b}}\left(1+\frac{\rho_x-n_x/n}{n_x/n}\right)\frac{1-F(\uptau_G(x)|x)}{n\rho_x}\sum_{i=1}^n\indicator{X_i=x}\left\lbrace\tilde{\xi}_i(\uptau_{G}(x)-bu|x)-\tilde{\xi}_i(\uptau_{G}(x)|x)\right\rbrace\\
			&+\sqrt{\frac{n}{b}}\left(1+\frac{\rho_x-n_x/n}{n_x/n}\right)\frac{F(\uptau_G(x)|x)-F(\uptau_G(x)-bu|x)}{n\rho_x(1-F(\uptau_G(x)-bu|x))}\sum_{i=1}^n\indicator{X_i=x}{\xi}_i(\uptau_{G}(x)-bu|x)\\
			&+\sqrt{\frac{n}{b}}\left\lbrace
			R_n(\uptau_G(x)-bu|x)-R_n(\uptau_G(x)|x)
			\right\rbrace,
		\end{split}
	\end{align}
	where $\tilde{\xi}_i(t|x)=\xi_i(t|x)/(1-F(t|x))$ with $\xi_i(t|x)$ defined in Lemma~\ref{lemma:strong_rep_brownian_disc_x}. For the last term on the right-hand side of \eqref{eq:W_hat_split}, by Lemma~\ref{lemma:strong_rep_brownian_disc_x}\ref{enum:strong_rep_disc_x}, we have that $\sqrt{nb^{-1}}\lbrace
	R_n(\uptau_G(x)-bu|x)-R_n(\uptau_G(x)|x)
	\rbrace$ converges to zero in probability. For the second term on the right-hand side of \eqref{eq:W_hat_split}, from the proof of Lemma~\ref{lemma:strong_rep_brownian_disc_x}\ref{enum:brownian_disc_x}, we have ${(n\rho_x)^{-1}}\sum_{i=1}^n\indicator{X_i=x}\xi_i(\uptau_{G}(x)-bu|x)=O_P\left(n^{-1/2}\right)$. By the continuity of $f_u(\cdot|x)$, we have $F(\uptau_G(x)|x)-F(\uptau_G(x)-bu|x)\leq 
	\overline{f}_{u,x}b$, where 
	$\overline{f}_{u,x}:=\sup_{t\leq\uptau_{G}(x)}f_u(t|x)<\infty$. Together with ${(\rho_x-n_x/n)}/{(n_x/n)}=O_P\left(n^{-1/2}\right)$ and $0\leq F(\uptau_G(x)-bu|x)< F(\uptau_G(x)|x)<1$, the second term on the right-hand side of \eqref{eq:W_hat_split} converges to zero in probability. For the first term on the right-hand side of \eqref{eq:W_hat_split}, we will use Theorem~2.11.23 of \cite{VW1996} to show that the following term:
	\[
	\sqrt{\frac{n}{b}}\frac{1-F(\uptau_G(x)|x)}{n\rho_x}\sum_{i=1}^n\indicator{X_i=x}\left\lbrace\tilde{\xi}_i(\uptau_{G}(x)-bu|x)-\tilde{\xi}_i(\uptau_{G}(x)|x)\right\rbrace
	\]
	converges weakly to a Gaussian process in $\ell^\infty([0,1])$ and hence is the dominating term, since $({\rho_x-n_x/n})/({n_x/n})=O_P\left(n^{-1/2}\right)$. We can rewrite the term in the above display as:
	\begin{align*}
		\frac{1-F(\uptau_{G}(x)|x)}{\rho_x\sqrt{b}}\int
		&\indicator{z=x}\bigg\lbrace
		-\frac{\delta{1}_{(\uptau_{G}(x)-bu,\uptau_{G}(x)]}(y)}{1-H(y{-}|x)}\\
		&+\int_{\uptau_G(x)-bu}^{\uptau_G(x)}\frac{{1}_{(-\infty,y]}(v)}{(1-H(v{-}|x))^2}\dd{H_1(v|x)}
		\bigg\rbrace
		\dd{\sqrt{n}(\mathbb{P}_n-\mathbb{P}_0)(y,\delta,z)},
	\end{align*}
	where $\mathbb{P}_n$ denotes the empirical distribution of $(Y_i,\Delta_i,X_i)$ and $\mathbb{P}_0$ denotes the distribution of $(Y,\Delta,X)$.
	Consider the functions
	\begin{align*}
		f_{n,u}(y,\delta,z)
		=
		\frac{1-F(\uptau_{G}(x)|x)}{\rho_x\sqrt{b}}
		\indicator{z=x}\bigg\lbrace
		&-\frac{\delta{1}_{(\uptau_{G}(x)-bu,\uptau_{G}(x)]}(y)}{1-H(y{-}|x)}\\
		&\quad+\int_{\uptau_G(x)-bu}^{\uptau_G(x)}\frac{{1}_{(-\infty,y]}(v)}{(1-H(v{-}|x))^2}\dd{H_1(v|x)}
		\bigg\rbrace,
	\end{align*}
	for $u\in[0,1]$. Let
	\begin{align*}
		F_n(y,\delta,z)=\frac{1-F(\uptau_{G}(x)|x)}{\rho_x\sqrt{b}}
		\indicator{z=x}\bigg\lbrace&
		\frac{\delta{1}_{(\uptau_{G}(x)-b,\uptau_{G}(x)]}(y)}{1-H(\uptau_G(x){-}|x)}\\
		&\quad
		+\frac{H_1((\uptau_G(x)-b,\uptau_G(x)]|x)}{(1-H(\uptau_G(x){-}|x))^2}
		\bigg\rbrace,
	\end{align*}
	which are envelope functions of $f_{n,u}$. Here for simplicity we use the notation $H_1((\uptau_G(x)-b,\uptau_G(x)]|x)=H_1(\uptau_G(x)|x)-H_1(\uptau_G(x)-b|x)$.
	We will verify the three required conditions in (2.11.21) of \cite{VW1996}.
	\begin{enumerate}[wide, labelwidth=!, labelindent=0pt]
		\item
		For $\int F_n^2(y,\delta,z)\dd{\mathbb{P}_0(y,\delta,z)}=O(1)$, we have
		\begin{equation}
			\label{eq:envelope_sq_disc_x_split}
			\begin{split}
				&\int F_n^2(y,\delta,z)\dd{\mathbb{P}_0(y,\delta,z)}
				=\frac{(1-F(\uptau_G(x)|x))^2}{b\rho_x^2}
				\bigg\lbrace
				\frac{\rho_xH_1((\uptau_{G}(x)-b,\uptau_{G}(x)]|x)}{(1-H(\uptau_G(x){-}|x))^2}\\
				&\qquad+
				\frac{\rho_xH_1^2((\uptau_G(x)-b,\uptau_G(x)]|x)}{(1-H(\uptau_G(x){-}|x))^4}
				+\frac{2\rho_xH_1^2((\uptau_{G}(x)-b,\uptau_{G}(x)]|x)}{(1-H(\uptau_G(x){-}|x))^3}
				\bigg\rbrace.
			\end{split}
		\end{equation}
		By the continuity of $f_u(\cdot|x)$, we have  
		\[
		H_1((\uptau_{G}(x)-b,\uptau_{G}(x)]|x)
		=\int_{\uptau_G(x) - b}^{\uptau_{G}}\left\lbrace1-G(s{-}|x)\right\rbrace\dd{F(s|x)}
		\leq
		\overline{f}_{u,x}b,
		\]
		where 
		$\overline{f}_{u,x}=\sup_{t\leq\uptau_{G}(x)}f_u(t|x)<\infty$. 
		Together with $1-H(\uptau_G(x){-}|x)>0$, the three terms on the right-hand side of \eqref{eq:envelope_sq_disc_x_split} are bounded and hence $\int F_n^2\dd{\mathbb{P}_0}=O(1)$.
		
		\item 
		Next, we show that $\int F_n^2\indicator{F_n>\eta\sqrt{n}}\dd{\mathbb{P}_0}\converge0$ for all $\eta>0$. Specifically, we have that
		\[
		\begin{split}
			&\int F_n^2\indicator{F_n>\eta\sqrt{n}}\dd{\mathbb{P}_0}\\
			&\leq\frac{(1-F(\uptau_G(x)|x))^2}{b\rho_x^2}
			\bigg\lbrace
			\frac{\rho_xH_1^2((\uptau_G(x)-b,\uptau_G(x)]|x)}{(1-H(\uptau_G(x){-}|x))^4}
			+\frac{2\rho_xH_1^2((\uptau_{G}(x)-b,\uptau_{G}(x)]|x)}{(1-H(\uptau_G(x){-}|x))^3}\\
			&\quad+\int
			\frac{\indicator{z=x}\delta{1}_{(\uptau_{G}(x)-b,\uptau_{G}(x)]}(y)\indicator{F_n>\eta\sqrt{n}}}{(1-H(\uptau_G(x){-}|x))^2}
			\dd{\mathbb{P}_0}
			\bigg\rbrace.
		\end{split}
		\]
		As shown previously that $H_1((\uptau_{G}(x)-b,\uptau_{G}(x)]|x)=O(b)$, the first and second terms on the right-hand side of the above inequality are $O(b)$, which converge to $0$ as $n\converge\infty$. For the last term on the right-hand side of the above inequality, using Cauchy--Schwarz inequality and the upper bound of the first term in \eqref{eq:envelope_sq_disc_x_split}, it suffices to show that ${b^{-1/2}}\int\indicator{F_n>\eta\sqrt{n}}\dd{\mathbb{P}_0}$ converges to zero. By the triangle inequality, we have
		\begin{align*}
			&\frac{1}{\sqrt{b}}\int\indicator{F_n>\eta\sqrt{n}}\dd{\mathbb{P}_0}\\
			&\leq
			\frac{1}{\sqrt{b}}\int
			{1}\left\lbrace{\frac{\delta\indicator{z=x}{1}_{(\uptau_{G}(x)-b,\uptau_{G}(x)]}(y)}{1-H(\uptau_G(x){-}|x)}>\frac{\eta\rho_x\sqrt{nb}}{2(1-F(\uptau_{G}(x)|x))}}\right\rbrace\\
			&\qquad\qquad+
			{1}\left\lbrace{\frac{\indicator{z=x}H_1((\uptau_{G}(x)-b,\uptau_{G}(x)]|x)}{(1-H(\uptau_G(x){-}|x))^2}>\frac{\eta\rho_x\sqrt{nb}}{2(1-F(\uptau_{G}(x)|x))}}\right\rbrace
			\dd{\mathbb{P}_0}.
		\end{align*}
		The first and second integrands on the right-hand side of the above inequality are zero for sufficiently large $n$, since $\sqrt{nb}\converge\infty$ as $n\converge\infty$. Therefore  $\int F_n^2\indicator{F_n>\eta\sqrt{n}}\dd{\mathbb{P}_0}\converge0$ for all $\eta>0$. 
		
		\item 
		We then show that, for every sequence $\delta_n\downarrow0$, $\sup_{|u_1-u_2|<\delta_n}\int(f_{n,u_1}-f_{n,u_2})^2\dd{\mathbb{P}_0}\converge0$. For any $u_1$, $u_2\in[0,1]$ with $|u_1-u_2|<\delta_n$, it can be shown that, for every $\delta_n\downarrow0$,
		\[
		\begin{split}
			\int&(f_{n,u_1}-f_{n,u_2})^2\dd{\mathbb{P}_0}
			\leq\frac{(1-F(\uptau_G(x)|x))^2}{b\rho_x^2}
			\bigg\lbrace
			\frac{b
				\overline{f}_{u,x}\delta_n\rho_x}{(1-H(\uptau_G(x){-}|x))^2}\\
			&\quad
			+\frac{b^2
				\overline{f}_{u,x}^2\delta_n^2\rho_x}{(1-H(\uptau_G(x){-}|x))^4}
			+\frac{2b^2
				\overline{f}_{u,x}\delta_n^2\rho_x}{(1-H(\uptau_G(x){-}|x))^3}
			\bigg\rbrace\converge0.
		\end{split}
		\]
	\end{enumerate}
	Next we verify the bracketing entropy condition. Since $\family{F}_n=\set{f_{n,u}}{u\in[0,1]}$ is a class of functions consisting of sums and products of monotone functions, by Theorem~2.7.5 of \cite{VW1996}, the bracketing entropy is bounded as follows:
	\[
	\log N_{[\,]}\left(\epsilon, \family{F}_n, \|\cdot\|_{\mathbb{P}_0,2}\right)\leq C/\epsilon,
	\]
	for some constant $C>0$, where $\|\cdot\|_{\mathbb{P}_0,2}$ is the $L_2$--norm corresponding to $\mathbb{P}_0$. Thus, we have, for every $\delta_n\downarrow0$,
	\[
	\int_0^{\delta_n}\sqrt{\log N_{[\,]}\left(\epsilon\|F_n\|_{\mathbb{P}_0,2}, \family{F}_n, L_2(\mathbb{P}_0)\right)}\dd{\epsilon}
	\leq C_1\sqrt{\delta_n}\converge0,
	\]
	where $C_1>0$ is a constant. Finally, we show that both the sequence of mean functions, $\mathbb{P}_0f_{n,s}$, and the sequence of the covariance functions, $\mathbb{P}_0f_{n,s}f_{n,t}-\mathbb{P}_0f_{n,s}\mathbb{P}_0f_{n,t}$, converge.
	We have, for  any $s\in[0,1]$,
	\[
	\left\vert\int f_{n,s}\dd{\mathbb{P}_0}\right\vert
	\leq
	\frac{
		\overline{f}_{u,x}(1-F(\uptau_G(x)|x))\sqrt{b}}{1-H(\uptau_G(x){-}|x)}+
	\frac{
		\overline{f}_{u,x}(1-F(\uptau_G(x)|x))\sqrt{b}}{(1-H(\uptau_G(x){-}|x))^2}\converge0.
	\]
	Furthermore, for any $s,t\in[0,1]$,
	\begin{align*}
		\int &f_{n,s}f_{n,t}\dd{\mathbb{P}_0}=
		\frac{(1-F(\uptau_{G}(x)|x))^2}{b\rho_x^2}
		\int
		\indicator{z=x}\bigg\lbrace
		\frac{\delta{1}_{(\uptau_{G}(x)-b(s\wedge t),\uptau_{G}(x)]}(y)}{(1-H(y{-}|x))^2}\\
		&+\int_{\uptau_G(x)-bs}^{\uptau_G(x)}\frac{{1}_{(-\infty,y]}(u)}{(1-H(u{-}|x))^2}\dd{H_1(u|x)}
		\int_{\uptau_G(x)-bt}^{\uptau_G(x)}\frac{{1}_{(-\infty,y]}(v)}{(1-H(v{-}|x))^2}\dd{H_1(v|x)}\\
		&-\frac{\delta{1}_{(\uptau_{G}(x)-bs,\uptau_{G}(x)]}(y)}{1-H(y{-}|x)}
		\int_{\uptau_G(x)-bt}^{\uptau_G(x)}\frac{{1}_{(-\infty,y]}(u)}{(1-H(u{-}|x))^2}\dd{H_1(u|x)}\\
		&-\frac{\delta{1}_{(\uptau_{G}(x)-bt,\uptau_{G}(x)]}(y)}{1-H(y{-}|x)}
		\int_{\uptau_G(x)-bs}^{\uptau_G(x)}\frac{{1}_{(-\infty,y]}(v)}{(1-H(v{-}|x))^2}\dd{H_1(v|x)}
		\bigg\rbrace\dd{\mathbb{P}_0}.
	\end{align*}
	It can be shown that the sum of the last three terms on the right-hand side of the above equation vanishes, while the first term equals
	\begin{equation}
		\label{eq:W_hat_cov_func_disc_x}
		\frac{(1-F(\uptau_{G}(x)|x))^2}{b\rho_x}\int_{\uptau_G(x)-b(s\wedge t)}^{\uptau_G(x)}\frac{\dd{H_1(u|x)}}{(1-H(u{-}|x))^2}
		\converge
		\frac{(s\wedge t)f(\uptau_G(x)|x)}{\rho_x(1-G(\uptau_G(x){-}|x))}.
	\end{equation}
	Therefore $\hat{W}_n(u|x)$ converges weakly in $\ell^\infty([0,1])$ to a mean zero Gaussian process with the covariance function given on the right-hand side of \eqref{eq:W_hat_cov_func_disc_x}, which has the same distribution as $$\sqrt{\frac{f(\uptau_G(x)|x)}{\rho_x(1-G(\uptau_G(x){-}|x))}}\hat{W}(u),$$ with a Brownian motion $\hat{W}(u)$, $u\in[0,1]$.
	Thus, we have
	\[
	\begin{split}
		n^{2/5}(II)
		&=\frac{1}{\sqrt{bn^{1/5}}}\int_0^1k_{B}(u)\dd{\hat{W}_n(u|x)}\\
		&\converge[d]
		\sqrt{\frac{f(\uptau_G(x)|x)}{c\rho_x(1-G(\uptau_G(x){-}|x))}}\int_0^1k_{B}(u)\dd{\hat{W}(u)}\\
		&\overset{d}{=}
		N\left(0, \frac{f(\uptau_G(x)|x)}{c\rho_x(1-G(\uptau_G(x){-}|x))}\int_0^1k_{B}^2(v)\dd{v}\right),
	\end{split}
	\]
	which completes the proof when the bandwidth $b=cn^{-1/5}$ is used.
	\paragraph*{Part 3.2}
	When the bandwidth $b_x=c_xn_x^{-1/5}$ for some constant $c_x>0$ is used, we note that $n_x/n\converge[a.s.]\rho_x\in(0,1)$ and so $b_x\converge[a.s.]0$ and $b_xn^{1/5}\converge[a.s.]c_x\rho_x^{-1/5}$. Therefore it can be shown that the convergence in \eqref{eq:sg_discrete_normality_uncond} using a similar argument as before. Specifically, for $(I)$ in \eqref{eq:sg_normality_split_disc_cov}, it can be shown that $n^{2/5}(I)\converge[a.s.]\frac{1}{2}\rho_x^{-2/5}c_x^{2}f^{\prime\prime}(\uptau_G(x)|x)\int_{0}^{1}v^{2}k_{B}(v)\dd{v}$, as $n\converge\infty$. 
	For $(II)$ in \eqref{eq:sg_normality_split_disc_cov}, we have
	\[
	n^{2/5}(II)=\int_0^1k_{B}(u)\dd{\tilde{W}_n(u|x)},
	\]
	where \[
	\tilde{W}_n(u|x)
	=\frac{n^{2/5}}{b_x}\left\lbrace
	\hat{F}_n(\uptau_{G}(x)-b_xu|x)-\hat{F}_n(\uptau_{G}(x)|x)-{F}(\uptau_{G}(x)-b_xu|x)+{F}(\uptau_{G}(x)|x)
	\right\rbrace.
	\] Similar to that in \eqref{eq:W_hat_split}, we can rewrite $\tilde{W}_n(u|x)$ as:
	\begin{align}
		\label{eq:W_tilde_split}
		\begin{split}
			\tilde{W}_n(u|x)
			&=
			\frac{n^{2/5}}{b_x}(1-F(\uptau_G(x)|x))\frac{1}{n_x}\sum_{i=1}^n\indicator{X_i=x}\left\lbrace\tilde{\xi}_i(\uptau_{G}(x)-b_xu|x)-\tilde{\xi}_i(\uptau_{G}(x)|x)\right\rbrace\\
			&\qquad+
			\frac{n^{2/5}}{b_x}\frac{F(\uptau_G(x)|x)-F(\uptau_G(x)-b_xu|x)}{1-F(\uptau_G(x)-b_xu|x)}\frac{1}{n_x}\sum_{i=1}^n\indicator{X_i=x}{\xi}_i(\uptau_{G}(x)-b_xu|x)\\
			&\qquad+
			\frac{n^{2/5}}{b_x}\left\lbrace
			R_n(\uptau_G(x)-b_xu|x)-R_n(\uptau_G(x)|x)
			\right\rbrace.
		\end{split}
	\end{align}
	The third term on the right-hand side of \eqref{eq:W_tilde_split} is $O_P(n^{3/5}(\log{n}/n)^{3/4})$ using Lemma~\ref{lemma:strong_rep_brownian_disc_x}\ref{enum:strong_rep_disc_x} and $n_x/n\converge[a.s.]\rho_x$.
	The second term  is $O_P(n^{-1/10})$, since $F(\uptau_G(x)|x)-F(\uptau_G(x)-b_xu|x)=O(b_x)$ together with the result of Lemma~\ref{lemma:strong_rep_brownian_disc_x}\ref{enum:brownian_disc_x}.
	The first term can be rewritten as:
	\begin{equation}
		\label{eq:first_term_W_tilde}
		\frac{n^{4/5}}{\sqrt{c_x}n_x^{4/5}}\sqrt{\frac{n}{b}}\frac{1-F(\uptau_G(x)|x)}{n}\sum_{i=1}^n\indicator{X_i=x}\left\lbrace\tilde{\xi}_i(\uptau_{G}(x)-\rho_{n}^{-1/5}bu|x)-\tilde{\xi}_i(\uptau_{G}(x)|x)\right\rbrace,
	\end{equation}
	where $b=c_xn^{-1/5}$ and $\rho_n=n_x/n$.
	We have shown in Part 1 that
	\[
	\sqrt{\frac{n}{b}}\frac{1-F(\uptau_G(x)|x)}{n\rho_x}\sum_{i=1}^n\indicator{X_i=x}\left\lbrace\tilde{\xi}_i(\uptau_{G}(x)-bu|x)-\tilde{\xi}_i(\uptau_{G}(x)|x)\right\rbrace
	\]
	converges weakly in $\ell^\infty([0,1])$ to a mean zero Gaussian process with the covariance function given on the right-hand side of \eqref{eq:W_hat_cov_func_disc_x}. By the continuous mapping theorem (see the lemma on page 151 of \cite{B1999}), since $\rho_n^{-1/5}\converge[a.s.]\rho_x^{-1/5}$ together with the weak convergence above, the term in \eqref{eq:first_term_W_tilde}
	converges weakly in $\ell^\infty([0,1])$ to a mean zero Gaussian process with the covariance function:
	\[
	\frac{(s\wedge t)f(\uptau_G(x)|x)}{c_x\rho_x^{4/5}(1-G(\uptau_G(x){-}|x))},\quad\text{for any }s,t\in[0,1].
	\]
	Therefore,
	\[
	n^{2/5}(II)\converge[d]
	N\left(0, \frac{f(\uptau_G(x)|x)}{c_x\rho_x^{4/5}(1-G(\uptau_G(x){-}|x))}\int_0^1k_{B}^2(v)\dd{v}\right),
	\]
	which completes the proof when $b_x=c_xn_x^{-1/5}$ is used.
	
	\paragraph*{Part 4}
	We then claim the convergences in \eqref{eq:sg_discrete_normality_uncond_undersmooth} for \ref{enum:sg_discrete_normality_uncond_undersmooth} and begin with $b=cn^{-\kappa}$. Referring to \eqref{eq:sg_normality_split} and Part 1, it can be shown that $n^{(1-\kappa)/2}(A)$ converges to zero in probability, since $n^{(1-\kappa)/2}(A)=O_P\left(n^{1/2+3\kappa/2}|Y_{z(n_z)}-\uptau_G(z)|\right)$ and $\kappa\in(1/5,1/3)$. Such convergence also holds when $b_x=c_xn_x^{-\kappa}$.
	For $n^{(1-\kappa)/2}(B)$, similar to that in Part 3.1, it can be shown that $n^{(1-\kappa)/2}(I)$ converges to zero, since $n^{(1-\kappa)/2}(I)=O_P\left(n^{(1-5\kappa)/2}\right)$ and $\kappa\in(1/5,1/3)$, and also that $n^{(1-\kappa)/2}(III)$ converges to zero in probability, since $n^{(1-\kappa)/2}(III)=O_P\left(n^{(1+\kappa)/2}(\log{n}/n)^{2/3}\right)$ and $\kappa\in(1/5,1/3)$.
	For $(II)$, it can be shown that
	\[
	n^{(1-\kappa)/2}(II)=\frac{1}{\sqrt{bn^\kappa}}\int_0^1k_{B}(u)\dd{\hat{W}_n(u|x)},
	\]
	where $\hat{W}_n$ is defined as in Part 3.1. Then the convergence of $n^{(1-\kappa)/2}(II)$ to the required normal distribution follows from a similar argument in Part 3.1. When $b_x=c_xn_x^{-\kappa}$ is used, 
	the asymptotic normality can be established by using a similar argument in Part 3.2. Specifically, $n^{(1-\kappa)/2}(I)\converge[a.s.]0$ as $n\to\infty$. For $n^{(1-\kappa)/2}(II)$, we note that \eqref{eq:first_term_W_tilde} becomes
	\[
	\frac{n^{1-\kappa}}{\sqrt{c_x}n_x^{1-\kappa}}\sqrt{\frac{n}{b}}\frac{1-F(\uptau_G(x)|x)}{n}\sum_{i=1}^n\indicator{X_i=x}\left\lbrace\tilde{\xi}_i(\uptau_{G}(x)-\rho_{n}^{-\kappa}bu|x)-\tilde{\xi}_i(\uptau_{G}(x)|x)\right\rbrace,
	\]
	where $b=c_xn^{-\kappa}$ and $\rho_n=n_x/n$. Thus, it can be shown that
	\[
	n^{(1-\kappa)/2}(II)\converge[d]
	N\left(0, \frac{f(\uptau_G(x)|x)}{c_x\rho_x^{1-\kappa}(1-G(\uptau_G(x){-}|x))}\int_0^1k_{B}^2(v)\dd{v}\right).
	\]
\end{proof}

\begin{lemma}
	\label{lemma:max_surv_time_discrete_conv}
	Suppose that Assumptions~\ref*{assumption:discrete_cov}--\ref*{assumption:positive_cure_prob} hold. Then,
	for each $x\in\family{X}$, we have that, for any $a\in(0,1)$,
	\[
	n^a\left(
	\uptau_G(x) - Y_{x(n_x)}
	\right)\converge[\mathbb{P}] 0
	\quad\text{as }n\converge\infty.
	\]
\end{lemma}
\begin{proof}[Proof of Lemma~\ref{lemma:max_surv_time_discrete_conv}]
	Let $\mat{X}_1^n=(X_1,\cdots,X_n)$. For all $\epsilon>0$, we have
	\begin{align*}
		\begin{split}
			\prob{\left\vert
				n^a\left(\uptau_G(x) - Y_{x(n_x)}\right)
				\right\vert>\epsilon}
			&=\prob{
				Y_{x(n_x)}<\uptau_G(x)-{\epsilon}/{n^a}
			}\\
			&=\expect[\sbracket]{
				\prob{Y_{x(n_x)}<\uptau_G(x)-{\epsilon}/{n^a}\middle\vert \mat{X}_1^n
			}}\\
			&=\sum_{k=0}^{n}\binom{n}{k}\rho_x^k(1-\rho_x)^{n-k}H^k\left(\uptau_G(x)-{\epsilon}/{n^a}\middle\vert x\right)\\
			&=\left(\rho_xH\left(\uptau_G(x)-{\epsilon}/{n^a}\middle\vert x\right)+1-\rho_x\right)^n\\
			&=\exp\left\lbrace n\log\left(
			1-\rho_x+\rho_xH(\uptau_G(x)-{\epsilon}/{n^a}\vert x)
			\right)\right\rbrace\\
			&\leq\exp\left\lbrace
			-n\rho_x(1-H(\uptau_G(x)-{\epsilon}/{n^a}\vert x))
			\right\rbrace\\
			&\leq\exp\left\lbrace
			-n\rho_x(1-F_u(\uptau_G(x)|x))(1-G(\uptau_{G}(x){-}|x))
			\right\rbrace\\
			&\converge0\quad\text{as }n\converge\infty.
		\end{split}
	\end{align*}
	Here the second last inequality holds since, under the assumptions that $F_u(\uptau_G(x)|x)<1$ and $1-G(\uptau_{G}(x){-}|x)>0$, we have $$1-H(\uptau_G(x)-{\epsilon}/{n^a}|x)\geq(1-F_u(\uptau_G(x)|x))
	(1-G(\uptau_{G}(x){-}|x))>0.$$
	Also note that $1-\rho_x+\rho_xH(\uptau_G(x)-{\epsilon}/{n^a}\vert x)\in(0,1)$ for all $n$.
\end{proof}

\begin{proof}[Proof of Theorem~\ref{theorem:test_discrete_level}]
	For \ref{enum:test_M1_level},
	under $H_0: q_{1-\epsilon}(z) \geq \uptau_{G}(z)$ for some $z\in\family{X}$, we have
	\begin{align}
		\label{eq:test_discrete_level}
		\begin{split}
			&\prob{\bigcap_{x\in\family{X}}\cbracket{
					\hat{f}_{nb}(Y_{x(n_x)}|x)-\frac{\epsilon \hat{F}_{n}(Y_{x(n_x)}|x)}{\tau-Y_{x(n_x)}} 
					< n_x^{-\frac{1-\kappa}2}\xi_{x,\alpha}
				}\middle\vert H_0}\\
			&\leq\prob{
				\hat{f}_{nb}(Y_{z(n_z)}|z)-\frac{\epsilon \hat{F}_{n}(Y_{z(n_z)}|z)}{\tau-Y_{z(n_z)}} 
				< n_z^{-\frac{1-\kappa}2}\xi_{z,\alpha}\middle\vert H_0
			}\\
			&\leq\prob{
				\hat{f}_{nb}(Y_{z(n_z)}|z)-f(\uptau_G(z)|z)-
				\frac{\epsilon\hat{F}_{n}(Y_{z(n_z)}|z)}{\tau-Y_{z(n_z)}} +
				\frac{(\epsilon-\eta_n) {F}(\uptau_G(z)|z)}{\tau-\uptau_G(z)}
				< n_z^{-\frac{1-\kappa}2}\xi_{z,\alpha}\middle\vert H_0
			}.
		\end{split}
	\end{align}
	We further expand the stochastic terms in \eqref{eq:test_discrete_level} as follows and will show that $(I)$ below is the dominating term:
	{
		\allowdisplaybreaks
		\begin{align*}
			&n^{(1-\kappa)/2}\left\lbrace
			\hat{f}_{nb,Y_{z(n_z)}}(Y_{z(n_z)}|z)-f(\uptau_G(z)|z)-
			\frac{\epsilon \hat{F}_{n}(Y_{z(n_z)}|z)}{\tau-Y_{z(n_z)}} +
			\frac{(\epsilon-\eta_n) {F}(\uptau_G(z)|z)}{\tau-\uptau_G(z)}
			\right\rbrace\\
			&=
			n^{(1-\kappa)/2}\left\lbrace
			\hat{f}_{nb,Y_{z(n_z)}}(Y_{z(n_z)}|z)-f(\uptau_G(z)|z)
			\right\rbrace\\
			&\quad-
			\frac{\epsilon n^{(1-\kappa)/2}\lbrace\hat{F}_{n}(Y_{z(n_z)}|z)-F(Y_{z(n_z)}|z)\rbrace}{\tau-Y_{z(n_z)}}-
			\frac{\epsilon n^{(1-\kappa)/2}\lbrace F(Y_{z(n_z)}|z)-F(\uptau_G(z)|z) \rbrace}{\tau-Y_{z(n_z)}}\\
			&\quad-
			\epsilon n^{(1-\kappa)/2}F(\uptau_G(z)|z)\left(
			\frac{1}{\tau-Y_{z(n_z)}}-\frac{1}{\tau-\uptau_{G}(z)}
			\right)-
			\frac{n^{(1-\kappa)/2}\eta_nF(\uptau_G(z)|z)}{\tau-\uptau_G(z)}
			\\
			&=(I)-(II)-(III)-(IV)-(V).
		\end{align*}
	}
	\begin{itemize}[wide,labelwidth=!, labelindent=0pt]
		\item
		For $(II)$, using Lemma~\ref{lemma:strong_rep_brownian_disc_x}\ref{enum:brownian_disc_x}
		and the fact that $\kappa\in(1/5,1/3)$, we have
		\[
		n^{(1-\kappa)/2}\lbrace\hat{F}_{n}(Y_{z(n_z)}|z)-F(Y_{z(n_z)}|z)\rbrace \leq
		n^{(1-\kappa)/2}\sup_{t\in[0,\uptau_{G}(z)]}\vert
		\hat{F}_{n}(t|z)-F(t|z)
		\vert=o_P(1).
		\]
		Also $1/(\tau-Y_{z(n_z)})=O_P(1)$, since $Y_{z(n_z)}\converge \uptau_{G}(z)$ almost surely. Thus, $(II)$ converges to zero in probability.
		\item
		For $(III)$, since $f_u(\cdot|z)$ is continuously differentiable, we have
		\[
		n^{(1-\kappa)/2}\left\vert
		F(Y_{z(n_z)}|z)-F(\uptau_G(z)|z)
		\right\vert\leq
		n^{(1-\kappa)/2}p(z)\sup_t\vert f_u(t|z)\vert\left\vert
		\uptau_G(z)-Y_{z(n_z)}
		\right\vert,
		\]
		which converges to zero in probability by 
		Lemma~\ref{lemma:max_surv_time_discrete_conv} and the assumption that $\sup_t\vert f_u(t|z)\vert<\infty$. Also $1/(\tau-Y_{z(n_z)})=O_P(1)$. Therefore, $(III)$ converges to zero in probability.
		\item
		For $(IV)$, its convergence to zero in probability follows from $1/(\tau-Y_{z(n_z)})=O_P(1)$, the boundedness of $\tau-\uptau_{G}(z)$ and Lemma~\ref{lemma:max_surv_time_discrete_conv}.
		\item
		For $(V)$, its convergence to zero follows from the assumption that $n^{(1-\kappa)/2}\eta_n\converge0$ as $n\converge\infty$.
	\end{itemize}
	The dominant term $(I)$ converges to a normal distribution by Theorem~\ref{thm:sg_discrete_normality} and the $\alpha$-quantile of which is $\xi_{\alpha}$, which implies that the probability in \eqref{eq:test_discrete_level} is asymptotically bounded by $\alpha$.
	
	For \ref{enum:test_M2_level}, 
		recall that $Q_{1-\gamma,n}^\star(x)$ is the $(1-\gamma)$-quantile of the bootstrap estimate $\hat{T}_n^\star(x)$, for $\gamma\in(0,1)$, and $\hat{x}_{\ast,n}=\argmax_{x\in\family{X}}Q_{1-\gamma,n}^\star(x)$.
		We wish to show that the test given in \eqref{eq:test_M2} has an asymptotic level of $\alpha$.
		Let $\family{X}_\ast=\set{x\in\family{X}}{T(x)=\max_{x\in\family{X}}T(x)}$.
		For notational simplicity, we denote $\mathbb{P}_{H_0}(A):=\mathbb{P}(A \mid H_0)$ for any event $A$.
		Under $H_0: q_{1-\epsilon}(z) \geq \uptau_{G}(z)$ for some $z\in\family{X}$, we have
		\begin{align}
			\label{eq:rej_prob_M2_decomp}
			\begin{split}
				&\mathbb{P}_{H_0}\big(\hat{T}_n(\hat{x}_{\ast,n})< {n^{-(1-\kappa)/2}}\xi_{\hat{x}_{\ast,n},\alpha}\big)\\
				&\leq
				\mathbb{P}_{H_0}\big(\hat{T}_n(\hat{x}_{\ast,n})< {n^{-(1-\kappa)/2}}\xi_{\hat{x}_{\ast,n},\alpha},\ \hat{x}_{\ast,n}\in\family{X}_\ast\big)
				+
				\mathbb{P}_{H_0}\big(\hat{x}_{\ast,n}\notin\family{X}_\ast\big)\\
				&=
				\sum_{x\in\family{X}_\ast}
				\mathbb{P}_{H_0}\big(\hat{T}_n(x)< {n^{-(1-\kappa)/2}}\xi_{x,\alpha}\big| \hat{x}_{\ast,n}=x\big)
				\mathbb{P}_{H_0}\big(\hat{x}_{\ast,n}=x\big)
				+
				\mathbb{P}_{H_0}\big(\hat{x}_{\ast,n}\notin\family{X}_\ast\big)\\
				&=
				\mathbb{P}_{H_0}\big(\hat{T}_n(\tilde{x}_\ast)< {n^{-(1-\kappa)/2}}\xi_{\tilde{x}_\ast,\alpha}\big| \hat{x}_{\ast,n}=\tilde{x}_\ast\big)\mathbb{P}_{H_0}\big(\hat{x}_{\ast,n}=\tilde{x}_\ast\big)\\
				&\qquad+
				\sum_{x\in\family{X}_\ast\setminus\{\tilde{x}_\ast\}}\mathbb{P}_{H_0}\big(\hat{T}_n(x)< {n^{-(1-\kappa)/2}}\xi_{x,\alpha}\big| \hat{x}_{\ast,n}=x\big)\mathbb{P}_{H_0}\big(\hat{x}_{\ast,n}=x\big)
				+
				\mathbb{P}_{H_0}\big(\hat{x}_{\ast,n}\notin\family{X}_\ast\big)\\
				&\leq\mathbb{P}_{H_0}\big(\hat{T}_n(\tilde{x}_\ast)< {n^{-(1-\kappa)/2}}\xi_{\tilde{x}_\ast,\alpha}\big)\\
				&\qquad+
				\sum_{x\in\family{X}_\ast\setminus\{\tilde{x}_\ast\}}\mathbb{P}_{H_0}\big(\hat{T}_n(x)< {n^{-(1-\kappa)/2}}\xi_{x,\alpha}\big| \hat{x}_{\ast,n}=x\big)\mathbb{P}_{H_0}\big(\hat{x}_{\ast,n}=x\big)
				+\mathbb{P}_{H_0}\big(\hat{x}_{\ast,n}\notin\family{X}_\ast\big),
			\end{split}
		\end{align}
		for some $\tilde{x}_\ast\in\family{X}_\ast$ such that $\mathbb{P}_{H_0}\big(\hat{x}_{\ast,n}=\tilde{x}_\ast\big)\to1$ and $\mathbb{P}_{H_0}\big(\hat{x}_{\ast,n}=x\big)\to0$ for $x\in\family{X}_\ast\setminus\{\tilde{x}_\ast\}$.
		For the existence of such $\tilde{x}_\ast$, we note that $Q_{1-\gamma,n}^\star(x)$ is asymptotically approximated by $T(x)+r_nR(x)$ for some sequence $r_n\to0$ and some function $R$ that depends on the covariate $x$. 
		Specifically, we have that
		\[
		Q_{1-\gamma,n}^\star(x)=T(x) + Q_{1-\gamma}\left(\hat{T}_n^\star(x) - \tilde{T}_n(x) + \tilde{T}_n(x)-T(x)\right),
		\]
		where $Q_\alpha(Z)$ denotes the $\alpha$-quantile of the law of $Z$. 
		For the second term on the right-hand side of the equality in the above display, we have, conditional on the data, that
		\[
		Q_{1-\gamma}\left(\hat{T}_n^\star(x) - \tilde{T}_n(x) + \tilde{T}_n(x)-T(x)\right)
		=Q_{1-\gamma}\left(\hat{T}_n^\star(x) - \tilde{T}_n(x)\right) + \tilde{T}_n(x)-T(x).
		\]

		For $Q_{1-\gamma}\left(\hat{T}_n^\star(x) - \tilde{T}_n(x)\right)$,
		since the bootstrap is consistent, we have
		\[
		Q_{1-\gamma}\left(n^{(1-\kappa)/2}\lbrace\hat{T}_n^\star(x) - \tilde{T}_n(x)\rbrace\right)\converge[\mathbb{P}]Q_{1-\gamma}\left(N(0, \rho_x^{-(1-\kappa)}\sigma_x^2)\right)
		=\rho_x^{-(1-\kappa)/2}\sigma_xz_{1-\gamma},
		\]
		where $z_\alpha$ is the $\alpha$-quantile of the standard normal distribution.
		Therefore, 
		\[Q_{1-\gamma}\left(\hat{T}_n^\star(x) - \tilde{T}_n(x)\right)\]
		is asymptotically approximated by $n^{-(1-\kappa)/2}\rho_x^{-(1-\kappa)/2}\sigma_xz_{1-\gamma}$.
		For $\tilde{T}_n(x) - T(x)$, it can be shown that it is dominated by $$\hat{f}_{nb_0^x}(\uptau_{G}(x)|x)-f(\uptau_{G}(x)|x),$$ which is of order $O_P\left((b_0^x)^2f^{\prime\prime}(\uptau_{G}(x)|x)\right)$.			
		Thus, $Q_{1-\gamma}\left(\hat{T}_n^\star(x) - \tilde{T}_n(x) + \tilde{T}_n(x) - T(x)\right)$ is asymptotically approximated by $r_nR(x)$ for a sequence $r_n\to0$ and a function $R$ that depends on $x$.
		Hence, $\tilde{x}_\ast$ is the maximizer of $R(x)$ on $\family{X}_\ast$ asymptotically.
		We assume that we are not in the setting where the covariate $X$ has no effect on the density functions of both the survival and censoring times (affecting both $\sigma_x$ and $f^{\prime\prime}(\uptau_{G}(x)|x)$) and all the values of $X$ are equally likely (i.e., $\rho_x$ is the same for all $x$). In such case, $R(x)$ would be the same for all $x$.%
	
	Since, for each $x_\ast\in\family{X}_\ast$, $T(x_\ast)\geq\max_{x\in\family{X}}T(x)\geq0$ under $H_0$ and $\family{X}_\ast\subseteq\family{X}$ is finite, the first term on the right-hand side of the last inequality in \eqref{eq:rej_prob_M2_decomp} converges to $\alpha$ as $n\to\infty$ by a similar argument used in \ref{enum:test_M1_level}. 
	The second term converges to zero by the definition of $\tilde{x}_\ast$.%
	Therefore, it remains to show that the last term converges to zero.
	To show the convergence, we note that
	\begin{align*}
		\prob{\hat{x}_{\ast,n}\notin \family{X}_\ast\bigg| H_0}
		&=\prob{T(\hat{x}_{\ast,n})<\max_{y\in\family{X}}T(y)\bigg| H_0}\\
		&=\prob{\max_{y\in\family{X}}T(y)-T(\hat{x}_{\ast,n})\geq
			\min_{z\notin\family{X}_\ast}\big\vert\max_{x\in\family{X}}T(x)-T(z)\big\vert
			\bigg| H_0}.
	\end{align*}
	Then, it suffices to show that for all $\epsilon>0$,
	\[
	\mathbb{P}_{H_0}\left(
	\max_{y\in\family{X}}T(y)-T(\hat{x}_{\ast,n})\geq\epsilon
	\right)
	:=\prob{\max_{y\in\family{X}}T(y)-T(\hat{x}_{\ast,n})\geq\epsilon\bigg| H_0}
	\to0\quad\text{as }n\to\infty.
	\]
	We proceed with the following inequalities:
	\begin{align*}
		&\mathbb{P}_{H_0}\left(
		\max_{y\in\family{X}}T(y)-T(\hat{x}_{\ast,n})\geq\epsilon
		\right)\\
		&\leq
		\mathbb{P}_{H_0}\left(
		\big\vert
		\max_{y\in\family{X}}T(y)-Q_{1-\gamma,n}^\star(\hat{x}_{\ast,n})
		\big\vert
		\geq\epsilon/2
		\right)
		+
		\mathbb{P}_{H_0}\left(
		\big\vert
		Q_{1-\gamma,n}^\star(\hat{x}_{\ast,n})-T(\hat{x}_{\ast,n})
		\big\vert
		\geq\epsilon/2
		\right)\\
		&\leq
		\mathbb{P}_{H_0}\left(
		\big\vert
		\max_{y\in\family{X}}T(y)-\max_{x\in\family{X}}Q_{1-\gamma,n}^\star(x)
		\big\vert
		\geq\epsilon/2
		\right)
		+
		\mathbb{P}_{H_0}\left(
		\max_{x\in\family{X}}
		\big\vert
		Q_{1-\gamma,n}^\star(x)-T(x)
		\big\vert
		\geq\epsilon/2
		\right),
	\end{align*}
	where the last inequality above follows from the definition of $\hat{x}_{\ast,n}$.
	Since the maximum function is continuous and $\family{X}$ is finite, it is sufficient to show that, 
	for each $x\in\family{X}$, for all $\epsilon>0$:
	\[
	\mathbb{P}\left(
	\big\vert
	Q_{1-\gamma,n}^\star(x)-T(x)
	\big\vert
	>\epsilon/2
	\right)
	\to0\quad\text{as }n\to\infty.
	\]
	Recall that $\hat{f}_{nb_0^x}(\cdot|x)$ is the 
	smoothed Grenander estimator
	used to generate the bootstrap sample.
	Denote by
	\[
	\tilde{T}_n(x)=\hat{f}_{nb_0^x}(Y_{x(n_x)}|x)-\frac{\epsilon\hat{F}_{n}(Y_{x(n_x)}|x)}{\tau-Y_{x(n_x)}},
	\]
	we have that
	\[
	\mathbb{P}\left(
	\big\vert
	Q_{1-\gamma,n}^\star(x)-T(x)
	\big\vert
	>\epsilon/2
	\right)
	\leq
	\mathbb{P}\left(
	\big\vert
	Q_{1-\gamma,n}^\star(x)-\tilde{T}_n(x)
	\big\vert
	>\epsilon/4
	\right)
	+
	\mathbb{P}\left(
	\big\vert
	\tilde{T}_n(x)-T(x)
	\big\vert
	>\epsilon/4
	\right).
	\]
	Here, the second term on the right-hand side of the above inequality converges to zero by the consistency of $\tilde{T}_n(x)$, while the first term converges to zero by the bootstrap consistency.
	Specifically, let $Q_n^{x,\star}(\alpha)$ be the $\alpha$-quantile of $n^{(1-\kappa)/2}\lbrace\hat{T}_n^\star(x)-\tilde{T}_n(x)\rbrace$. Denote the $\alpha$-quantile of the normal distribution $N(\mu, \sigma^2)$ by $\Phi_{\mu,\sigma}^{-1}(\alpha)$. We have, by Theorem~\ref{thm:sg_discrete_normality} and the bootstrap consistency, that
	\[
	Q_n^{x,\star}(\alpha)\converge[\mathbb{P}]\Phi_{\mu,\sigma}^{-1}(\alpha)\in(-\infty,\infty),
	\]
	where $\mu=0$ and $\sigma^2=\rho_x^{-(1-\kappa)}\sigma_x^2$.
	Thus $Q_{1-\gamma,n}^\star(x)-\tilde{T}_n(x)=n^{-(1-\kappa)/2}Q_n^{x,\star}(1-\gamma)\converge[\mathbb{P}]0$ as $n\to\infty$.
	Hence, the proof is completed.
\end{proof}

\begin{prop}
	\label{prop:DKW_ineq_disc_x}
	Suppose that Assumptions~\ref*{assumption:discrete_cov}--\ref*{assumption:censoring_discrete_cov} hold.
	Then, for all $\epsilon>0$, we have for all $x\in\family{X}$
	\begin{equation}
		\label{eq:DKW_ineq_disc_x}
		\prob{\sup_{t\leq\uptau_G(x)}\left\vert
			\hat{H}_n(t\vert x) - H(t\vert x)
			\right\vert > \epsilon
		} \leq 
		2\exp
		\left(
		-\frac{2n\underline{\rho}\epsilon^2}{1+2\epsilon^2}
		\right),
	\end{equation}
	where $\underline{\rho}=\min_{x\in\family{X}}\rho_x$. The result also holds for the estimator $\hat{H}_{k,n}(t\vert x)$ of $H_k(t\vert x)$, $k=0,1$.
\end{prop}
\begin{proof}[Proof of Proposition~\ref{prop:DKW_ineq_disc_x}]
	Recall that $\hat{H}_n(t\vert x)=\sum_{i=1}^{n}\indicator{Y_i\leq t}\indicator{X_i = x}/n_x$. 
	Let $\mat{X}_1^n=(X_1,\\\cdots,X_n)$. 
	Using the DKW inequality \citep{DKW1956} for the subsample $\lbrace{(Y_{xi},\Delta_{xi}), i=1,\dots,n_x}\rbrace$, conditionally on $\mat{X}_1^n$, 
	we have for
	$\epsilon>0$
	\[
	\prob{\sup_{t\leq\uptau_G(x)}\left\vert
		\hat{H}_n(t\vert x) - H(t\vert x)
		\right\vert > \epsilon
		\middle\vert\mat{X}_1^n
	}
	\leq
	2\exp\left(
	-2n_x\epsilon^2
	\right).
	\]
	Since $n_x\sim B(n,\rho_x)$, we have
	\[
	\begin{split}
		\prob{\sup_{t\leq\uptau_G(x)}\left\vert
			\hat{H}_n(t\vert x) - H(t\vert x)
			\right\vert > \epsilon
		}
		&\leq
		\expect[\sbracket]{2\exp\left(
			-2n_x\epsilon^2
			\right)}\\
		&=
		2\left(
		1-\rho_x+\rho_x e^{-2\epsilon^2}
		\right)^n\\
		&\leq
		2\exp\lbrace
		-n\rho_x(1-e^{-2\epsilon^2})
		\rbrace\\
		&\leq
		2\exp\left(
		-\frac{2n\underline{\rho}\epsilon^2}{1+2\epsilon^2}
		\right).\\
	\end{split}
	\]
	Here, the last inequality holds since ${x}/{(x+1)} < 1-e^{-x}$ for $x>-1$ and $\underline{\rho}\leq\rho_x$. We can use the same argument the estimator $\hat{H}_{k,n}(t\vert x)$ of $H_k(t\vert x)$, $k=0,1$.
\end{proof}

\begin{lemma}
	\label{lemma:rate_unif_strong_consistency_disc_x}
	Suppose that Assumptions~\ref*{assumption:discrete_cov}--\ref*{assumption:censoring_discrete_cov} hold.
	Then, as $n\converge\infty$,
	\begin{equation}
		\label{eq:rate_unif_strong_consistency_disc_x}
		\max_{x\in\family{X}}\sup_{t\leq\uptau_G(x)}\left\vert
		\hat{H}_n(t\vert x) - H(t\vert x)
		\right\vert = O\left(\sqrt{\log{n}/n}\right),\quad\text{a.s.}
	\end{equation}
	The result also holds for the estimator $\hat{H}_{k,n}(t\vert x)$ of $H_k(t\vert x)$, $k=0,1$.
\end{lemma}
\begin{proof}[Proof of Lemma~\ref{lemma:rate_unif_strong_consistency_disc_x}]
	We apply the Borel--Cantelli lemma to show \eqref{eq:rate_unif_strong_consistency_disc_x}. Therefore, it suffices to show that, for some constant $K_1>0$,
	\begin{equation}
		\label{eq:rate_unif_strong_consistency_disc_x_series}
		\sum_{n=1}^\infty
		\prob{\max_{x\in\family{X}}\sup_{t\leq\uptau_G(x)}\left\vert
			\hat{H}_n(t\vert x) - H(t\vert x)
			\right\vert > K_1\sqrt{\log{n}/n}
		}<\infty.
	\end{equation}
	Using Proposition~\ref{prop:DKW_ineq_disc_x} and set $\epsilon=K_1\sqrt{\log{n}/n}$, the summand of the above series is bound above by
	\[
	a_n:=%
	2\vert\family{X}\vert
	\exp\left(
	-\frac{2\underline{\rho}K_1^2\log{n}}{1+2K_1^2m_n}
	\right),
	\]
	where $m_n=\log{n}/n$. Since $2K_1^2m_n\converge 0$ as $n\converge\infty$, there exists an $N_0\in\naturalnumber$ such that $1+2K_1^2m_n < 2$ whenever $n>N_0$. Therefore, we have, for $n>N_0$
	\[
	\begin{split}
		a_n
		&\leq
		2\vert\family{X}\vert
		\exp\left(-\underline{\rho}K_1^2\log{n}\right)\\
		&\leq
		2\vert\family{X}\vert\exp\left(-\underline{\rho}K_1^2\log{n}\right)
		=2\vert\family{X}\vert n^{-\underline{\rho}K_1^2}.
	\end{split}
	\]
	If we set $K_1>0$ such that $\underline{\rho}K_1^2\geq2$ (or equivalently $K_1\geq\sqrt{2/\underline{\rho}}$), then
	\[
	a_n\leq2\vert\family{X}\vert n^{-\underline{\rho}K_1^2}\leq2\vert\family{X}\vert n^{-2}\quad\text{for~}n>N_0.
	\]
	Thus the series in \eqref{eq:rate_unif_strong_consistency_disc_x_series} is convergent, since $\sum_{n=1}^\infty n^{-2}<\infty$.
	By using the same argument we can show that \eqref{eq:rate_unif_strong_consistency_disc_x} also holds for the estimator $\hat{H}_{k,n}(t\vert x)$ of $H_k(t\vert x)$, $k=0,1$.
\end{proof}

From Assumptions~\ref*{assumption:uncured_survival_discrete_cov}--\ref*{assumption:censoring_discrete_cov}, we assume that there is a jump point for the conditional censoring distribution $G(t|x)$ at $\uptau_G(x)$ for all $x\in\family{X}$. We define the continuous part of $H(t|x)$, 
denoted by $H_c(t|x)$, as its left limit at the jump point $\uptau_G(x)$ and $H(t|x)$ otherwise, and
its nonparametric estimator $\hat{H}_{c,n}(t|x)=\sum_{i=1}^{n_x}\indicator{Y_{xi}\leq t, Y_{xi}\neq \uptau_G(x)}/{n_x}$. 
Therefore, under Assumptions~\ref*{assumption:uncured_survival_discrete_cov}--\ref*{assumption:censoring_discrete_cov}, we have, for each $x\in\family{X}$, $|H(t_1|x)-H(t_2|x)|\leq C_x|t_1-t_2|$, for all $t_1$, $t_2\in(0, \uptau_G(x))$, for some finite constant $C_x>0$. Then we define $L_1(t)=C_xt$, which is continuous and non-decreasing.
Also, we define similarly for the continuous part of $H_{k,c}(t|x)$
and
$\hat{H}_{k,c,n}(t|x)=\sum_{i=1}^{n_x}\indicator{Y_{xi}\leq t, \Delta_i=k,Y_{xi}\neq \uptau_G(x)}/{n_x}$ for $k=0,1$.
With the above defined continuous parts, the next lemma is established, which will be applied for showing Lemma~\ref{lemma:strong_rep_remainder_disc_x} eventually.
\begin{lemma}
	\label{lemma:modulus_of_continuity_disc_x}
	Suppose that Assumptions~\ref*{assumption:discrete_cov}--\ref*{assumption:censoring_discrete_cov} hold.
	Let $a_n$ be a sequence of positive numbers satisfying $a_n\converge0$ as $n\converge\infty$ and $a_n(n/\log{n})>\epsilon>0$, for all $n$ sufficiently large. 
	Then as $n\converge\infty$,
	\begin{equation}
		\label{eq:modulus_of_continuity_disc_x}
		\begin{split}
			\max_{x\in\family{X}}
			&\sup_{s,t\leq\uptau_G(x)}
			\sup_{\vert L_1(t)-L_1(s)\vert\leq a_n}\left\vert
			\hat{H}_{c,n}(t|x)-H_c(t|x)-\hat{H}_{c,n}(s|x)+H_c(s|x)
			\right\vert\\%
			&=O\left(\sqrt{a_n\log{n}/n}\right)\quad\text{a.s.},
		\end{split}
	\end{equation}
	where $L_1$ is the function defined above.
	The result also holds for the estimator $\hat{H}_{k,c,n}(t\vert x)$ of $H_{k,c}(t\vert x)$, $k=0,1$.
\end{lemma}
\begin{proof}[Proof of Lemma~\ref{lemma:modulus_of_continuity_disc_x}]
	While the proof is similar to that in \cite{DA2002} with a different estimator $\hat{H}_{c,n}(t|x)$, we provide a complete proof here for the sake of completeness.
	For each $x\in\family{X}$, we first divide $[0, \uptau_G(x)]$ into $m_n^x=\left[{L_1(\uptau_G(x))}/{a_n}\right]$ sub-intervals by $0=t_0^x<t_1^x<\cdots<t_{m_n}^x=\uptau_G(x)$, such that $L_1(t_{i+1}^x)-L_1(t_i^x)=\bar{a}_n^x:=L_1(\uptau_G(x))/m_n^x$. Define $I_{ni}^x:=\left[t_{i-1}^x,t_{i+1}^x\right]$ for each $i=1,\cdots,m_n^x-1$. 
	We further partition each sub-interval $I_{ni}^x$ into $[t_{ij}^x,t_{i,j+1}^x]$ for $j=-b_n^x,\cdots,b_n^x$, where $b_n^x=O(\sqrt{a_n n/\log{n}})$,
	such that $L_1(t_{i,j+1}^x)-L_1(t_{ij}^x)=\bar{a}_n^x/b_n^x$. Using the property of the integer part function, we can show that $a_n\leq\bar{a}_n^x\leq2a_n$ for large $n$. Therefore, for any $s,t\leq\uptau_G(x)$ satisfying $\vert L_1(t)-L_1(s)\vert\leq a_{n}$, there exist an interval $I_{ni}^x$ such that $s,t\in I_{ni}^x$. Since $\hat{H}_{c,n}(t|x)$ and $H_{c,n}(t|x)$ are non-decreasing, we have that
	\[
	\begin{split}
		\max_{x\in\family{X}}&\sup_{s,t\leq\uptau_G(x)}\sup_{\vert L_1(t)-L_1(s)\vert\leq a_n}\left\vert
		\hat{H}_{c,n}(t|x)-H_c(t|x)-\hat{H}_{c,n}(s|x)+H_c(s|x)
		\right\vert\\
		&\leq
		\max_{x\in\family{X}}\max_{1\leq i\leq m_n^x-1}\,\max_{-b_n^x\leq j,k\leq b_n^x}\left\vert
		\hat{H}_{c,n}(t_{ik}^x|x)-H_c(t_{ik}^x|x)-\hat{H}_{c,n}(t_{ij}^x|x)+H_c(t_{ij}^x|x)
		\right\vert\\
		&\quad
		+2\max_{x\in\family{X}}\max_{1\leq i\leq m_n^x-1}\,\max_{-b_n^x\leq j\leq b_n^x}\left\vert
		{H}_{c}(t_{i,j+1}^x|x)-{H}_{c}(t_{ij}^x|x)
		\right\vert\\
		&=(I)+(II).
	\end{split}
	\]
	For $(II)$, we have for each $x\in\family{X}$,  $|H_c(y_1|x)-H_c(y_2|x)|\leq|L_1(y_1)-L_1(y_2)|$ for all $y_1,y_2\leq\uptau_G(x)$ and for some finite constant $C_x$. Hence, we have
	\[
	(II)\leq2\max_{x\in\family{X}}\bar{a}_n^x/b_n^x=O\left(\sqrt{a_n\log{n}/n}\right).
	\]
	For $(I)$, we will apply the Borel--Cantelli lemma to show that $(I)=O(\sqrt{a_n\log{n}/n})$ almost surely. We first apply Bernstein's inequality to bound the tail probability of $(I)$. Note that
	
	\[
	\begin{split}
		&\hat{H}_{c,n}(t_{ik}^x|x)-H_c(t_{ik}^x|x)-\hat{H}_{c,n}(t_{ij}^x|x)+H_c(t_{ij}^x|x)\\
		&=\frac{1}{n_x}\sum_{l=1}^{n}\indicator{X_l=x}\left(
		\indicator{Y_l\leq t_{ik}^x}-\indicator{Y_l\leq t_{ij}^x}
		-H_c(t_{ik}^x|x)+H_c(t_{ij}^x|x)
		\right)=\sum_{l=1}^nV_{ijkl}^x,
	\end{split}
	\]
	where $V_{ijkl}^x:=\indicator{X_l=x}(
	\indicator{Y_l\leq t_{ik}^x}-\indicator{Y_l\leq t_{ij}^x}
	-H_c(t_{ik}^x|x)+H_c(t_{ij}^x|x))/{n_x}$ and $n_x=\sum_{l=1}^{n}\indicator{X_l=x}\sim B(n, \rho_x)$. Assuming $n_x>0$ and conditioning on $\mat{X}_1^n$, we have $$|V_{ijkl}^x|\leq\frac{1}{n_x}\indicator{X_l=x}(1+\sup_{s,t\leq\uptau_G(x)}\vert L_1(s)-L_1(t)\vert)\leq\frac{1}{n_x}(1+a_n)$$ for all $l$, $\mathbb{E}({V_{ijkl}^x|\mat{X}_1^n})=0$, and
	\[
	\begin{split}
		\sum_{l=1}^n\text{Var}\left(V_{ijkl}^x|\mat{X}_1^n\right)
		&=\sum_{l=1}^n\frac{\indicator{X_l=x}}{n_x^2}\left\lbrace
		H_c(t_{ij}^x|x)(1-H_c(t_{ij}^x|x)) + H_c(t_{ik}^x|x)(1-H_c(t_{ik}^x|x))\right.\\
		&\qquad\qquad\left.
		-2(H_c(t_{ij}^x\wedge t_{ik}^x|x)-H_c(t_{ij}^x|x)H_c(t_{ik}^x|x))\right\rbrace\\
		&\leq
		2\sup_{s,t\leq\uptau_G(x)}\vert L_1(s)-L_1(t)\vert/{n_x}\leq2a_n/n_x.
	\end{split}
	\]
	Using Bernstein's inequality, we have for $\lambda>0$,
	\[
	\prob{
		\left\vert
		\sum_{l=1}^{n}V_{ijkl}^x
		\right\vert>\lambda \middle\vert\mat{X}_1^n}
	\leq
	2\exp\left(
	-\frac{3\lambda^2n_x}{2\lambda(1+a_n)+12a_n}
	\right).
	\]
	Since the previous upper bound also holds when $n_x=0$, we have
	\[
	\begin{split}
		\prob{
			\left\vert
			\sum_{l=1}^{n}V_{ijkl}^x
			\right\vert>\lambda}
		&\leq
		\expect[\sbracket]{
			2\exp\left(
			-\frac{3\lambda^2n_x}{2\lambda(1+a_n)+12a_n}
			\right)
		}\\
		&=2\left(
		1-\rho_x+\rho_x e^{-3\lambda^2/(2\lambda(1+a_n)+12a_n)}
		\right)^n\\
		&\leq
		2\exp\left\lbrace
		-n\rho_x(1-e^{-3\lambda^2/(2\lambda(1+a_n)+12a_n)})
		\right\rbrace\\
		&\leq
		2\exp\left\lbrace
		-3n\underline{\rho}\lambda^2/(2\lambda(1+a_n)+12a_n+3\lambda^2)
		\right\rbrace,
	\end{split}
	\]
	where the last inequality holds since ${x}/{(x+1)} < 1-e^{-x}$ for $x>-1$ and $\underline{\rho}\leq\rho_x$.
	Let $\lambda_n=C\sqrt{a_n\log{n}/n}$ for some positive constant $C$. We have
	\[
	\begin{split}
		&\prob{
			\max_{x\in\family{X}}\max_{1\leq i\leq m_n^x-1}\,\max_{-b_n^x\leq j,k\leq b_n^x}\left\vert
			\hat{H}_{c,n}(t_{ik}^x|x)-H_c(t_{ik}^x|x)-\hat{H}_{c,n}(t_{ij}^x|x)+H_c(t_{ij}^x|x)
			\right\vert > \lambda_n}\\
		&\leq
		\sum_{x\in\family{X}}\sum_{i=1}^{m_n^x-1}\sum_{-b_n^x\leq j,k\leq b_n^x}
		\prob{
			\left\vert
			\sum_{l=1}^{n}V_{ijkl}^x
			\right\vert>\lambda_n}\\
		&\leq
		2|\family{X}|(\overline{m}_n-1)(2\overline{b}_n+1)^2\exp\left\lbrace
		-3n\underline{\rho}\lambda_n^2/(2\lambda_n(1+a_n)+12a_n+3\lambda_n^2)
		\right\rbrace
	\end{split}
	\]
	where $\overline{m}_n=\max_{x\in\family{X}}m_n^x$ and $\overline{b}_n=\max_{x\in\family{X}}b_n^x$.
	We complete the proof by using the Borel--Cantelli lemma, which suffices to show that
	\begin{equation}
		\label{eq:modulus_of_continuity_disc_x_series}
		\sum_{n=1}^\infty			2|\family{X}|(\overline{m}_n-1)(2\overline{b}_n+1)^2\exp\left\lbrace
		-3n\underline{\rho}\lambda_n^2/(2\lambda_n(1+a_n)+12a_n+3\lambda_n^2)
		\right\rbrace<\infty,
	\end{equation}
	for some $C>0$. We note that for $n$ sufficiently large,
	\[
	\frac{3n\underline{\rho}\lambda_n^2}{2\lambda_n(1+a_n)+12a_n+3\lambda_n^2}
	=\frac{3\underline{\rho}C^2\log{n}}{12+2C\sqrt{\log{n}/{na_n}}(1+a_n)+3C^2\log{n}/n}
	\geq
	\frac{3\underline{\rho}C^2\log{n}}{K_1},
	\]
	where $K_1$ is a positive constant. Also, we have that $(\overline{m}_n-1)(2\overline{b}_n+1)^2=O\left(n/\log{n}\right)$. Therefore, we can show that by choosing a suitable constant $C>0$ we have
	\[
	2|\family{X}|(\overline{m}_n-1)(2\overline{b}_n+1)^2\exp\left\lbrace
	-3n\underline{\rho}\lambda_n^2/(2\lambda_n(1+a_n)+12a_n+3\lambda_n^2)
	\right\rbrace
	\leq
	\frac{K_2}{n^2\log{n}},
	\]
	where $K_2$ is a positive constant. Since $\sum_{n=2}^\infty{1}/{(n^2\log{n})}$ is convergent, the series in \eqref{eq:modulus_of_continuity_disc_x_series} is also convergent.
	
	By using the same argument we can show that \eqref{eq:modulus_of_continuity_disc_x} also holds for the estimator $\hat{H}_{k,c,n}(t\vert x)$ of $H_{k,c}(t\vert x)$, $k=0,1$.
\end{proof}

\begin{lemma}
	\label{lemma:strong_rep_remainder_disc_x}
	Suppose that Assumptions~\ref*{assumption:discrete_cov}--\ref*{assumption:censoring_discrete_cov} hold.
	Then as $n\converge\infty$,
	\begin{equation}
		\label{eq:strong_rep_remainder_disc_x}
		\max_{x\in\family{X}}\sup_{t\leq\uptau_G(x)}\left\vert
		\int_{0}^{t}
		\frac{\hat{H}_n(s{-}|x)-H(s{-}|x)}{(1-H(s{-}|x))^2}
		\dd{(\hat{H}_{1,n}-H_1)(s|x)}
		\right\vert
		=O\left(
		\frac{\log{n}}{n}
		\right)^{\frac{3}{4}}\quad\text{a.s.}
	\end{equation}
\end{lemma}
\begin{proof}[Proof of Lemma~\ref{lemma:strong_rep_remainder_disc_x}]
	To prove \eqref{eq:strong_rep_remainder_disc_x}, we follow the proof of Proposition~4.1 in \cite{DA2002} by using our Lemmas~\ref{lemma:rate_unif_strong_consistency_disc_x} and \ref{lemma:modulus_of_continuity_disc_x} instead of their Lemmas~4.2 and 4.3, respectively. 
	Also, for each $x\in\family{X}$, we first divide the interval $[0,\uptau_G(x)]$ into $k_n^{x}$ sub-intervals $(t_i^{x}, t_{i+1}^{x}]$, $i=0,\dots,k_n^{x}$ such that $L_1(t_{i+1}^x)-L_1(t_i^x)\leq\sqrt{\log{n}/n}$. We then consider a finer partition with partition points $t_0<t_1<\cdots<t_{k_n}$, which are the ordered points in $\lbrace{t_0^x,\dots,t_{k_n^x}^x}:{x\in\family{X}}\rbrace$. Therefore, $L(t_j)-L(t_{j-1})\leq\sqrt{\log{n}/n}$ and $H(t_j{-}|x)-H(t_{j-1}|x)\leq\sqrt{\log{n}/n}$ for all $j$ and $x\in\family{X}$, and $k_n=O(|\family{X}|\sqrt{n/\log{n}})$.
\end{proof}

\begin{lemma}
	\label{lemma:strong_rep_brownian_disc_x}
	Suppose that Assumptions~\ref*{assumption:discrete_cov}--\ref*{assumption:censoring_discrete_cov} hold.
	Define the stochastic process
	\[
	\xi_i(t|x)=\left(1-F(t|x)\right)\left\lbrace
	\frac{\indicator{Y_i\leq t, \Delta_{i}=1}}{1-H(Y_i{-}|x)}-\int_{0}^t\frac{\indicator{Y_i\geq s}}{(1-H(s{-}|x))^2}\dd{H_1(s|x)}
	\right\rbrace.
	\]
	Then, we have
	\begin{enumerate}[label={(\alph*)}]
		\item\label{enum:strong_rep_disc_x}
		\begin{equation}
			\label{eq:strong_rep_disc_x}
			\hat{F}_n(t|x)-F(t|x)=\frac{1}{n_x}\sum_{i=1}^n\indicator{X_i=x}\xi_i(t|x)+R_n(t|x),
		\end{equation}
		where $\max_{x\in\family{X}}\sup_{t\leq\uptau_G(x)}\left\vert R_n(t|x)\right\vert=O(
		\log{n}/{n}
		)^{\frac{3}{4}}$ a.s.;
		\item\label{enum:brownian_disc_x}
		for each $x\in\family{X}$, $\sqrt{n}\lbrace\hat{F}_n(t|x)-F(t|x)\rbrace$ converges weakly  in $D[0,\uptau_G(x)]$ to a mean zero Gaussian process $W(t|x)$ with covariance function, for $0\leq s,t \leq\uptau_G(x)$
		\[
		\expect[\sbracket]{
			W(s|x)W(t|x)
		}=\frac{(1-F(s|x))(1-F(t|x))}{\rho_x}\int_{0}^{s\wedge t}\frac{\dd{F}(u)}{(1-H(u{-}|x))(1-F(u|x))};
		\]
		\item\label{enum:iid_sum_gaussian_approx}
		for each $x\in\family{X}$,
		\[
		\sup_{t \leq \uptau_{G}(x)}\left\vert
		\frac{1}{n_x}\sum_{i=1}^n\indicator{X_i=x}\xi_i(t|x)-n^{-1/2}{W}(t|x)
		\right\vert
		=O\left(\frac{\log{n}}{n}\right)\quad\text{a.s.},
		\]
		where ${W}(t|x)$ is the mean zero Gaussian process defined in \ref{enum:brownian_disc_x}.
	\end{enumerate}
	
\end{lemma}
\begin{proof}[Proof of Lemma~\ref{lemma:strong_rep_brownian_disc_x}]
	For \ref{enum:strong_rep_disc_x}, we follow the proof of Theorem~3.2 in \cite{DA2002} by using Lemma~\ref{lemma:strong_rep_remainder_disc_x} instead of their Proposition~4.1 and using our Lemma~\ref{lemma:rate_unif_strong_consistency_disc_x} instead of their Lemma~4.2.
	
	For \ref{enum:brownian_disc_x}, from \eqref{eq:strong_rep_disc_x} we have that
	\begin{equation}
		\label{eq:strong_rep_disc_x_decomp}
		\begin{split}
			\frac{1}{n_x}\sum_{i=1}^n\indicator{X_i=x}\xi_i(t|x)
			&=
			\left(1+\frac{\rho_x-n_x/n}{n_x/n}\right)\frac{1}{n\rho_x}\sum_{i=1}^n\indicator{X_i=x}\xi_i(t|x)\\
			&=\left(1+\frac{\rho_x-n_x/n}{n_x/n}\right)(I),
		\end{split}
	\end{equation}
	Denote by 
	\[A_{1n}(t|x)=\sum_{i=1}^n\indicator{X_i=x}\indicator{Y_i\leq t, \Delta_i=1}/n-\rho_xH_1(t|x)\]
	and
	\[A_{2n}(t|x)=\sum_{i=1}^n\indicator{X_i=x}\indicator{Y_i\geq t}/n-\rho_x(1-H(t{-}|x)),\]
	then using integration by parts, it can be shown that
	\[
	\begin{split}
		\sqrt{n}(I)=\frac{1-F(t|x)}{\rho_x}
		\biggl\lbrace
		&\frac{\sqrt{n}A_{1n}(t|x)}{1-H(t{-}|x)}-
		\int_{0}^t\frac{\sqrt{n}A_{1n}(s|x)}{(1-H(s{-}|x))^2}\dd{H(s{-}|x)}\\
		&\quad
		-\int_{0}^t\frac{\sqrt{n}A_{2n}(s|x)}{(1-H(s{-}|x))^2}\dd{H_1(s|x)}
		\biggl\rbrace.
	\end{split}
	\]
	Following the proof of Theorem~3 in \cite{BC1974}, it can be shown that for each $x\in\family{X}$, $\sqrt{n}(A_{1n}(t|x), \\A_{2n}(t|x))$ converges weakly in $D[0,\uptau_G(x)]\times D[0,\uptau_G(x)]$ to a mean zero bivariate Gaussian process $(W_1(t|x), W_2(t|x))$ with covariance function, for $0\leq s,t \leq\uptau_G(x)$
	\[
	\begin{split}
		&\expect[\sbracket]{W_1(s|x)W_1(t|x)}=\rho_xH_1([0, s\wedge t]|x)-\rho_x^2H_1([0, s]|x)H_1([0, t]|x),\\
		&\expect[\sbracket]{W_2(s|x)W_2(t|x)}=\rho_xH([s\vee t, \infty)|x)-\rho_x^2H([s, \infty)|x)H([t,\infty)|x),\\
		&\expect[\sbracket]{W_1(s|x)W_2(t|x)}=\rho_xH_1([t, s]|x)-\rho_x^2H_1([0,s]|x)H([t,\infty)|x),\\
	\end{split}
	\]
	where we define $H_1(\emptyset|x)=0$ for the last term. Then following the construction of Theorems~4 and 5 in \cite{BC1974}, it can be shown that $\sqrt{n}(I)$ converges weakly in $D[0,\uptau_G(x)]$ to a mean zero Gaussian process $W(t|x)$ with covariance function, for $0\leq s,t \leq\uptau_G(x)$
	\[
	\expect[\sbracket]{W(s|x)W(t|x)}=
	\frac{(1-F(s|x))(1-F(t|x))}{\rho_x}\int_{0}^{s\wedge t}\frac{\dd{F(u|x)}}{(1-H(u{-}|x))(1-F(u|x))}.
	\]
	For the second term in \eqref{eq:strong_rep_disc_x_decomp}, since $(\rho_x-n_x/n)/(n_x/n)=O_P(n^{-1/2})$ uniformly in $t$ and together with the previous result, this term converges to zero in probability. Also the remainder term $\sqrt{n}R_n(t|x)=o_P(1)$ from the result in \eqref{eq:strong_rep_disc_x} and hence completes the proof of \ref{enum:brownian_disc_x}.
	
	For \ref{enum:iid_sum_gaussian_approx}, we have that
	\begin{align}
		\label{eq:iid_sum_gaussian_approx_split}
		\begin{split}
			&\sup_{t \leq \uptau_{G}(x)}\left\vert
			\frac{1}{n_x}\sum_{i=1}^n\indicator{X_i=x}\xi_i(t|x)-n^{-1/2}{W}(t|x)
			\right\vert\\
			&\quad\leq
			\sup_{t \leq \uptau_{G}(x)}\left\vert
			\frac{1}{n_x}\sum_{i=1}^n{\xi}_i(t|x)-n_x^{-1/2}\rho_x^{1/2}{W}(t|x)
			\right\vert+
			\rho_x^{1/2}\left\vert
			\frac{1}{\sqrt{n_x}}-\frac{1}{\sqrt{n\rho_x}}
			\right\vert\sup_{t \leq \uptau_{G}(x)}\left\vert{W}(t|x)\right\vert\\
			&=(I)+(II).
		\end{split}
	\end{align}
	We will apply the Borel--Cantelli lemma to show that both $(I)$ and $(II)$ are $O(\log{n}/n)$ almost surely in Parts I and II below, respectively. 
	\paragraph*{Part I}
	For $(I)$, using the construction in \cite{MR1988}, we can show that the i.i.d. sum can be approximated by a Gaussian process, conditionally on $\mat{X}_1^n=(X_1,\dots,X_n)$. 
	Denote by 
	$\tilde{A}_{1n}(t|x)=\sum_{i=1}^n\indicator{X_i=x}\indicator{Y_i\leq t, \Delta_i=1}/{n_x}-H_1(t|x)$ and $\tilde{A}_{2n}(t|x)=\sum_{i=1}^n\indicator{X_i=x}\indicator{Y_i\geq t}/{n_x}-(1-H(t{-}|x))$.
	It can be shown, using integration by parts and the continuity of $H(\cdot|x)$ on $(0,\uptau_G(x))$, that
	\[
	\begin{split}
		&\frac{1}{n_x}\sum_{i=1}^n\indicator{X_i=x}\xi_i(t|x)
		=
		(1-F(t|x))
		\biggl\lbrace
		\frac{\tilde{A}_{1n}(t|x)}{1-H(t{-}|x)}\\
		&\quad
		-\int_{0}^t\frac{\tilde{A}_{1n}(s|x)}{(1-H(s{-}|x))^2}\dd{H(s{-}|x)}
		-\int_{0}^t\frac{\tilde{A}_{2n}(s|x)}{(1-H(s{-}|x))^2}\dd{H_1(s|x)}
		\biggl\rbrace.
	\end{split}
	\]
	With the help of the Koml\'{o}s--Major--Tusn\'{a}dy approximation, it can be shown that, conditionally on $\mat{X}_1^n$, $\tilde{A}_{1n}(t|x)$ and $\tilde{A}_{2n}(t|x)$ can be approximated by $S(H_1(t|x))$ and $S(H_1(t|x))-S(1-H_0(t|x))$, respectively, where $S(u)$ is a Brownian bridge. Consequently, the i.i.d. sum can be approximated by a Gaussian process, conditionally on $\mat{X}_1^n$. Specifically, for $x\in\family{X}$, define $T_1^x,T_2^x:\real\rightarrow[0,1]$ by $T_1^x(y)=H_1(y|x)$ and $T_2^x(y)=1-H_0(y|x)$. Define the random variables $U_1^x,\dots,U_{n_x}^x$ based on the subsample $\left\lbrace(Y_{xi},\Delta_{xi}),i=1,\dots,n_x\right\rbrace$ as
	\[
	U_i^x=
	\begin{cases}
		T_1^x(Y_{xi}),& \text{if }\Delta_{xi}=1,\\
		T_2^x(Y_{xi}),& \text{if }\Delta_{xi}=0.
	\end{cases}
	\]
	Then, conditionally on $\mat{X}_1^n$, $U_1^x,\dots,U_{n_x}^x$ form a sequence of independent standard uniform random variables. Consequently, with the help of Theorem~3 of \cite{KMT1975} and following the construction in \cite{MR1988}, we have that
	\[
	\begin{split}
		&\mathbb{P}\left[
		\sup_{t \leq \uptau_{G}(x)}{n_x}
		\left\vert
		\frac{1}{n_x}\sum_{i=1}^n\indicator{X_i=x}\xi_i(t|x) - n_x^{-1/2}\rho_x^{1/2}W(t|x)
		\right\vert
		> K_{1}\log n_x + u
		\middle\vert\mat{X}_1^n
		\right]\\
		&\quad
		<K_2\exp(-K_3u)
	\end{split}
	\]
	for all $u$, where $K_1$, $K_2$ and $K_3$ are positive constants. Therefore, let $\varepsilon>0$, we have the following probability bound:
	\begin{equation*}
		\begin{split}
			&\mathbb{P}\left(
			\sup_{t \leq \uptau_{G}(x)}\left\vert
			\frac{1}{n_x}\sum_{i=1}^n{\xi}_i(t|x)-n_x^{-1/2}\rho_x^{1/2}{W}(t|x)
			\right\vert
			>\varepsilon
			\right)\\
			&\leq
			\expect[\sbracket]{
				K_2\exp\left\lbrace
				-K_3(\varepsilon n_x-K_1\log{n_x})
				\right\rbrace
			}\\
			&\leq
			\expect[\sbracket]{
				K_2\exp\left\lbrace
				-K_3(\varepsilon n_x-K_1\log{n})
				\right\rbrace
			}\\
			&=K_2n^{K_1K_3}\expect[\sbracket]{
				\exp\left(-\varepsilon K_3n_x\right)
			}\\
			&\leq
			K_2n^{K_1K_3}\exp\left(
			-\frac{\varepsilon\rho_xK_3n}{1+\varepsilon K_3}
			\right),
		\end{split}
	\end{equation*}
	where the last inequality follows from a similar argument in the proof of Proposition~\ref{prop:DKW_ineq_disc_x}.
	Let $M>0$ and set $\varepsilon=M\log{n}/n$, we have
	\[
	\begin{split}
		&\mathbb{P}\left(
		\sup_{t \leq \uptau_{G}(x)}\left\vert
		\frac{1}{n_x}\sum_{i=1}^n{\xi}_i(t|x)-n_x^{-1/2}\rho_x^{1/2}{W}(t|x)
		\right\vert
		>M\frac{\log{n}}{n}
		\right)\\
		&\leq
		K_2n^{K_1K_3}\exp\left(
		-\frac{\rho_xMK_3\log{n}}{1+MK_3\log{n}/n}
		\right)=:a_n.
	\end{split}
	\]
	Since $\log{n}/n\converge0$ as $n\converge\infty$, there exists an $N_0\in\mathbb{N}$ such that, for $n\geq N_0$
	\[
	a_n
	\leq
	K_2n^{K_1K_3}\exp\left(
	-\frac{1}{2}\rho_xMK_3\log{n}
	\right)
	=K_2n^{K_3(K_1-\rho_xM/2)}.
	\]
	Thus, if $K_3(K_1-\rho_xM/2)<-1$ (or $M>2(K_1+1/K_3)/\rho_x$ equivalently), the series converges:
	\[
	\begin{split}
		&\sum_{n=1}^{\infty}
		\mathbb{P}\left(
		\sup_{t \leq \uptau_{G}(x)}\left\vert
		\frac{1}{n_x}\sum_{i=1}^n{\xi}_i(t|x)-n_x^{-1/2}\rho_x^{1/2}{W}(t|x)
		\right\vert
		>M\frac{\log{n}}{n}
		\right)\\
		&\leq 
		K_2\sum_{n=1}^{\infty}n^{K_3(K_1-\rho_xM/2)}<\infty.
	\end{split}
	\]
	Hence, by the Borel--Cantelli lemma, $(I)=O(\log{n}/{n})$ almost surely. 
	\paragraph*{Part II}
	For $(II)$ in \eqref{eq:iid_sum_gaussian_approx_split}, we note that
	\[
	\begin{split}
		&\rho_x^{1/2}\left\vert
		\frac{1}{\sqrt{n_x}}-\frac{1}{\sqrt{n\rho_x}}
		\right\vert\sup_{t \leq \uptau_{G}(x)}\left\vert{W}(t|x)\right\vert
		=\rho_x^{1/2}n^{-1/2}\left\vert
		\sqrt{{n}/{n_x}}-\rho_x^{-1/2}
		\right\vert\sup_{t \leq \uptau_{G}(x)}\left\vert{W}(t|x)\right\vert\\
		&\leq
		\frac{1}{2}\rho_x^{-1}n^{-1/2}\left(
		\left\vert
		{n_x}/{n}-\rho_x
		\right\vert
		+ O\left(
		\left\vert
		{n_x}/{n}-\rho_x
		\right\vert
		\right)
		\right)
		\sup_{t \leq \uptau_{G}(x)}\left\vert
		W(t|x)
		\right\vert.
	\end{split}
	\]
	Therefore, it suffices to show that $n^{-1/2}\left\vert{n_x}/{n}-\rho_x\right\vert\sup_{t\leq\uptau_{G}(x)}\left\vert W(t|x)\right\vert=O(\log{n}/n)$ almost surely. 
	By the Borell--TIS inequality (Theorem 2.1.1 of \cite{AT2007}), if the mean zero Gaussian process $W(t|x)$ is bounded almost surely on $[0,\uptau_{G}(x)]$, which will be claimed at the end of the proof,
	we have that%
	, for $u>\big\vert\mathbb{E}\lbrack
	\sup_{t\leq\uptau_{G}(x)}W(t|x)
	\rbrack\big\vert$,
	\begin{align}
		\label{eq:borell_ineq}
		\begin{split}
			&\prob{
				\left\vert
				\sup_{t\leq\uptau_{G}(x)}\left\vert W(t|x)\right\vert
				-\mathbb{E}\bigg({
					\sup_{t\leq\uptau_{G}(x)}W(t|x)
				}\bigg)\right\vert>u
			}\\
			&\le
			2\prob{
				\sup_{t\leq\uptau_{G}(x)}W(t|x)-
				\mathbb{E}\bigg({
					\sup_{t\leq\uptau_{G}(x)}W(t|x)
				}\bigg)
				>u
			}
			\le2\exp(-u^2/2\sigma_{\uptau_{G}(x)}^2),
		\end{split}
	\end{align}
	where $\sigma_{\uptau_{G}(x)}^2=\sup_{t \leq \uptau_{G}(x)}\expect{W^2(t|x)}$.
	Here the first inequality in the above display follows by symmetry and $\mathbb{E}\lbrack
	\sup_{t\leq\uptau_{G}(x)}W(t|x)
	\rbrack<\infty$.
	We note that, for our purpose, it is not necessary to consider subtracting $\mathbb{E}\big\lbrack
	\sup_{t\leq\uptau_{G}(x)}\big\vert W(t|x)\big\vert
	\big\rbrack$ in the above display, since $\mathbb{E}\lbrack
	\sup_{t\leq\uptau_{G}(x)}W(t|x)
	\rbrack$ is a constant and we will use the probability bound for sufficiently large $u$.
	Let $M>0$ and $C_1>2\sigma_{\uptau_{G}(x)}^2>0$, we have
	\begin{equation}
		\label{eq:n_x_prob_split}
		\begin{split}
			&\prob{
				n^{-1/2}\left\vert\frac{n_x}{n}-\rho_x\right\vert\sup_{t\leq\uptau_{G}(x)}\left\vert W(t|x)\right\vert
				>\frac{M\log{n}}{n}
			}\\
			&=
			\prob{
				\left\vert\frac{n_x}{n}-\rho_x\right\vert\sup_{t\leq\uptau_{G}(x)}\left\vert W(t|x)\right\vert
				>\frac{M\log{n}}{\sqrt n},\right.\\
				&\left.\qquad\qquad\qquad
				\left\vert
				\sup_{t\leq\uptau_{G}(x)}\left\vert W(t|x)\right\vert
				-\mathbb{E}\bigg({
					\sup_{t\leq\uptau_{G}(x)}W(t|x)
				}\bigg)
				\right\vert
				\le\sqrt{C_1\log{n}}
			}\\
			&\quad
			+
			\prob{
				\left\vert\frac{n_x}{n}-\rho_x\right\vert\sup_{t\leq\uptau_{G}(x)}\left\vert W(t|x)\right\vert
				>\frac{M\log{n}}{\sqrt n},\right.\\
				&\left.\qquad\qquad\qquad
				\left\vert
				\sup_{t\leq\uptau_{G}(x)}\left\vert W(t|x)\right\vert
				-\mathbb{E}\bigg({
					\sup_{t\leq\uptau_{G}(x)}W(t|x)
				}\bigg)
				\right\vert
				>\sqrt{C_1\log{n}}
			}
			=(A)+(B).
		\end{split}
	\end{equation}
	$(A)$ in \eqref{eq:n_x_prob_split} can be bounded using the reverse triangle inequality as follows:
	\[
	\begin{split}
		(A)
		&\le
		\prob{
			\left\vert\frac{n_x}{n}-\rho_x\right\vert
			\big\lbrace
			\big\vert\mathbb{E}\big({
				\sup W(t|x)
			}\big)\big\vert
			+\sqrt{C_1\log n}
			\big\rbrace
			>\frac{M\log{n}}{\sqrt n}
		}\\
		&\le
		\prob{
			\left\vert\frac{n_x}{n}-\rho_x\right\vert
			\cdot\big\vert
			\mathbb{E}\big({
				\sup W(t|x)
			}\big)\big\vert
			>\frac{M\log{n}}{2\sqrt n}
		}+
		\prob{
			\left\vert\frac{n_x}{n}-\rho_x\right\vert
			\sqrt{C_1\log n}
			>\frac{M\log{n}}{2\sqrt n}
		}.
	\end{split}
	\]
	Here and hereafter, we use $\sup$ instead of $\sup_{t\leq\uptau_{G}(x)}$ to simplify the notation.
	Both terms on the right-hand side of the last inequality in the above display can be bounded using Hoeffding's inequality, since $n_x=\sum_{i=1}^n\indicator{X_i=x}\sim B(n,\rho_x)$:
	\[
	\begin{split}
		\prob{
			\left\vert\frac{n_x}{n}-\rho_x\right\vert
			\cdot\big\vert
			\mathbb{E}\big({
				\sup W(t|x)
			}\big)\big\vert
			>\frac{M\log{n}}{2\sqrt n}
		}
		&\le
		2\exp\left(
		-\frac{(M\log n)^2}{2\vert
			\mathbb{E}\big(
			\sup W(t|x)
			\big)\vert^2
		}\right)\\
		&\le
		2n^{-M^2/2\vert
			\mathbb{E}(
			\sup W(t|x)
			)\vert^2},
	\end{split}
	\]
	where the last inequality holds since $\log n \le(\log n)^2$ for $n\ge\exp(1)$. Also, for the other term, we have:
	\[
	\prob{
		\left\vert\frac{n_x}{n}-\rho_x\right\vert
		\sqrt{C_1\log n}
		>\frac{M\log{n}}{2\sqrt n}
	}
	\leq
	2\exp\left(
	-\frac{M^2}{4C_1}\log{n}
	\right)
	=2n^{-M^2/4C_1}.
	\]
	Thus, if $M>\sqrt{2}\vert\mathbb{E}(\sup W(t|x))\vert\vee2\sqrt{C_1}>0$, the series converges:
	\[
	\begin{split}
		&\sum_{n=1}^{\infty}\prob{
			\left\vert\frac{n_x}{n}-\rho_x\right\vert\sup\left\vert W(t|x)\right\vert
			>\frac{M\log{n}}{\sqrt n},~
			\left\vert
			\sup\left\vert W(t|x)\right\vert
			-\mathbb{E}\big({
				\sup W(t|x)
			}\big)
			\right\vert
			\le\sqrt{C_1\log{n}}
		}\\
		&\leq
		2+
		2\sum_{n=3}^\infty
		n^{-M^2/2\vert
			\mathbb{E}(
			\sup W(t|x)
			)\vert^2}
		+
		2\sum_{n=1}^{\infty}n^{-M^2/4C_1}<\infty.
	\end{split}
	\]
	$(B)$ in \eqref{eq:n_x_prob_split} can be bounded using the inequality in \eqref{eq:borell_ineq}, for $n>\exp\lbrace\lbrack\mathbb{E}(\sup W(t|x))\rbrack^2/C_1\rbrace$
	\[
	\begin{split}
		(B)
		&\leq
		\prob{
			\left\vert
			\sup\left\vert W(t|x)\right\vert
			-\mathbb{E}\big({
				\sup W(t|x)
			}\big)\right\vert>\sqrt{C_1\log{n}}
		}\\
		&\le2\exp(-C_1\log n/2\sigma_{\uptau_{G}(x)}^2)
		=2n^{-C_1/2\sigma_{\uptau_{G}(x)}^2}.
	\end{split}
	\]
	Since $C_1>2\sigma_{\uptau_{G}(x)}^2$ and with $n_0>\exp\lbrace\lbrack\mathbb{E}(\sup W(t|x))\rbrack^2/C_1\rbrace$, the series converges:
	\[
	\begin{split}
		&\sum_{n=n_0}^\infty
		\prob{
			\left\vert\frac{n_x}{n}-\rho_x\right\vert\sup\left\vert W(t|x)\right\vert
			>\frac{M\log{n}}{\sqrt n},~
			\left\vert
			\sup\left\vert W(t|x)\right\vert
			-\mathbb{E}\big({
				\sup W(t|x)
			}\big)
			\right\vert
			>\sqrt{C_1\log{n}}
		}\\
		&\le2\sum_{n=n_0}^\infty
		n^{-C_1/2\sigma_{\uptau_{G}(x)}^2}<\infty.
	\end{split}
	\]
	Thus, by the Borel--Cantelli lemma, $(II)=O(\log{n}/n)$ almost surely. It remains to show that $W(t|x)$ is bounded almost surely on $[0,\uptau_G(x)]$. This can be verified using Theorem 1.4.1 in \cite{AT2007}. Specifically, it suffices to show that, for some $\delta>0$,
	\[
	\int_{\delta}^{\infty}q(\exp(-u^2))\dd{u}<\infty,
	\]
	where $q^2(u)=\sup_{|s-t|\le u}\mathbb{E}\left\vert W(s|x)-W(t|x)\right\vert^2$.
	From the covariance function of $W(t|x)$, it can be shown that
	$\expect{\left\vert W(s|x)-W(t|x)\right\vert^2}\le C_2|s-t|$,
	for some positive constant $C_2$.
	Thus
	\[
	\int_{\delta}^{\infty}q(\exp(-u^2))\dd{u}
	\le C_2\int_{\delta}^{\infty}\exp(-u^2/2)\dd{u}
	=C_2\sqrt{2\pi}(1-\Phi(\delta))
	<\infty,
	\]
	where $\Phi(\cdot)$ is the distribution function of the standard normal distribution. Hence the proof of \ref{enum:iid_sum_gaussian_approx} is completed.
\end{proof}

\begin{lemma}
	\label{lemma:lcm_kme_kme_unif_dist_order}
	Suppose that Assumptions~\ref*{assumption:discrete_cov}--\ref*{assumption:positive_cure_prob} hold.
	Then, we have, for each $x\in\family{X}$,
	\begin{equation}
		\label{eq:lcm_kme_kme_unif_dist_order}
		\sup_{t\in[a_x,\uptau_{G}(x)]}\left\vert
		\hat{F}_n^G(t|x)-\hat{F}_n(t|x)
		\right\vert=O_P\left(\frac{\log{n}}{n}\right)^{2/3},
	\end{equation}
	where $\hat{F}_n$ is the conditional Kaplan--Meier estimator in \eqref{eq:kme} and $\hat{F}_n^G$ is its least concave majorant.
\end{lemma}
\begin{proof}[Proof of Lemma~\ref{lemma:lcm_kme_kme_unif_dist_order}]
	We apply Theorem~2.2 in \cite{DL2014} to show \eqref{eq:lcm_kme_kme_unif_dist_order}, for which it suffices to verify the following conditions.
	Condition~(A1) in \cite{DL2014} is satisfied under our assumptions on the function $t\mapsto f_u(t|x)$ for $t\in[a_x,\uptau_{G}(x)]$. Denote by 
	$$L(t|x)=\int_0^t\frac{\dd{F(s|x)}}{(1-H(s{-}|x))(1-F(s|x))},$$ for $t\in[a_x, \uptau_G(x)]$, which is non-decreasing and continuously differentiable with respect to $t$ on $[a_x, \uptau_G(x)]$ satisfying $0<\inf_{t\in[a_x,\uptau_{G}(x)]}L'(t|x)\leq\sup_{t\in[a_x,\uptau_{G}(x)]}L'(t|x)<\infty$. Therefore condition~(A4) in \cite{DL2014} is satisfied. Let $$B(t|x)={\rho_x^{-1/2}}\lbrace1-F(L^{-1}(t|x)|x)\rbrace\tilde{W}(t)$$ for $t\in[L(a_x|x),L(\uptau_G(x)|x)]$, where $\tilde{W}(t)$ is a Brownian motion. From Lemma~\ref{lemma:strong_rep_brownian_disc_x}, we have
	\[
	\sup_{t\in[a_x,\uptau_{G}(x)]}\left\vert
	\hat{F}_n(t|x)-F(t|x)-n^{-1/2}B\circ L(t|x)
	\right\vert=O_P\left(\frac{\log{n}}{n}\right)^{3/4}.
	\]
	Hence condition~(2) in \cite{DL2014} is satisfied. 
	To show that $B(t|x)$ satisfies conditions (A2) and (A3) in \cite{DL2014} with $\tau=1$, we can apply the same arguments as in the proof of Lemma~4.3 in \cite{LM2017}, which are omitted for brevity.
\end{proof}

\bibliography{references.bib}

\clearpage
\pagebreak
\begin{center}
	\textbf{\large Supplementary Material: \\
		Testing for sufficient follow-up in cure models with categorical
		covariates\\[5pt]}
\end{center}

\setcounter{section}{0}
\setcounter{equation}{0}
\setcounter{figure}{0}
\setcounter{table}{0}
\setcounter{page}{1}
\makeatletter
\renewcommand{\thesection}{S\arabic{section}}
\renewcommand{\theequation}{S\arabic{equation}}
\renewcommand{\thefigure}{S\arabic{figure}}
\renewcommand{\thetable}{S\arabic{table}}

\section{Simulation study}
\subsection{Censoring rate}\label{supp_sec:cens_rate}
\subsubsection*{Setting~\ref*{enum:sim_logistic_weibull_unif}}
\begin{table}[H]
	\begin{center}
		\caption{Censoring rate for Setting~\ref*{enum:sim_logistic_weibull_unif} with $\Delta G(\uptau_G(x))=0.01$.\label{supp_tab:sim_cens_rate_logistic_weibull_unif_j_001_rho_comb}}
		\scriptsize
		\begin{tabular}{ccc@{\extracolsep{6pt}}ccc@{\extracolsep{6pt}}ccc}
& \multicolumn{2}{c}{$\uptau_G(x)$} & \multicolumn{3}{c}{$\rho=0.3$} & \multicolumn{3}{c}{$\rho=0.5$} \\ \cline{2-3}\cline{4-6}\cline{7-9}
Scenario & $x=0$ & $x=1$ & all & $x=0$ & $x=1$ & all & $x=0$ & $x=1$ \\ \hline
\multirow{3}{*}{A} & 0.950 & 0.950 & 0.669 & 0.740 & 0.506 & 0.624 & 0.739 & 0.508 \\ 
   & 0.950 & 0.975 & 0.659 & 0.740 & 0.471 & 0.606 & 0.739 & 0.474 \\ 
   & 0.975 & 0.975 & 0.646 & 0.721 & 0.471 & 0.597 & 0.721 & 0.474 \\ 
   \hline
\multirow{8}{*}{B} & 0.975 & 0.990 & 0.636 & 0.721 & 0.439 & 0.581 & 0.721 & 0.442 \\ 
   & 0.975 & 0.999 & 0.624 & 0.721 & 0.397 & 0.559 & 0.721 & 0.398 \\ 
   & 0.990 & 0.975 & 0.634 & 0.704 & 0.471 & 0.589 & 0.704 & 0.474 \\ 
   & 0.990 & 0.990 & 0.625 & 0.704 & 0.439 & 0.573 & 0.704 & 0.442 \\ 
   & 0.990 & 0.999 & 0.612 & 0.704 & 0.397 & 0.551 & 0.704 & 0.398 \\ 
   & 0.995 & 0.975 & 0.628 & 0.696 & 0.471 & 0.584 & 0.695 & 0.474 \\ 
   & 0.995 & 0.990 & 0.619 & 0.696 & 0.439 & 0.568 & 0.695 & 0.442 \\ 
   & 0.999 & 0.990 & 0.609 & 0.682 & 0.439 & 0.562 & 0.682 & 0.442 \\ 
   \hline
\multirow{4}{*}{C} & 0.995 & 0.995 & 0.613 & 0.696 & 0.423 & 0.560 & 0.695 & 0.424 \\ 
   & 0.995 & 0.999 & 0.606 & 0.696 & 0.397 & 0.547 & 0.695 & 0.398 \\ 
   & 0.999 & 0.995 & 0.604 & 0.682 & 0.423 & 0.553 & 0.682 & 0.424 \\ 
   & 0.999 & 0.999 & 0.596 & 0.682 & 0.397 & 0.540 & 0.682 & 0.398 \\ 
   \hline
\end{tabular}

	\end{center}
\end{table}

\subsubsection*{Setting~\ref*{enum:sim_exp_exp}}
\begin{table}[H]
	\begin{center}
		\caption{Censoring rate for Setting~\ref*{enum:sim_exp_exp} with $p(0)=0.7$, $p(1)=0.7$, $\rho\in\{0.3,0.5\}$.\label{supp_tab:sim_cens_rate_exp_exp_exp_p0_07_p1_07_rho_comb}}
		\scriptsize
		\begin{tabular}{ccc@{\extracolsep{12pt}}ccc@{\extracolsep{12pt}}ccc}
& \multicolumn{2}{c}{$\uptau_G(x)$} & \multicolumn{3}{c}{$\rho=0.3$} & \multicolumn{3}{c}{$\rho=0.5$} \\ \cline{2-3}\cline{4-6}\cline{7-9}
Scenario & $x=0$ & $x=1$ & all & $x=0$ & $x=1$ & all & $x=0$ & $x=1$ \\ \hline
\multirow{4}{*}{A} & 0.950 & 0.950 & 0.469 & 0.433 & 0.553 & 0.494 & 0.432 & 0.555 \\ 
   & 0.950 & 0.975 & 0.468 & 0.433 & 0.550 & 0.492 & 0.432 & 0.552 \\ 
   & 0.975 & 0.975 & 0.462 & 0.424 & 0.550 & 0.488 & 0.423 & 0.552 \\ 
   & 0.999 & 0.950 & 0.458 & 0.417 & 0.553 & 0.486 & 0.417 & 0.555 \\ 
   \hline
\multirow{7}{*}{B} & 0.975 & 0.990 & 0.462 & 0.424 & 0.549 & 0.487 & 0.423 & 0.551 \\ 
   & 0.975 & 0.999 & 0.461 & 0.424 & 0.549 & 0.487 & 0.423 & 0.551 \\ 
   & 0.990 & 0.975 & 0.459 & 0.419 & 0.550 & 0.486 & 0.419 & 0.552 \\ 
   & 0.990 & 0.990 & 0.458 & 0.419 & 0.549 & 0.485 & 0.419 & 0.551 \\ 
   & 0.990 & 0.999 & 0.458 & 0.419 & 0.549 & 0.485 & 0.419 & 0.551 \\ 
   & 0.995 & 0.975 & 0.458 & 0.418 & 0.550 & 0.485 & 0.418 & 0.552 \\ 
   & 0.999 & 0.990 & 0.457 & 0.417 & 0.549 & 0.484 & 0.417 & 0.551 \\ 
   \hline
\multirow{4}{*}{C} & 0.995 & 0.995 & 0.457 & 0.418 & 0.549 & 0.484 & 0.418 & 0.551 \\ 
   & 0.995 & 0.999 & 0.457 & 0.418 & 0.549 & 0.484 & 0.418 & 0.551 \\ 
   & 0.999 & 0.995 & 0.457 & 0.417 & 0.549 & 0.484 & 0.417 & 0.551 \\ 
   & 0.999 & 0.999 & 0.457 & 0.417 & 0.549 & 0.484 & 0.417 & 0.551 \\ 
   \hline
\end{tabular}

	\end{center}
\end{table}

\begin{table}[H]
	\begin{center}
		\caption{Censoring rate for Setting~\ref*{enum:sim_exp_exp} with $\rho=0.5$, $p(x)\in\{0.3,0.7\}$.\label{supp_tab:sim_cens_rate_exp_exp_exp_rho_05_p_comb}}
		\scriptsize
		\setlength{\tabcolsep}{5pt}
		\begin{tabular}{ccc@{\extracolsep{6pt}}ccc@{\extracolsep{6pt}}ccc@{\extracolsep{6pt}}ccc}
& \multicolumn{2}{c}{$\uptau_G(x)$} & \multicolumn{3}{c}{$p(0)=p(1)=0.7$} & \multicolumn{3}{c}{$p(0)=0.3$, $p(1)=0.7$} & \multicolumn{3}{c}{$p(0)=p(1)=0.3$} \\ \cline{2-3}\cline{4-6}\cline{7-9}\cline{10-12}
Scenario & $x=0$ & $x=1$ & all & $x=0$ & $x=1$ & all & $x=0$ & $x=1$ & all & $x=0$ & $x=1$ \\ \hline
\multirow{4}{*}{A} & 0.950 & 0.950 & 0.494 & 0.432 & 0.555 & 0.656 & 0.757 & 0.555 & 0.783 & 0.757 & 0.809 \\ 
   & 0.950 & 0.975 & 0.492 & 0.432 & 0.552 & 0.655 & 0.757 & 0.552 & 0.782 & 0.757 & 0.808 \\ 
   & 0.975 & 0.975 & 0.488 & 0.423 & 0.552 & 0.653 & 0.754 & 0.552 & 0.781 & 0.754 & 0.808 \\ 
   & 0.999 & 0.950 & 0.486 & 0.417 & 0.555 & 0.653 & 0.751 & 0.555 & 0.780 & 0.751 & 0.809 \\ 
   \hline
\multirow{7}{*}{B} & 0.975 & 0.990 & 0.487 & 0.423 & 0.551 & 0.652 & 0.754 & 0.551 & 0.780 & 0.754 & 0.807 \\ 
   & 0.975 & 0.999 & 0.487 & 0.423 & 0.551 & 0.652 & 0.754 & 0.551 & 0.780 & 0.754 & 0.807 \\ 
   & 0.990 & 0.975 & 0.486 & 0.419 & 0.552 & 0.652 & 0.752 & 0.552 & 0.780 & 0.752 & 0.808 \\ 
   & 0.990 & 0.990 & 0.485 & 0.419 & 0.551 & 0.651 & 0.752 & 0.551 & 0.780 & 0.752 & 0.807 \\ 
   & 0.990 & 0.999 & 0.485 & 0.419 & 0.551 & 0.651 & 0.752 & 0.551 & 0.779 & 0.752 & 0.807 \\ 
   & 0.995 & 0.975 & 0.485 & 0.418 & 0.552 & 0.652 & 0.751 & 0.552 & 0.779 & 0.751 & 0.808 \\ 
   & 0.999 & 0.990 & 0.484 & 0.417 & 0.551 & 0.651 & 0.751 & 0.551 & 0.779 & 0.751 & 0.807 \\ 
   \hline
\multirow{4}{*}{C} & 0.995 & 0.995 & 0.484 & 0.418 & 0.551 & 0.651 & 0.751 & 0.551 & 0.779 & 0.751 & 0.807 \\ 
   & 0.995 & 0.999 & 0.484 & 0.418 & 0.551 & 0.651 & 0.751 & 0.551 & 0.779 & 0.751 & 0.807 \\ 
   & 0.999 & 0.995 & 0.484 & 0.417 & 0.551 & 0.651 & 0.751 & 0.551 & 0.779 & 0.751 & 0.807 \\ 
   & 0.999 & 0.999 & 0.484 & 0.417 & 0.551 & 0.651 & 0.751 & 0.551 & 0.779 & 0.751 & 0.807 \\ 
   \hline
\end{tabular}

	\end{center}
\end{table}

\subsubsection*{Setting~\ref*{enum:sim_exp_unif_2cov}}
\begin{table}[H]
	\begin{center}
		\caption{Censoring rate for Setting~\ref*{enum:sim_exp_unif_2cov}.\label{supp_tab:sim_cens_rate_exp_unif_2cov}}
		\scriptsize
		\begin{tabular}{ccccc@{\extracolsep{12pt}}ccccc}
\multirow{2}{*}{Scenario} & \multicolumn{4}{c}{$\uptau_G(x)$} & \multicolumn{5}{c}{Censoring rate} \\ \cline{2-5}\cline{6-10}
& $(0,0)$ & $(0,1)$ & $(1,0)$ & $(1,1)$ & all & $(0,0)$ & $(0,1)$ & $(1,0)$ & $(1,1)$ \\ \hline
\multirow{7}{*}{A} & 0.950 & 0.950 & 0.950 & 0.950 & 0.741 & 0.758 & 0.830 & 0.594 & 0.727 \\ 
   & 0.950 & 0.975 & 0.950 & 0.950 & 0.741 & 0.758 & 0.830 & 0.594 & 0.727 \\ 
   & 0.975 & 0.975 & 0.950 & 0.950 & 0.716 & 0.739 & 0.813 & 0.554 & 0.690 \\ 
   & 0.975 & 0.975 & 0.975 & 0.950 & 0.716 & 0.739 & 0.813 & 0.554 & 0.690 \\ 
   & 0.975 & 0.975 & 0.975 & 0.975 & 0.716 & 0.739 & 0.813 & 0.554 & 0.690 \\ 
   & 0.975 & 0.990 & 0.975 & 0.975 & 0.716 & 0.739 & 0.813 & 0.554 & 0.690 \\ 
   & 0.999 & 0.999 & 0.950 & 0.999 & 0.648 & 0.697 & 0.769 & 0.456 & 0.581 \\ 
   \hline
\multirow{7}{*}{B} & 0.975 & 0.990 & 0.999 & 0.999 & 0.716 & 0.739 & 0.813 & 0.554 & 0.690 \\ 
   & 0.975 & 0.999 & 0.999 & 0.999 & 0.716 & 0.739 & 0.813 & 0.554 & 0.690 \\ 
   & 0.990 & 0.990 & 0.990 & 0.990 & 0.689 & 0.722 & 0.796 & 0.515 & 0.650 \\ 
   & 0.990 & 0.995 & 0.995 & 0.990 & 0.689 & 0.722 & 0.796 & 0.515 & 0.650 \\ 
   & 0.990 & 0.999 & 0.990 & 0.990 & 0.689 & 0.722 & 0.796 & 0.515 & 0.650 \\ 
   & 0.990 & 0.999 & 0.995 & 0.990 & 0.689 & 0.722 & 0.796 & 0.515 & 0.650 \\ 
   & 0.999 & 0.999 & 0.990 & 0.999 & 0.648 & 0.697 & 0.769 & 0.456 & 0.581 \\ 
   \hline
\multirow{5}{*}{C} & 0.995 & 0.995 & 0.995 & 0.995 & 0.674 & 0.712 & 0.786 & 0.492 & 0.625 \\ 
   & 0.995 & 0.999 & 0.995 & 0.995 & 0.674 & 0.712 & 0.786 & 0.492 & 0.625 \\ 
   & 0.999 & 0.999 & 0.995 & 0.995 & 0.648 & 0.697 & 0.769 & 0.456 & 0.581 \\ 
   & 0.999 & 0.999 & 0.999 & 0.995 & 0.648 & 0.697 & 0.769 & 0.456 & 0.581 \\ 
   & 0.999 & 0.999 & 0.999 & 0.999 & 0.648 & 0.697 & 0.769 & 0.456 & 0.581 \\ 
   \hline
\end{tabular}

	\end{center}
\end{table}

\clearpage
\subsubsection*{Setting~\ref*{enum:sim_exp_unif}}
\begin{table}[H]
	\begin{center}
		\caption{Censoring rate for Setting~\ref*{enum:sim_exp_unif} with $\rho=0.5$, $p(0)=0.6$, $p(1)=0.6$, $\Delta G(\uptau_G(x))\in\{0,0.01\}$.\label{supp_tab:sim_cens_rate_exp_exp_unif_rho_05_p0_06_p1_06_j_comb}}
		\scriptsize
		\begin{tabular}{ccc@{\extracolsep{12pt}}ccc@{\extracolsep{12pt}}ccc}
& \multicolumn{2}{c}{$\uptau_G(x)$} & \multicolumn{3}{c}{$\Delta G(\uptau_G(x))=0$} & \multicolumn{3}{c}{$\Delta G(\uptau_G(x))=0.01$} \\ \cline{2-3}\cline{4-6}\cline{7-9}
Scenario & $x=0$ & $x=1$ & all & $x=0$ & $x=1$ & all & $x=0$ & $x=1$ \\ \hline
\multirow{5}{*}{A} & 0.950 & 0.950 & 0.590 & 0.590 & 0.590 & 0.588 & 0.588 & 0.588 \\ 
   & 0.950 & 0.975 & 0.574 & 0.590 & 0.559 & 0.573 & 0.588 & 0.557 \\ 
   & 0.975 & 0.975 & 0.558 & 0.558 & 0.559 & 0.557 & 0.557 & 0.557 \\ 
   & 0.990 & 0.950 & 0.559 & 0.529 & 0.590 & 0.558 & 0.528 & 0.588 \\ 
   & 0.999 & 0.950 & 0.538 & 0.486 & 0.590 & 0.537 & 0.486 & 0.587 \\ 
   \hline
\multirow{7}{*}{B} & 0.975 & 0.990 & 0.544 & 0.558 & 0.529 & 0.542 & 0.557 & 0.528 \\ 
   & 0.975 & 0.999 & 0.522 & 0.558 & 0.486 & 0.521 & 0.556 & 0.485 \\ 
   & 0.990 & 0.975 & 0.544 & 0.529 & 0.559 & 0.542 & 0.528 & 0.557 \\ 
   & 0.990 & 0.990 & 0.529 & 0.529 & 0.529 & 0.528 & 0.527 & 0.528 \\ 
   & 0.990 & 0.999 & 0.508 & 0.529 & 0.486 & 0.506 & 0.527 & 0.485 \\ 
   & 0.999 & 0.975 & 0.523 & 0.486 & 0.559 & 0.521 & 0.486 & 0.557 \\ 
   & 0.999 & 0.990 & 0.508 & 0.486 & 0.529 & 0.507 & 0.485 & 0.528 \\ 
   \hline
\multirow{4}{*}{C} & 0.995 & 0.995 & 0.513 & 0.512 & 0.513 & 0.511 & 0.511 & 0.512 \\ 
   & 0.995 & 0.999 & 0.499 & 0.512 & 0.486 & 0.498 & 0.511 & 0.485 \\ 
   & 0.999 & 0.995 & 0.500 & 0.486 & 0.513 & 0.498 & 0.485 & 0.511 \\ 
   & 0.999 & 0.999 & 0.486 & 0.486 & 0.486 & 0.485 & 0.485 & 0.485 \\ 
   \hline
\end{tabular}

	\end{center}
\end{table}

\begin{table}[H]
	\begin{center}
		\caption{Censoring rate for Setting~\ref*{enum:sim_exp_unif} with $p(0)=0.6$, $p(1)=0.6$, $\Delta G(\uptau_G(x))=0.01$, $\rho\in\{0.3,0.5\}$.\label{supp_tab:sim_cens_rate_exp_exp_unif_p0_06_p1_06_j_001_rho_comb}}
		\scriptsize
		\begin{tabular}{ccc@{\extracolsep{12pt}}ccc@{\extracolsep{12pt}}ccc}
& \multicolumn{2}{c}{$\uptau_G(x)$} & \multicolumn{3}{c}{$\rho=0.3$} & \multicolumn{3}{c}{$\rho=0.5$} \\ \cline{2-3}\cline{4-6}\cline{7-9}
Scenario & $x=0$ & $x=1$ & all & $x=0$ & $x=1$ & all & $x=0$ & $x=1$ \\ \hline
\multirow{5}{*}{A} & 0.950 & 0.950 & 0.588 & 0.589 & 0.587 & 0.588 & 0.588 & 0.588 \\ 
   & 0.950 & 0.975 & 0.578 & 0.588 & 0.555 & 0.573 & 0.588 & 0.557 \\ 
   & 0.975 & 0.975 & 0.556 & 0.557 & 0.555 & 0.557 & 0.557 & 0.557 \\ 
   & 0.990 & 0.950 & 0.545 & 0.528 & 0.586 & 0.558 & 0.528 & 0.588 \\ 
   & 0.999 & 0.950 & 0.516 & 0.486 & 0.586 & 0.537 & 0.486 & 0.587 \\ 
   \hline
\multirow{7}{*}{B} & 0.975 & 0.990 & 0.548 & 0.557 & 0.526 & 0.542 & 0.557 & 0.528 \\ 
   & 0.975 & 0.999 & 0.535 & 0.556 & 0.484 & 0.521 & 0.556 & 0.485 \\ 
   & 0.990 & 0.975 & 0.536 & 0.528 & 0.555 & 0.542 & 0.528 & 0.557 \\ 
   & 0.990 & 0.990 & 0.527 & 0.528 & 0.526 & 0.528 & 0.527 & 0.528 \\ 
   & 0.990 & 0.999 & 0.514 & 0.527 & 0.484 & 0.506 & 0.527 & 0.485 \\ 
   & 0.999 & 0.975 & 0.507 & 0.486 & 0.554 & 0.521 & 0.486 & 0.557 \\ 
   & 0.999 & 0.990 & 0.498 & 0.486 & 0.525 & 0.507 & 0.485 & 0.528 \\ 
   \hline
\multirow{4}{*}{C} & 0.995 & 0.995 & 0.511 & 0.512 & 0.509 & 0.511 & 0.511 & 0.512 \\ 
   & 0.995 & 0.999 & 0.503 & 0.512 & 0.484 & 0.498 & 0.511 & 0.485 \\ 
   & 0.999 & 0.995 & 0.493 & 0.486 & 0.509 & 0.498 & 0.485 & 0.511 \\ 
   & 0.999 & 0.999 & 0.485 & 0.486 & 0.484 & 0.485 & 0.485 & 0.485 \\ 
   \hline
\end{tabular}

	\end{center}
\end{table}

\begin{table}[H]
	\begin{center}
		\caption{Censoring rate for Setting~\ref*{enum:sim_exp_unif} with $\rho=0.5$, $\Delta G(\uptau_G(x))=0.01$, $p(x)\in\{0.4,0.6\}$.\label{supp_tab:sim_cens_rate_exp_exp_unif_rho_05_j_001_p_comb}}
		\scriptsize
		\setlength{\tabcolsep}{5.5pt}
		\begin{tabular}{ccc@{\extracolsep{6pt}}ccc@{\extracolsep{6pt}}ccc@{\extracolsep{6pt}}ccc}
& \multicolumn{2}{c}{$\uptau_G(x)$} & \multicolumn{3}{c}{$p(0)=p(1)=0.6$} & \multicolumn{3}{c}{$p(0)=0.6$, $p(1)=0.4$} & \multicolumn{3}{c}{$p(0)=p(1)=0.4$} \\ \cline{2-3}\cline{4-6}\cline{7-9}\cline{10-12}
Scenario & $x=0$ & $x=1$ & all & $x=0$ & $x=1$ & all & $x=0$ & $x=1$ & all & $x=0$ & $x=1$ \\ \hline
\multirow{5}{*}{A} & 0.950 & 0.950 & 0.588 & 0.588 & 0.588 & 0.657 & 0.588 & 0.726 & 0.726 & 0.726 & 0.726 \\ 
   & 0.950 & 0.975 & 0.573 & 0.588 & 0.557 & 0.647 & 0.588 & 0.705 & 0.715 & 0.726 & 0.705 \\ 
   & 0.975 & 0.975 & 0.557 & 0.557 & 0.557 & 0.631 & 0.557 & 0.705 & 0.705 & 0.705 & 0.705 \\ 
   & 0.990 & 0.950 & 0.558 & 0.528 & 0.588 & 0.627 & 0.528 & 0.725 & 0.706 & 0.686 & 0.725 \\ 
   & 0.999 & 0.950 & 0.537 & 0.486 & 0.587 & 0.605 & 0.486 & 0.725 & 0.691 & 0.658 & 0.725 \\ 
   \hline
\multirow{7}{*}{B} & 0.975 & 0.990 & 0.542 & 0.557 & 0.528 & 0.621 & 0.557 & 0.686 & 0.695 & 0.705 & 0.686 \\ 
   & 0.975 & 0.999 & 0.521 & 0.556 & 0.485 & 0.607 & 0.556 & 0.657 & 0.681 & 0.704 & 0.657 \\ 
   & 0.990 & 0.975 & 0.542 & 0.528 & 0.557 & 0.616 & 0.528 & 0.705 & 0.695 & 0.686 & 0.705 \\ 
   & 0.990 & 0.990 & 0.528 & 0.527 & 0.528 & 0.607 & 0.527 & 0.686 & 0.686 & 0.685 & 0.686 \\ 
   & 0.990 & 0.999 & 0.506 & 0.527 & 0.485 & 0.592 & 0.527 & 0.657 & 0.671 & 0.685 & 0.657 \\ 
   & 0.999 & 0.975 & 0.521 & 0.486 & 0.557 & 0.595 & 0.486 & 0.705 & 0.681 & 0.658 & 0.705 \\ 
   & 0.999 & 0.990 & 0.507 & 0.485 & 0.528 & 0.586 & 0.485 & 0.685 & 0.672 & 0.658 & 0.685 \\ 
   \hline
\multirow{4}{*}{C} & 0.995 & 0.995 & 0.511 & 0.511 & 0.512 & 0.593 & 0.511 & 0.675 & 0.675 & 0.675 & 0.675 \\ 
   & 0.995 & 0.999 & 0.498 & 0.511 & 0.485 & 0.584 & 0.511 & 0.657 & 0.666 & 0.675 & 0.657 \\ 
   & 0.999 & 0.995 & 0.498 & 0.485 & 0.511 & 0.580 & 0.485 & 0.674 & 0.666 & 0.658 & 0.674 \\ 
   & 0.999 & 0.999 & 0.485 & 0.485 & 0.485 & 0.571 & 0.485 & 0.657 & 0.657 & 0.658 & 0.657 \\ 
   \hline
\end{tabular}

	\end{center}
\end{table}

\subsection{Simulation results}
\subsubsection*{Setting~\ref*{enum:sim_logistic_weibull_unif}}
\begin{table}[H]
	\begin{center}
		\caption{Rejection rate for Setting~\ref*{enum:sim_logistic_weibull_unif} with $\Delta G(\uptau_G(x))=0.01$, $n=1000$.\label{supp_tab:sim_res_logistic_weibull_unif_j_001_n1000_rho_comb}}
		\scriptsize
		\begin{tabular}{ccc@{\extracolsep{6pt}}cc@{\extracolsep{6pt}}cc@{\extracolsep{6pt}}cc@{\extracolsep{6pt}}cc}
& & & \multicolumn{4}{c}{$\rho=0.3$} & \multicolumn{4}{c}{$\rho=0.5$} \\ \cline{4-7}\cline{8-11}
& \multicolumn{2}{c}{$\uptau_G(x)$} & \multicolumn{2}{c}{\makecell{Rej. rate of\\$H_{0x}$}} & \multicolumn{2}{c}{\makecell{Rej. rate of\\$H_0$}} & \multicolumn{2}{c}{\makecell{Rej. rate of\\$H_{0x}$}} & \multicolumn{2}{c}{\makecell{Rej. rate of\\$H_0$}} \\ \cline{2-3}\cline{4-5}\cline{6-7}\cline{8-9}\cline{10-11}
Scenario & $x=0$ & $x=1$ & $x=0$ & $x=1$ & Method 1 & Method 2 & $x=0$ & $x=1$ & Method 1 & Method 2 \\ \hline
\multirow{3}{*}{A} & 0.950 & 0.950 & 0.008 & 0.014 & 0.000 & 0.016 & 0.004 & 0.012 & 0.000 & 0.006 \\ 
   & 0.950 & 0.975 & 0.008 & 0.030 & 0.000 & 0.010 & 0.004 & 0.022 & 0.000 & 0.008 \\ 
   & 0.975 & 0.975 & 0.016 & 0.026 & 0.000 & 0.016 & 0.016 & 0.020 & 0.000 & 0.026 \\ 
   \hline
\multirow{6}{*}{B} & 0.975 & 0.990 & 0.016 & 0.230 & 0.006 & 0.018 & 0.016 & 0.170 & 0.004 & 0.018 \\ 
   & 0.975 & 0.999 & 0.016 & 0.946 & 0.016 & 0.016 & 0.016 & 0.974 & 0.016 & 0.016 \\ 
   & 0.990 & 0.975 & 0.096 & 0.026 & 0.002 & 0.090 & 0.096 & 0.016 & 0.002 & 0.090 \\ 
   & 0.990 & 0.990 & 0.096 & 0.220 & 0.022 & 0.094 & 0.096 & 0.160 & 0.012 & 0.092 \\ 
   & 0.990 & 0.999 & 0.096 & 0.946 & 0.092 & 0.096 & 0.096 & 0.980 & 0.094 & 0.096 \\ 
   & 0.999 & 0.990 & 0.848 & 0.224 & 0.194 & 0.790 & 0.832 & 0.164 & 0.132 & 0.776 \\ 
   \hline
\multirow{4}{*}{C} & 0.995 & 0.995 & 0.318 & 0.486 & 0.156 & 0.308 & 0.304 & 0.420 & 0.132 & 0.298 \\ 
   & 0.995 & 0.999 & 0.318 & 0.952 & 0.304 & 0.316 & 0.304 & 0.980 & 0.298 & 0.304 \\ 
   & 0.999 & 0.995 & 0.848 & 0.482 & 0.402 & 0.826 & 0.832 & 0.410 & 0.326 & 0.812 \\ 
   & 0.999 & 0.999 & 0.848 & 0.954 & 0.818 & 0.844 & 0.832 & 0.982 & 0.818 & 0.832 \\ 
   \hline
\end{tabular}

	\end{center}
\end{table}

\subsubsection*{Setting~\ref*{enum:sim_exp_exp}}
\begin{table}[H]
	\begin{center}
		\caption{Rejection rate for Setting~\ref*{enum:sim_exp_exp} with $p(0){=}0.7$, $p(1){=}0.7$, $n{=}500$, $\rho{\in}\{0.3,0.5\}$ (M1/M2: Method 1/ Method 2.)\label{supp_tab:sim_res_exp_exp_exp_p0_07_p1_07_n500_rho_comb}}
		\scriptsize
		\begin{tabular}{ccc@{\extracolsep{6pt}}cc@{\extracolsep{6pt}}cc@{\extracolsep{6pt}}cc@{\extracolsep{6pt}}cc}
& & & \multicolumn{4}{c}{$\rho=0.3$} & \multicolumn{4}{c}{$\rho=0.5$} \\ \cline{4-7}\cline{8-11}
& \multicolumn{2}{c}{$\uptau_G(x)$} & \multicolumn{2}{c}{Rej. rate of $H_{0x}$} & \multicolumn{2}{c}{Rej. rate of $H_0$} & \multicolumn{2}{c}{Rej. rate of $H_{0x}$} & \multicolumn{2}{c}{Rej. rate of $H_0$} \\ \cline{2-3}\cline{4-5}\cline{6-7}\cline{8-9}\cline{10-11}
Scenario & $x=0$ & $x=1$ & $x=0$ & $x=1$ & M1 & M2 & $x=0$ & $x=1$ & M1 & M2 \\ \hline
\multirow{4}{*}{A} & 0.950 & 0.950 & 0.038 & 0.014 & 0.000 & 0.026 & 0.018 & 0.018 & 0.000 & 0.018 \\ 
   & 0.950 & 0.975 & 0.038 & 0.062 & 0.002 & 0.030 & 0.018 & 0.044 & 0.002 & 0.024 \\ 
   & 0.975 & 0.975 & 0.048 & 0.058 & 0.002 & 0.044 & 0.036 & 0.038 & 0.000 & 0.036 \\ 
   & 0.999 & 0.950 & 0.876 & 0.016 & 0.014 & 0.018 & 0.802 & 0.016 & 0.014 & 0.016 \\ 
   \hline
\multirow{7}{*}{B} & 0.975 & 0.990 & 0.048 & 0.346 & 0.020 & 0.044 & 0.036 & 0.282 & 0.012 & 0.032 \\ 
   & 0.975 & 0.999 & 0.048 & 0.726 & 0.038 & 0.062 & 0.036 & 0.858 & 0.030 & 0.038 \\ 
   & 0.990 & 0.975 & 0.074 & 0.058 & 0.004 & 0.048 & 0.056 & 0.042 & 0.004 & 0.054 \\ 
   & 0.990 & 0.990 & 0.074 & 0.342 & 0.028 & 0.100 & 0.056 & 0.278 & 0.016 & 0.086 \\ 
   & 0.990 & 0.999 & 0.074 & 0.720 & 0.056 & 0.092 & 0.056 & 0.860 & 0.048 & 0.060 \\ 
   & 0.995 & 0.975 & 0.150 & 0.056 & 0.006 & 0.070 & 0.112 & 0.038 & 0.000 & 0.050 \\ 
   & 0.999 & 0.990 & 0.876 & 0.346 & 0.324 & 0.352 & 0.802 & 0.264 & 0.206 & 0.266 \\ 
   \hline
\multirow{4}{*}{C} & 0.995 & 0.995 & 0.150 & 0.548 & 0.068 & 0.240 & 0.112 & 0.562 & 0.064 & 0.174 \\ 
   & 0.995 & 0.999 & 0.150 & 0.714 & 0.096 & 0.222 & 0.112 & 0.860 & 0.096 & 0.132 \\ 
   & 0.999 & 0.995 & 0.876 & 0.568 & 0.518 & 0.554 & 0.802 & 0.578 & 0.452 & 0.562 \\ 
   & 0.999 & 0.999 & 0.876 & 0.708 & 0.630 & 0.688 & 0.802 & 0.858 & 0.680 & 0.784 \\ 
   \hline
\end{tabular}

	\end{center}
\end{table}

\begin{table}[H]
	\begin{center}
		\caption{Rejection rate for Setting~\ref*{enum:sim_exp_exp} with $p(0){=}0.7$, $p(1){=}0.7$, $n{=}1000$, $\rho{\in}\{0.3,0.5\}$ (M1/M2: Method 1/ Method 2.)\label{supp_tab:sim_res_exp_exp_exp_p0_07_p1_07_n1000_rho_comb}}
		\scriptsize
		\begin{tabular}{ccc@{\extracolsep{6pt}}cc@{\extracolsep{6pt}}cc@{\extracolsep{6pt}}cc@{\extracolsep{6pt}}cc}
& & & \multicolumn{4}{c}{$\rho=0.3$} & \multicolumn{4}{c}{$\rho=0.5$} \\ \cline{4-7}\cline{8-11}
& \multicolumn{2}{c}{$\uptau_G(x)$} & \multicolumn{2}{c}{Rej. rate of $H_{0x}$} & \multicolumn{2}{c}{Rej. rate of $H_0$} & \multicolumn{2}{c}{Rej. rate of $H_{0x}$} & \multicolumn{2}{c}{Rej. rate of $H_0$} \\ \cline{2-3}\cline{4-5}\cline{6-7}\cline{8-9}\cline{10-11}
Scenario & $x=0$ & $x=1$ & $x=0$ & $x=1$ & M1 & M2 & $x=0$ & $x=1$ & M1 & M2 \\ \hline
\multirow{4}{*}{A} & 0.950 & 0.950 & 0.012 & 0.020 & 0.000 & 0.018 & 0.030 & 0.038 & 0.000 & 0.020 \\ 
   & 0.950 & 0.975 & 0.012 & 0.036 & 0.000 & 0.010 & 0.030 & 0.034 & 0.004 & 0.026 \\ 
   & 0.975 & 0.975 & 0.044 & 0.030 & 0.000 & 0.032 & 0.034 & 0.032 & 0.000 & 0.034 \\ 
   & 0.999 & 0.950 & 0.936 & 0.024 & 0.022 & 0.028 & 0.894 & 0.036 & 0.030 & 0.036 \\ 
   \hline
\multirow{7}{*}{B} & 0.975 & 0.990 & 0.044 & 0.218 & 0.010 & 0.044 & 0.034 & 0.142 & 0.006 & 0.032 \\ 
   & 0.975 & 0.999 & 0.044 & 0.896 & 0.040 & 0.046 & 0.034 & 0.936 & 0.032 & 0.036 \\ 
   & 0.990 & 0.975 & 0.092 & 0.034 & 0.008 & 0.058 & 0.064 & 0.034 & 0.002 & 0.046 \\ 
   & 0.990 & 0.990 & 0.092 & 0.222 & 0.018 & 0.086 & 0.064 & 0.136 & 0.012 & 0.066 \\ 
   & 0.990 & 0.999 & 0.092 & 0.896 & 0.084 & 0.092 & 0.064 & 0.930 & 0.060 & 0.064 \\ 
   & 0.995 & 0.975 & 0.184 & 0.034 & 0.012 & 0.080 & 0.172 & 0.032 & 0.008 & 0.058 \\ 
   & 0.999 & 0.990 & 0.936 & 0.222 & 0.208 & 0.244 & 0.894 & 0.150 & 0.140 & 0.156 \\ 
   \hline
\multirow{4}{*}{C} & 0.995 & 0.995 & 0.184 & 0.536 & 0.090 & 0.196 & 0.172 & 0.452 & 0.054 & 0.158 \\ 
   & 0.995 & 0.999 & 0.184 & 0.900 & 0.172 & 0.190 & 0.172 & 0.936 & 0.162 & 0.174 \\ 
   & 0.999 & 0.995 & 0.936 & 0.526 & 0.484 & 0.508 & 0.894 & 0.450 & 0.382 & 0.436 \\ 
   & 0.999 & 0.999 & 0.936 & 0.892 & 0.834 & 0.852 & 0.894 & 0.928 & 0.832 & 0.890 \\ 
   \hline
\end{tabular}

	\end{center}
\end{table}

\begin{table}[H]
	\begin{center}
		\caption{Rejection rate for Setting~\ref*{enum:sim_exp_exp} with $\rho{=}0.5$, $n{=}500$, $p(x){\in}\{0.3,0.7\}$ (M1/M2: Method 1/ Method 2.)\label{supp_tab:sim_res_exp_exp_exp_rho_05_n500_p_comb}}
		\scriptsize
				\setlength{\tabcolsep}{3.5pt}
		\begin{tabular}{ccc@{\extracolsep{6pt}}cc@{\extracolsep{6pt}}cc@{\extracolsep{6pt}}cc@{\extracolsep{6pt}}cccc@{\extracolsep{6pt}}cc@{\extracolsep{6pt}}cc}
& & & \multicolumn{4}{c}{$p(0)=p(1)=0.7$} & \multicolumn{4}{c}{$p(0)=0.3$, $p(1)=0.7$} & \multicolumn{4}{c}{$p(0)=0.3$, $p(1)=0.3$} \\ \cline{4-7}\cline{8-11}\cline{12-15}
& \multicolumn{2}{c}{$\uptau_G(x)$} & \multicolumn{2}{c}{\makecell{Rej. rate of\\$H_{0x}$}} & \multicolumn{2}{c}{\makecell{Rej. rate of\\$H_0$}} & \multicolumn{2}{c}{\makecell{Rej. rate of\\$H_{0x}$}} & \multicolumn{2}{c}{\makecell{Rej. rate of\\$H_0$}} & \multicolumn{2}{c}{\makecell{Rej. rate of\\$H_{0x}$}} & \multicolumn{2}{c}{\makecell{Rej. rate of\\$H_0$}} \\ \cline{2-3}\cline{4-5}\cline{6-7}\cline{8-9}\cline{10-11}\cline{12-13}\cline{14-15}
Scenario & $x=0$ & $x=1$ & $x=0$ & $x=1$ & M1 & M2 & $x=0$ & $x=1$ & M1 & M2 & $x=0$ & $x=1$ & M1 & M2 \\ \hline
\multirow{4}{*}{A} & 0.950 & 0.950 & 0.018 & 0.018 & 0.000 & 0.018 & 0.036 & 0.018 & 0.000 & 0.034 & 0.036 & 0.034 & 0.004 & 0.038 \\ 
   & 0.950 & 0.975 & 0.018 & 0.044 & 0.002 & 0.024 & 0.036 & 0.040 & 0.002 & 0.042 & 0.036 & 0.184 & 0.006 & 0.042 \\ 
   & 0.975 & 0.975 & 0.036 & 0.038 & 0.000 & 0.036 & 0.024 & 0.038 & 0.000 & 0.040 & 0.024 & 0.182 & 0.002 & 0.066 \\ 
   & 0.999 & 0.950 & 0.802 & 0.016 & 0.014 & 0.016 & 0.854 & 0.018 & 0.016 & 0.108 & 0.854 & 0.038 & 0.032 & 0.038 \\ 
   \hline
\multirow{7}{*}{B} & 0.975 & 0.990 & 0.036 & 0.282 & 0.012 & 0.032 & 0.024 & 0.268 & 0.008 & 0.022 & 0.024 & 0.590 & 0.014 & 0.064 \\ 
   & 0.975 & 0.999 & 0.036 & 0.858 & 0.030 & 0.038 & 0.024 & 0.858 & 0.018 & 0.026 & 0.024 & 0.990 & 0.024 & 0.024 \\ 
   & 0.990 & 0.975 & 0.056 & 0.042 & 0.004 & 0.054 & 0.096 & 0.040 & 0.008 & 0.070 & 0.096 & 0.186 & 0.026 & 0.134 \\ 
   & 0.990 & 0.990 & 0.056 & 0.278 & 0.016 & 0.086 & 0.096 & 0.286 & 0.022 & 0.090 & 0.096 & 0.592 & 0.060 & 0.250 \\ 
   & 0.990 & 0.999 & 0.056 & 0.860 & 0.048 & 0.060 & 0.096 & 0.864 & 0.084 & 0.094 & 0.096 & 0.986 & 0.094 & 0.096 \\ 
   & 0.995 & 0.975 & 0.112 & 0.038 & 0.000 & 0.050 & 0.324 & 0.036 & 0.010 & 0.162 & 0.324 & 0.184 & 0.054 & 0.182 \\ 
   & 0.999 & 0.990 & 0.802 & 0.264 & 0.206 & 0.266 & 0.854 & 0.282 & 0.232 & 0.728 & 0.854 & 0.598 & 0.502 & 0.588 \\ 
   \hline
\multirow{4}{*}{C} & 0.995 & 0.995 & 0.112 & 0.562 & 0.064 & 0.174 & 0.324 & 0.572 & 0.166 & 0.302 & 0.324 & 0.792 & 0.250 & 0.508 \\ 
   & 0.995 & 0.999 & 0.112 & 0.860 & 0.096 & 0.132 & 0.324 & 0.860 & 0.276 & 0.314 & 0.324 & 0.990 & 0.322 & 0.326 \\ 
   & 0.999 & 0.995 & 0.802 & 0.578 & 0.452 & 0.562 & 0.854 & 0.566 & 0.466 & 0.782 & 0.854 & 0.810 & 0.684 & 0.800 \\ 
   & 0.999 & 0.999 & 0.802 & 0.858 & 0.680 & 0.784 & 0.854 & 0.864 & 0.728 & 0.808 & 0.854 & 0.986 & 0.842 & 0.884 \\ 
   \hline
\end{tabular}

	\end{center}
\end{table}

\begin{table}[H]
	\begin{center}
		\caption{Rejection rate for Setting~\ref*{enum:sim_exp_exp} with $\rho{=}0.5$, $n{=}1000$, $p(x){\in}\{0.3,0.7\}$ (M1/M2: Method 1/ Method 2.)\label{supp_tab:sim_res_exp_exp_exp_rho_05_n1000_p_comb}}
		\scriptsize
		\setlength{\tabcolsep}{3.5pt}
		\begin{tabular}{ccc@{\extracolsep{6pt}}cc@{\extracolsep{6pt}}cc@{\extracolsep{6pt}}cc@{\extracolsep{6pt}}cccc@{\extracolsep{6pt}}cc@{\extracolsep{6pt}}cc}
& & & \multicolumn{4}{c}{$p(0)=p(1)=0.7$} & \multicolumn{4}{c}{$p(0)=0.3$, $p(1)=0.7$} & \multicolumn{4}{c}{$p(0)=0.3$, $p(1)=0.3$} \\ \cline{4-7}\cline{8-11}\cline{12-15}
& \multicolumn{2}{c}{$\uptau_G(x)$} & \multicolumn{2}{c}{\makecell{Rej. rate of\\$H_{0x}$}} & \multicolumn{2}{c}{\makecell{Rej. rate of\\$H_0$}} & \multicolumn{2}{c}{\makecell{Rej. rate of\\$H_{0x}$}} & \multicolumn{2}{c}{\makecell{Rej. rate of\\$H_0$}} & \multicolumn{2}{c}{\makecell{Rej. rate of\\$H_{0x}$}} & \multicolumn{2}{c}{\makecell{Rej. rate of\\$H_0$}} \\ \cline{2-3}\cline{4-5}\cline{6-7}\cline{8-9}\cline{10-11}\cline{12-13}\cline{14-15}
Scenario & $x=0$ & $x=1$ & $x=0$ & $x=1$ & M1 & M2 & $x=0$ & $x=1$ & M1 & M2 & $x=0$ & $x=1$ & M1 & M2 \\ \hline
\multirow{4}{*}{A} & 0.950 & 0.950 & 0.030 & 0.038 & 0.000 & 0.020 & 0.032 & 0.038 & 0.000 & 0.028 & 0.032 & 0.028 & 0.000 & 0.032 \\ 
   & 0.950 & 0.975 & 0.030 & 0.034 & 0.004 & 0.026 & 0.032 & 0.038 & 0.004 & 0.026 & 0.032 & 0.086 & 0.004 & 0.032 \\ 
   & 0.975 & 0.975 & 0.034 & 0.032 & 0.000 & 0.034 & 0.010 & 0.032 & 0.000 & 0.030 & 0.010 & 0.096 & 0.000 & 0.020 \\ 
   & 0.999 & 0.950 & 0.894 & 0.036 & 0.030 & 0.036 & 0.874 & 0.040 & 0.032 & 0.048 & 0.874 & 0.022 & 0.018 & 0.022 \\ 
   \hline
\multirow{7}{*}{B} & 0.975 & 0.990 & 0.034 & 0.142 & 0.006 & 0.032 & 0.010 & 0.138 & 0.000 & 0.012 & 0.010 & 0.474 & 0.008 & 0.014 \\ 
   & 0.975 & 0.999 & 0.034 & 0.936 & 0.032 & 0.036 & 0.010 & 0.932 & 0.010 & 0.014 & 0.010 & 0.976 & 0.010 & 0.010 \\ 
   & 0.990 & 0.975 & 0.064 & 0.034 & 0.002 & 0.046 & 0.092 & 0.034 & 0.004 & 0.052 & 0.092 & 0.088 & 0.010 & 0.090 \\ 
   & 0.990 & 0.990 & 0.064 & 0.136 & 0.012 & 0.066 & 0.092 & 0.144 & 0.010 & 0.076 & 0.092 & 0.460 & 0.036 & 0.166 \\ 
   & 0.990 & 0.999 & 0.064 & 0.930 & 0.060 & 0.064 & 0.092 & 0.940 & 0.086 & 0.094 & 0.092 & 0.970 & 0.090 & 0.092 \\ 
   & 0.995 & 0.975 & 0.172 & 0.032 & 0.008 & 0.058 & 0.300 & 0.034 & 0.006 & 0.104 & 0.300 & 0.090 & 0.034 & 0.110 \\ 
   & 0.999 & 0.990 & 0.894 & 0.150 & 0.140 & 0.156 & 0.874 & 0.156 & 0.140 & 0.630 & 0.874 & 0.472 & 0.408 & 0.462 \\ 
   \hline
\multirow{4}{*}{C} & 0.995 & 0.995 & 0.172 & 0.452 & 0.054 & 0.158 & 0.300 & 0.436 & 0.132 & 0.270 & 0.300 & 0.750 & 0.226 & 0.388 \\ 
   & 0.995 & 0.999 & 0.172 & 0.936 & 0.162 & 0.174 & 0.300 & 0.922 & 0.270 & 0.298 & 0.300 & 0.970 & 0.288 & 0.300 \\ 
   & 0.999 & 0.995 & 0.894 & 0.450 & 0.382 & 0.436 & 0.874 & 0.458 & 0.394 & 0.798 & 0.874 & 0.748 & 0.654 & 0.734 \\ 
   & 0.999 & 0.999 & 0.894 & 0.928 & 0.832 & 0.890 & 0.874 & 0.938 & 0.822 & 0.864 & 0.874 & 0.972 & 0.848 & 0.878 \\ 
   \hline
\end{tabular}

	\end{center}
\end{table}

\subsubsection*{Setting~\ref*{enum:sim_exp_unif}}
\begin{table}[H]
	\begin{center}
		\caption{Rejection rate for Setting~\ref*{enum:sim_exp_unif} with $\rho=0.5$, $p(0)=0.6$, $p(1)=0.6$, $\Delta G(\uptau_G(x))=0$ and $n\in\lbrace500,1000\rbrace$.\label{supp_tab:sim_res_exp_exp_unif_rho_05_p0_06_p1_06_j000}}
		\scriptsize
		\begin{tabular}{ccc@{\extracolsep{6pt}}cc@{\extracolsep{6pt}}cc@{\extracolsep{6pt}}cc@{\extracolsep{6pt}}cc}
& & & \multicolumn{4}{c}{$n=500$} & \multicolumn{4}{c}{$n=1000$} \\ \cline{4-7}\cline{8-11}
& \multicolumn{2}{c}{$\uptau_G(x)$} & \multicolumn{2}{c}{\makecell{Rej. rate of\\$H_{0x}$}} & \multicolumn{2}{c}{\makecell{Rej. rate of\\$H_0$}} & \multicolumn{2}{c}{\makecell{Rej. rate of\\$H_{0x}$}} & \multicolumn{2}{c}{\makecell{Rej. rate of\\$H_0$}} \\ \cline{2-3}\cline{4-5}\cline{6-7}\cline{8-9}\cline{10-11}
Scenario & $x=0$ & $x=1$ & $x=0$ & $x=1$ & M1 & M2 & $x=0$ & $x=1$ & M1 & M2 \\ \hline
\multirow{5}{*}{A} & 0.950 & 0.950 & 0.002 & 0.008 & 0.000 & 0.006 & 0.014 & 0.016 & 0.000 & 0.024 \\ 
   & 0.950 & 0.975 & 0.002 & 0.016 & 0.000 & 0.004 & 0.014 & 0.008 & 0.000 & 0.018 \\ 
   & 0.975 & 0.975 & 0.020 & 0.016 & 0.000 & 0.034 & 0.010 & 0.012 & 0.000 & 0.016 \\ 
   & 0.990 & 0.950 & 0.104 & 0.006 & 0.000 & 0.006 & 0.062 & 0.016 & 0.002 & 0.016 \\ 
   & 0.999 & 0.950 & 0.892 & 0.006 & 0.006 & 0.010 & 0.944 & 0.010 & 0.010 & 0.010 \\ 
   \hline
\multirow{7}{*}{B} & 0.975 & 0.990 & 0.020 & 0.132 & 0.004 & 0.030 & 0.010 & 0.074 & 0.002 & 0.010 \\ 
   & 0.975 & 0.999 & 0.020 & 0.880 & 0.018 & 0.020 & 0.010 & 0.934 & 0.008 & 0.010 \\ 
   & 0.990 & 0.975 & 0.104 & 0.018 & 0.000 & 0.024 & 0.062 & 0.006 & 0.000 & 0.008 \\ 
   & 0.990 & 0.990 & 0.104 & 0.122 & 0.018 & 0.154 & 0.062 & 0.064 & 0.000 & 0.074 \\ 
   & 0.990 & 0.999 & 0.104 & 0.868 & 0.100 & 0.104 & 0.062 & 0.926 & 0.058 & 0.060 \\ 
   & 0.999 & 0.975 & 0.892 & 0.014 & 0.010 & 0.018 & 0.944 & 0.008 & 0.008 & 0.008 \\ 
   & 0.999 & 0.990 & 0.892 & 0.118 & 0.092 & 0.116 & 0.944 & 0.070 & 0.068 & 0.070 \\ 
   \hline
\multirow{4}{*}{C} & 0.995 & 0.995 & 0.296 & 0.282 & 0.074 & 0.318 & 0.252 & 0.246 & 0.058 & 0.272 \\ 
   & 0.995 & 0.999 & 0.296 & 0.890 & 0.264 & 0.294 & 0.252 & 0.932 & 0.238 & 0.252 \\ 
   & 0.999 & 0.995 & 0.892 & 0.288 & 0.260 & 0.288 & 0.944 & 0.234 & 0.220 & 0.232 \\ 
   & 0.999 & 0.999 & 0.892 & 0.888 & 0.788 & 0.886 & 0.944 & 0.928 & 0.876 & 0.930 \\ 
   \hline
\end{tabular}

	\end{center}
\end{table}
\begin{table}[H]
	\begin{center}
		\caption{Rejection rate for Setting~\ref*{enum:sim_exp_unif} with $\rho=0.5$, $p(0)=0.6$, $p(1)=0.6$, $\Delta G(\uptau_G(x))=0.01$ and $n\in\lbrace500,1000\rbrace$.\label{supp_tab:sim_res_exp_exp_unif_rho_05_p0_06_p1_06_j001}}
		\scriptsize
		\begin{tabular}{ccc@{\extracolsep{6pt}}cc@{\extracolsep{6pt}}cc@{\extracolsep{6pt}}cc@{\extracolsep{6pt}}cc}
& & & \multicolumn{4}{c}{$n=500$} & \multicolumn{4}{c}{$n=1000$} \\ \cline{4-7}\cline{8-11}
& \multicolumn{2}{c}{$\uptau_G(x)$} & \multicolumn{2}{c}{\makecell{Rej. rate of\\$H_{0x}$}} & \multicolumn{2}{c}{\makecell{Rej. rate of\\$H_0$}} & \multicolumn{2}{c}{\makecell{Rej. rate of\\$H_{0x}$}} & \multicolumn{2}{c}{\makecell{Rej. rate of\\$H_0$}} \\ \cline{2-3}\cline{4-5}\cline{6-7}\cline{8-9}\cline{10-11}
Scenario & $x=0$ & $x=1$ & $x=0$ & $x=1$ & M1 & M2 & $x=0$ & $x=1$ & M1 & M2 \\ \hline
\multirow{5}{*}{A} & 0.950 & 0.950 & 0.006 & 0.006 & 0.000 & 0.008 & 0.020 & 0.024 & 0.002 & 0.032 \\ 
   & 0.950 & 0.975 & 0.006 & 0.024 & 0.000 & 0.008 & 0.020 & 0.006 & 0.000 & 0.018 \\ 
   & 0.975 & 0.975 & 0.012 & 0.018 & 0.000 & 0.026 & 0.012 & 0.010 & 0.000 & 0.010 \\ 
   & 0.990 & 0.950 & 0.114 & 0.006 & 0.000 & 0.006 & 0.066 & 0.026 & 0.004 & 0.026 \\ 
   & 0.999 & 0.950 & 0.898 & 0.000 & 0.000 & 0.000 & 0.952 & 0.014 & 0.014 & 0.014 \\ 
   \hline
\multirow{7}{*}{B} & 0.975 & 0.990 & 0.012 & 0.126 & 0.002 & 0.018 & 0.012 & 0.062 & 0.002 & 0.012 \\ 
   & 0.975 & 0.999 & 0.012 & 0.884 & 0.012 & 0.012 & 0.012 & 0.936 & 0.012 & 0.012 \\ 
   & 0.990 & 0.975 & 0.114 & 0.018 & 0.000 & 0.022 & 0.066 & 0.010 & 0.000 & 0.010 \\ 
   & 0.990 & 0.990 & 0.114 & 0.134 & 0.014 & 0.172 & 0.066 & 0.060 & 0.006 & 0.070 \\ 
   & 0.990 & 0.999 & 0.114 & 0.878 & 0.102 & 0.114 & 0.066 & 0.948 & 0.064 & 0.064 \\ 
   & 0.999 & 0.975 & 0.898 & 0.020 & 0.016 & 0.020 & 0.952 & 0.008 & 0.008 & 0.008 \\ 
   & 0.999 & 0.990 & 0.898 & 0.132 & 0.110 & 0.130 & 0.952 & 0.056 & 0.056 & 0.056 \\ 
   \hline
\multirow{4}{*}{C} & 0.995 & 0.995 & 0.350 & 0.300 & 0.096 & 0.364 & 0.278 & 0.260 & 0.062 & 0.282 \\ 
   & 0.995 & 0.999 & 0.350 & 0.882 & 0.314 & 0.350 & 0.278 & 0.934 & 0.266 & 0.278 \\ 
   & 0.999 & 0.995 & 0.898 & 0.290 & 0.254 & 0.286 & 0.952 & 0.260 & 0.248 & 0.258 \\ 
   & 0.999 & 0.999 & 0.898 & 0.884 & 0.792 & 0.886 & 0.952 & 0.932 & 0.886 & 0.940 \\ 
   \hline
\end{tabular}

	\end{center}
\end{table}

\begin{table}[H]
	\begin{center}
		\caption{Rejection rate for Setting~\ref*{enum:sim_exp_unif} with $p(0)=0.6$, $p(1)=0.6$, $\Delta G(\uptau_G(x))=0.01$, $n=1000$ and $\rho\in\lbrace0.3,0.5\rbrace$.\label{supp_tab:sim_res_exp_exp_unif_p0_06_p1_06_j_001_rho_comp}}
		\scriptsize
		\begin{tabular}{ccc@{\extracolsep{6pt}}cc@{\extracolsep{6pt}}cc@{\extracolsep{6pt}}cc@{\extracolsep{6pt}}cc}
& & & \multicolumn{4}{c}{$\rho=0.3$} & \multicolumn{4}{c}{$\rho=0.5$} \\ \cline{4-7}\cline{8-11}
& \multicolumn{2}{c}{$\uptau_G(x)$} & \multicolumn{2}{c}{\makecell{Rej. rate of\\$H_{0x}$}} & \multicolumn{2}{c}{\makecell{Rej. rate of\\$H_0$}} & \multicolumn{2}{c}{\makecell{Rej. rate of\\$H_{0x}$}} & \multicolumn{2}{c}{\makecell{Rej. rate of\\$H_0$}} \\ \cline{2-3}\cline{4-5}\cline{6-7}\cline{8-9}\cline{10-11}
Scenario & $x=0$ & $x=1$ & $x=0$ & $x=1$ & M1 & M2 & $x=0$ & $x=1$ & M1 & M2 \\ \hline
\multirow{5}{*}{A} & 0.950 & 0.950 & 0.014 & 0.014 & 0.000 & 0.020 & 0.020 & 0.024 & 0.002 & 0.032 \\ 
   & 0.950 & 0.975 & 0.014 & 0.012 & 0.000 & 0.018 & 0.020 & 0.006 & 0.000 & 0.018 \\ 
   & 0.975 & 0.975 & 0.020 & 0.012 & 0.000 & 0.024 & 0.012 & 0.010 & 0.000 & 0.010 \\ 
   & 0.990 & 0.950 & 0.076 & 0.010 & 0.000 & 0.014 & 0.066 & 0.026 & 0.004 & 0.026 \\ 
   & 0.999 & 0.950 & 0.968 & 0.014 & 0.014 & 0.014 & 0.952 & 0.014 & 0.014 & 0.014 \\ 
   \hline
\multirow{7}{*}{B} & 0.975 & 0.990 & 0.020 & 0.088 & 0.002 & 0.022 & 0.012 & 0.062 & 0.002 & 0.012 \\ 
   & 0.975 & 0.999 & 0.020 & 0.908 & 0.018 & 0.020 & 0.012 & 0.936 & 0.012 & 0.012 \\ 
   & 0.990 & 0.975 & 0.076 & 0.010 & 0.000 & 0.012 & 0.066 & 0.010 & 0.000 & 0.010 \\ 
   & 0.990 & 0.990 & 0.076 & 0.080 & 0.006 & 0.090 & 0.066 & 0.060 & 0.006 & 0.070 \\ 
   & 0.990 & 0.999 & 0.076 & 0.898 & 0.062 & 0.076 & 0.066 & 0.948 & 0.064 & 0.064 \\ 
   & 0.999 & 0.975 & 0.968 & 0.012 & 0.010 & 0.012 & 0.952 & 0.008 & 0.008 & 0.008 \\ 
   & 0.999 & 0.990 & 0.968 & 0.080 & 0.074 & 0.080 & 0.952 & 0.056 & 0.056 & 0.056 \\ 
   \hline
\multirow{4}{*}{C} & 0.995 & 0.995 & 0.282 & 0.290 & 0.078 & 0.312 & 0.278 & 0.260 & 0.062 & 0.282 \\ 
   & 0.995 & 0.999 & 0.282 & 0.896 & 0.258 & 0.282 & 0.278 & 0.934 & 0.266 & 0.278 \\ 
   & 0.999 & 0.995 & 0.968 & 0.284 & 0.272 & 0.280 & 0.952 & 0.260 & 0.248 & 0.258 \\ 
   & 0.999 & 0.999 & 0.968 & 0.894 & 0.862 & 0.914 & 0.952 & 0.932 & 0.886 & 0.940 \\ 
   \hline
\end{tabular}

	\end{center}
\end{table}

\begin{table}[H]
	\begin{center}
		\caption{Rejection rate for Setting~\ref*{enum:sim_exp_unif} with $\rho=0.5$, $\Delta G(\uptau_G(x))=0.01$, $n=1000$ and $p(x)\in\lbrace0.4,0.6\rbrace$ (M1/M2: Method 1/ Method 2.)\label{supp_tab:sim_res_exp_exp_unif_rho_05_j_001_p_comp}}
		\scriptsize
		\setlength{\tabcolsep}{3.5pt}
		\begin{tabular}{ccc@{\extracolsep{6pt}}cc@{\extracolsep{6pt}}cc@{\extracolsep{6pt}}cc@{\extracolsep{6pt}}cc@{\extracolsep{6pt}}cc@{\extracolsep{6pt}}cc}
& & & \multicolumn{4}{c}{$p(0)=p(1)=0.6$} & \multicolumn{4}{c}{$p(0)=0.6$, $p(1)=0.4$} & \multicolumn{4}{c}{$p(0)=p(1)=0.4$} \\ \cline{4-7}\cline{8-11}\cline{12-15}
& \multicolumn{2}{c}{$\uptau_G(x)$} & \multicolumn{2}{c}{\makecell{Rej. rate of\\$H_{0x}$}} & \multicolumn{2}{c}{\makecell{Rej. rate of\\$H_0$}} & \multicolumn{2}{c}{\makecell{Rej. rate of\\$H_{0x}$}} & \multicolumn{2}{c}{\makecell{Rej. rate of\\$H_0$}}& \multicolumn{2}{c}{\makecell{Rej. rate of\\$H_{0x}$}} & \multicolumn{2}{c}{\makecell{Rej. rate of\\$H_0$}} \\ \cline{2-3}\cline{4-5}\cline{6-7}\cline{8-9}\cline{10-11}\cline{12-13}\cline{14-15}
Scenario & $x=0$ & $x=1$ & $x=0$ & $x=1$ & M1 & M2 & $x=0$ & $x=1$ & M1 & M2 & $x=0$ & $x=1$ & M1 & M2 \\ \hline
\multirow{5}{*}{A} & 0.950 & 0.950 & 0.020 & 0.024 & 0.002 & 0.032 & 0.020 & 0.022 & 0.002 & 0.028 & 0.004 & 0.016 & 0.000 & 0.012 \\ 
   & 0.950 & 0.975 & 0.020 & 0.006 & 0.000 & 0.018 & 0.020 & 0.024 & 0.002 & 0.022 & 0.004 & 0.026 & 0.000 & 0.012 \\ 
   & 0.975 & 0.975 & 0.012 & 0.010 & 0.000 & 0.010 & 0.012 & 0.018 & 0.000 & 0.022 & 0.020 & 0.018 & 0.000 & 0.024 \\ 
   & 0.990 & 0.950 & 0.066 & 0.026 & 0.004 & 0.026 & 0.066 & 0.016 & 0.002 & 0.016 & 0.182 & 0.014 & 0.002 & 0.014 \\ 
   & 0.999 & 0.950 & 0.952 & 0.014 & 0.014 & 0.014 & 0.952 & 0.012 & 0.012 & 0.012 & 0.942 & 0.014 & 0.014 & 0.014 \\ 
   \hline
\multirow{7}{*}{B} & 0.975 & 0.990 & 0.012 & 0.062 & 0.002 & 0.012 & 0.012 & 0.168 & 0.002 & 0.140 & 0.020 & 0.174 & 0.004 & 0.024 \\ 
   & 0.975 & 0.999 & 0.012 & 0.936 & 0.012 & 0.012 & 0.012 & 0.928 & 0.012 & 0.012 & 0.020 & 0.928 & 0.020 & 0.020 \\ 
   & 0.990 & 0.975 & 0.066 & 0.010 & 0.000 & 0.010 & 0.066 & 0.020 & 0.002 & 0.018 & 0.182 & 0.018 & 0.002 & 0.030 \\ 
   & 0.990 & 0.990 & 0.066 & 0.060 & 0.006 & 0.070 & 0.066 & 0.180 & 0.004 & 0.172 & 0.182 & 0.172 & 0.036 & 0.204 \\ 
   & 0.990 & 0.999 & 0.066 & 0.948 & 0.064 & 0.064 & 0.066 & 0.920 & 0.062 & 0.066 & 0.182 & 0.920 & 0.170 & 0.182 \\ 
   & 0.999 & 0.975 & 0.952 & 0.008 & 0.008 & 0.008 & 0.952 & 0.030 & 0.028 & 0.030 & 0.942 & 0.012 & 0.010 & 0.012 \\ 
   & 0.999 & 0.990 & 0.952 & 0.056 & 0.056 & 0.056 & 0.952 & 0.184 & 0.170 & 0.184 & 0.942 & 0.168 & 0.158 & 0.168 \\ 
   \hline
\multirow{4}{*}{C} & 0.995 & 0.995 & 0.278 & 0.260 & 0.062 & 0.282 & 0.278 & 0.436 & 0.114 & 0.408 & 0.450 & 0.438 & 0.200 & 0.468 \\ 
   & 0.995 & 0.999 & 0.278 & 0.934 & 0.266 & 0.278 & 0.278 & 0.924 & 0.256 & 0.372 & 0.450 & 0.926 & 0.416 & 0.450 \\ 
   & 0.999 & 0.995 & 0.952 & 0.260 & 0.248 & 0.258 & 0.952 & 0.454 & 0.430 & 0.452 & 0.942 & 0.426 & 0.402 & 0.424 \\ 
   & 0.999 & 0.999 & 0.952 & 0.932 & 0.886 & 0.940 & 0.952 & 0.922 & 0.876 & 0.916 & 0.942 & 0.924 & 0.866 & 0.954 \\ 
   \hline
\end{tabular}

	\end{center}
\end{table}

\end{document}